%% file: main.tex
\documentclass[acmsmall]{acmart}

\makeatletter

\newcommand\addauthornote[1]{%
  \if@ACM@anonymous\else
    \g@addto@macro\addresses{\g@addto@macro\@currentauthors{{\small *}}}
  \fi}

\newcommand\@addauthornotemark[1]{\let\@tmpcnta\0
    \g@addto@macro\@currentauthors{\footnotemark\relax\let\0\@tmpcnta}}
    
\makeatother

\acmJournal{PACMPL}
\acmVolume{1}
\acmNumber{OOPSLA} % CONF = POPL or ICFP or OOPSLA
\acmArticle{1}
\acmYear{2022}
\acmMonth{11}
\acmDOI{} % TODO: \acmDOI{10.1145/nnnnnnn.nnnnnnn}
\startPage{1}

\setcopyright{none}

\usepackage{booktabs}   %% For formal tables:
\usepackage{subcaption} %% For complex figures with subfigures/subcaptions
\usepackage{amsmath}
\usepackage{array, multirow}
\usepackage[plain,Fig]{algorithm}
\usepackage{algorithmicx}
\usepackage[noend]{algpseudocode}
\usepackage{comment}
\usepackage{enumitem}
\usepackage{graphicx}
\usepackage[final]{listings}
\usepackage{pifont}
\usepackage{relsize}
\usepackage{stmaryrd}
\usepackage{syntax}
\usepackage{textcomp}
\usepackage{xspace}
\usepackage[frozencache,cachedir=minted-output]{minted}
\usepackage{subcaption}
\usepackage{dsfont}
\usepackage{placeins}
\usepackage{floatflt}
\usepackage{wrapfig}
\usepackage{enumitem}

\usepackage{mathtools,stmaryrd}

\input{macros}

\newif\iffull
\fulltrue

\newcommand{\irule}[2]%
{\mkern-2mu\displaystyle\frac{#1}{\vphantom{,}#2}\mkern-2mu}

\newcommand{\irulespace}{\ \ \ \  }

\newcommand{\saveparinfos}{%
  \edef\myindent{\the\parindent}%
  \edef\myparskip{\the\parskip}}

\newcommand{\useparinfo}{%
  \setlength{\parindent}{\myindent}%
  \setlength{\parskip}{\myparskip}}

\begin{document}

\saveparinfos

\title{Synthesizing Fine-Grained Synchronization Protocols for Implicit Monitors \iffull{(Extended Version)}\fi}
%%%%%%%%%%%%%%%%%%%%%%%%%%%%%%%%%%%%%%%%%%%%%%%%%%%%%%%%%%%%%%%%%%
% Authors
%%%%%%%%%%%%%%%%%%%%%%%%%%%%%%%%%%%%%%%%%%%%%%%%%%%%%%%%%%%%%%%%%%
\author{Kostas Ferles}
\authornote{Both authors contributed equally to the paper.}
\orcid{0000-0002-8370-5465}
\affiliation {
  \department{Computer Science Department}
  \institution{The University of Texas at Austin}
  \country{USA}
}
\email{kferles@cs.utexas.edu}

\author{Benjamin Sepanski}
\orcid{0000-0002-4924-3009}
\addauthornote{1}
\affiliation {
  \department{Computer Science Department}
  \institution{The University of Texas at Austin}
  \country{USA}
}
\email{ben\_sepanski@utexas.edu}

\author{Rahul Krishnan}
\orcid{0000-0003-0230-5185}
\affiliation {
  \department{Computer Science Department}
  \institution{The University of Texas at Austin}
  \country{USA}
}
\email{rahulk@cs.utexas.edu} 

\author{James Bornholt}
\orcid{0000-0002-3258-3226}
\affiliation {
  \department{Computer Science Department}
  \institution{The University of Texas at Austin}
  \country{USA}
}
\email{bornholt@cs.utexas.edu} 

\author{Isil Dillig}
\orcid{0000-0001-8006-1230}
\affiliation {
  \department{Computer Science Department}
  \institution{The University of Texas at Austin}
  \country{USA}
}
\email{isil@cs.utexas.edu} 

%%%%%%%%%%%%%%%%%%%%%%%%%%%%%%%%%%%%%%%%%%%%%%%%%%%%%%%%%%%%%%%%%%
% Abstract and Title
%%%%%%%%%%%%%%%%%%%%%%%%%%%%%%%%%%%%%%%%%%%%%%%%%%%%%%%%%%%%%%%%%%
\input{abstract}

\maketitle

%%%%%%%%%%%%%%%%%%%%%%%%%%%%%%%%%%%%%%%%%%%%%%%%%%%%%%%%%%%%%%%%%%
% Sections
%%%%%%%%%%%%%%%%%%%%%%%%%%%%%%%%%%%%%%%%%%%%%%%%%%%%%%%%%%%%%%%%%%
\input{intro}

\input{overview}

\input{langs}

\input{algorithm}
\input{implementation}

\input{evaluation}

\input{related}
\input{concl}

%%%%%%%%%%%%%%%%%%%%%%%%%%%%%%%%%%%%%%%%%%%%%%%%%%%%%%%%%%%%%%%%%%
% Acknowledgements
%%%%%%%%%%%%%%%%%%%%%%%%%%%%%%%%%%%%%%%%%%%%%%%%%%%%%%%%%%%%%%%%%%
\begin{acks}
 We would like to thank the anonymous reviewers, Shankara Pailoor, and Benjamin Mariano for their insightful feedback. This material is based upon work supported by the National Science Foundation under Grant Numbers CCF-1918889 and \#CCF-1811865,  the U.S. Department of Energy, Office of Science, Office of Advanced Scientific Computing Research, Department of Energy Computational Science Graduate Fellowship under Award Number SC0021110, and a gift from Relational AI.
\end{acks}

%%%%%%%%%%%%%%%%%%%%%%%%%%%%%%%%%%%%%%%%%%%%%%%%%%%%%%%%%%%%%%%%%%
% Bibliography
%%%%%%%%%%%%%%%%%%%%%%%%%%%%%%%%%%%%%%%%%%%%%%%%%%%%%%%%%%%%%%%%%%
\bibliography{main}

%%%%%%%%%%%%%%%%%%%%%%%%%%%%%%%%%%%%%%%%%%%%%%%%%%%%%%%%%%%%%%%%%%
% Appendix
%%%%%%%%%%%%%%%%%%%%%%%%%%%%%%%%%%%%%%%%%%%%%%%%%%%%%%%%%%%%%%%%%%
\iffull{
\appendix
\include{ablations}

\include{np-complete-proof}
\input{trg-lang-semantics}
\input{mtr-instrumentation}

\input{correctness-proof}
}\fi

\end{document}

%% file: macros.tex
\newcommand{\ap}{\pi}
\newcommand{\aps}{\mathit{AP}}
\newcommand{\args}{\nu}
\newcommand{\atomfld}{\mathcal{A}}
\newcommand{\code}[1]{{\small \texttt{#1}}}
\newcommand{\concr}[1]{\mathsf{Expand}_{#1}}
\newcommand{\cortado}{\textsc{Cortado}\xspace}
\newcommand{\fdg}{\mathcal{G}}
\newcommand{\holds}{\mathit{h}}
\newcommand{\lockhold}{\mathcal{H}_{L}}
\newcommand{\lockblock}{\mathcal{B}_{L}}
\newcommand{\lockmap}{\mathcal{L}}
\newcommand{\locknotif}{\mathcal{N}_{L}}
\newcommand{\mtr}{\mathit{M}}
\newcommand{\predmap}{\mathcal{P}}
\newcommand{\pstate}{\sigma}
\newcommand{\sems}{\mathcal{T}}
\newcommand{\semtupple}{(\sigblock, \signotif, \lockhold, \lockblock, \locknotif)}
\newcommand{\sigblock}{\mathcal{B}_{S}}
\newcommand{\signotif}{\mathcal{N}_{S}}
\newcommand{\tbr}[1]{\llbracket\mkern-5mu\llbracket #1 \rrbracket\mkern-5mu\rrbracket}
\newcommand{\threadset}{\mathit{T}}
\newcommand{\tis}{h}
\newcommand{\toatom}{\mathit{a}}
\newcommand{\toolname}{\textsc{Cortado}\xspace}
\newcommand{\toolnameminus}{\textsc{Ablated}\xspace}
\newcommand{\wuntil}{\code{waituntil}\xspace}

%% file: abstract.tex
\begin{abstract}
A \emph{monitor} is a widely-used concurrent programming abstraction that encapsulates all shared state between threads. Monitors can be classified as being either \emph{implicit} or \emph{explicit}  depending on the  primitives they provide. Implicit monitors are much easier to program but typically not as efficient. To address this gap,  there has been recent research on automatically {synthesizing} explicit-signal monitors from  an implicit specification~\cite{expresso}, but prior work does not exploit all paralellization opportunities due to the use of a \emph{single} lock for the entire monitor. 
This paper presents a new technique for synthesizing \emph{fine-grained} explicit-synchronization protocols from implicit monitors. Our method is based on two key  innovations: First, we present a new static analysis for inferring \emph{safe interleavings} that allow violating mutual exclusion of monitor operations \emph{without} changing its semantics. Second, we use the results of this static analysis to generate a MaxSAT instance whose models correspond to correct-by-construction synchronization protocols. 
We have implemented our  approach in a tool called \toolname and evaluate it on monitors that contain parallelization opportunities. Our evaluation shows  that \toolname can synthesize synchronization policies that are competitive with, or even better than, expert-written ones on these benchmarks. 
\end{abstract}

%% file: intro.tex
\section{Introduction}

Concurrent programming is difficult
because it requires developers to consider interactions between multiple threads of execution
and mediate access to shared resources and data.
Programming languages can offer higher-level abstractions
to  reduce this complexity
by making concurrent programming more declarative.
One such abstraction is the \emph{monitor}~\citep{hansen:osp,hoare:monitors},
which is an object that encapsulates shared state
and allows threads access to it only through a set of \emph{operations},
between which the monitor enforces mutual exclusion.

Ideally, developers would implement monitors using \emph{implicit synchronization},
wherein the \emph{only}  synchronization primitive is a \code{waituntil(P)} operation that blocks threads until condition \code{P} is satisfied.
The compiler or runtime can then automatically generate the necessary \emph{explicit synchronization}
operations (locks, condition variables, etc.) to implement the monitor in a way that respects the semantics of the implicit monitor. 
However, automatically deriving an efficient explicit monitor
from its implicit specification is a challenging problem, and there have been several recent research efforts, including both run-time techniques like AutoSynch~\cite{hung:autosynch} and compile-time tools like Expresso~\cite{expresso}, to support implicit-synchronization monitors.

While these state-of-the-art approaches make it possible to program using implicit monitors,
 they still achieve sub-optimal performance
because they adhere closely to the monitor's mutual exclusion requirement.
They generally
 use a \emph{single} lock for the entire monitor and
allow access by \emph{at most one} thread at a time across all monitor operations.
In practice, however, many monitors can admit additional concurrency
while still preserving the \emph{appearance} of mutual exclusion.
For example, consider a FIFO queue monitor that provides  \code{take} and \code{put} operations. These two operations
can safely run concurrently unless the queue is empty or full,
as they will not access the same slot in the queue.
Today, realizing this \emph{fine-grained} concurrency requires expert developers
to fall back to hand-written explicit synchronization.
These implementations are subtle and error-prone,
and there is no easy way for  developers to determine
when they have extracted the maximum possible concurrency from such an implementation.

This paper presents a new technique which automatically synthesizes fine-grained explicit-synchronization monitors.
Our technique takes as input
an implicit monitor that specifies the desired operations
and automatically generates an implementation
that allows as much concurrency as possible between those operations
while still preserving the appearance of mutual exclusion.
The key idea is to decompose each monitor operation into a set of \emph{fragments}
and allocate a \emph{set} of locks to each fragment
to enforce the mutual exclusion requirement
while allowing as many fragments as possible to run concurrently.
The resulting implementation selectively acquires and releases locks at fragment boundaries
within each operation
and signals condition variables as needed.

At a high level, our approach operates in three phases to generate a high-performance explicit synchronization monitor from its implicit version:

\begin{itemize}[leftmargin=*]
    \item {\bf Signal placement:} First, we use an off-the-shelf technique~\citep{expresso}
to infer a signaling regime which determines where to insert signaling operations on condition variables.
While the output of this tool is sufficient to synthesize a single-lock implementation, it does not  admit any additional concurrency wherein different threads can perform monitor operations simultaneously. %
\item {\bf Static analysis:} Second, we perform static analysis to infer sufficient conditions for correctness. That is, the output of the static analysis is a set of conditions such that if the synthesized monitor obeys them,  it is guaranteed to be correct-by-construction. A key challenge for this static analysis is to determine which fragments can safely execute concurrently
without creating a potential violation of the monitor semantics.
The analysis simulates \emph{interleaving} each fragment between the fragments of other operations
and determines which possible interleavings are safe.
\item {\bf Synchronization protocol synthesis via MaxSAT:}
Finally, we reduce the synthesis problem to a maximum satisfiability (MaxSAT) instance from whose solution an explicit sychronization protocol can be extracted. The hard constraints in the MaxSAT problem enforce the correctness requirements extracted by the static analysis, while the soft constraints encode two competing objective functions: minimizing the total number of locks used,
while maximizing the number of pairs of fragments that can run concurrently.

\end{itemize}

We have implemented our proposed approach in a tool called \cortado that operates on Java monitors
and evaluated it on a collection of monitor implementations that are (1) drawn from popular open-source projects and (2) contain parallelization opportunities that can  be achieved via fine-grained locking.
Given only the implicit monitor as input,
\cortado synthesizes  explicit-synchronization monitors that perform as well as,
or better than, hand-written explicit implementations by expert developers.
Compared to  state-of-the-art automated tools for synthesizing explicit monitors~\citep{expresso},
\cortado-synthesized monitors extract more concurrency
and therefore perform much better (up to $39.1\times$) on heavily contended workloads.

In summary, this paper makes four main contributions:
\begin{itemize}[leftmargin=*]
\item A new technique for automatically synthesizing fine-grained monitor implementations
that admit the maximum possible concurrency.
\item A novel static static analysis for inferring safe interleaving opportunities between threads.
\item A MaxSAT encoding to automate reasoning
about \emph{both} the correctness and performance of the synthesized explicit-synchronization monitor.
\item An implementation of our technique, \cortado, that outperforms both state-of-the-art automated
tools and expert-written code on benchmarks that can be parallelized via fine-grained locking.
\end{itemize}

%% file: overview.tex
\section{Overview}
\label{sec:overview}

In this section, we give an overview of our approach through a motivating example. 
Given the implicit-synchronization monitor shown in 
Figure~\ref{fig:motiv-src}, our goal is to automatically synthesize an efficient and semantically equivalent  explicit-synchronization
monitor like the one presented in Figure~\ref{fig:motiv-trg}. In what follows, we walk through this example and describe how our technique is able to automatically generate the code in Figure~\ref{fig:motiv-trg}.

\begin{figure}
  \begin{subfigure}{0.5\textwidth}
  \scriptsize
  \begin{minted}[linenos,xleftmargin=3em,autogobble,escapeinside=||]{java}
    class ArrayBlockingQueue {
      int first = 0, last = 0, count = 0;                   |\label{ln:src:counters}|
      Object[] queue;                                       |\label{ln:src:queue}|

      ArrayBlockingQueue(int capacity) {
        if (capacity < 1)
          throw new IllegalArgumentException();
        this.queue = new Object[capacity];
      }

      void put(Object o) {
        // Fragment 1
        waituntil(count < queue.length);                      |\label{ln:src:put:wait}|
        // Fragment 2
        queue[last] = o;                                      |\label{ln:src:enqueue}|
        // Fragment 3
        last = (last + 1) %
        // Fragment 4
        count++;                                              |\label{ln:src:countup}|
      }

      Object take() {
        // Fragment 5
        waituntil(count > 0);                                 |\label{ln:src:take:wait}|
        // Fragment 6
        Object r = queue[first];                              |\label{ln:src:dequeue1}|
        queue[first] = null;                                  |\label{ln:src:dequeue2}|
        // Fragment 7
        first = (first + 1) %
        // Fragment 8
        count--;                                              |\label{ln:src:countdown}|
        return r;
      }
    }
  \end{minted}
  \caption{Implicit-synchronization ArrayBlockingQueue.}\label{fig:motiv-src}
  \end{subfigure}%
  \begin{subfigure}{0.5\textwidth}
    \scriptsize
    \begin{minted}[linenos,xleftmargin=3em,autogobble,escapeinside=||]{java}
    class ArrayBlockingQueue {
      int first = 0, last = 0; Object[] queue;
      AtomicInteger count = new AtomicInteger(0);             |\label{ln:trg:atomic:fld}|

      Lock putLock = new Lock(), takeLock = new Lock();       |\label{ln:trg:locks}|
      Condition notFull = putLock.newCondition();             |\label{ln:trg:cv:notfull}|
      Condition notEmpty = takeLock.newCondition();           |\label{ln:trg:cv:notempty}|
      // Constructor is the same as the implicit version.
      
      void put(Object o) {
        putLock.lock()
        while (count.get() == queue.length)                    |\label{ln:trg:atomic:get1}|
          notFull.await();                                     |\label{ln:trg:wait:notfull}|
        queue[last] = o;
        last = (last + 1) %
        int c = count.getAndIncrement();                       |\label{ln:trg:atomic:inc}|
        putLock.unlock();
        if (c == 0) {                                          |\label{ln:trg:put:condsig}|
          takeLock.lock();                                     |\label{ln:trg:putsig:start}|
          notEmpty.signalAll();
          takeLock.unlock();}}                                   |\label{ln:trg:putsig:end}|

      Object take() {
        takeLock.lock();
        while (count.get() == 0)                                |\label{ln:trg:atomic:get2}|
          notEmpty.await();                                     |\label{ln:trg:wait:notempty}|
        Object r = queue[first];
        queue[first] = null;
        first = (first + 1) %
        int c = count.getAndDecrement();                        |\label{ln:trg:atomic:dec}|
        takeLock.unlock();
        if (c == queue.length) {                                |\label{ln:trg:take:condsig}|
          putLock.lock();                                       |\label{ln:trg:takesig:start}|
          notFull.signalAll();
          putLock.unlock();}                                 |\label{ln:trg:takesig:end}|
        return r;}}
  \end{minted}
  \caption{Explicit-synchronization ArrayBlockingQueue.}\label{fig:motiv-trg}
  \end{subfigure}
  \caption{Motivating example.}\label{fig:motiv-ex}
\end{figure}

\subsection{Implicit-synchronization monitor}

Our technique takes as input an \emph{implicit-synchronization monitor} that specifies which operations should execute atomically and when certain operations are allowed to proceed but does not fix a specific synchronization protocol for realizing that behavior. 
For example, Figure~\ref{fig:motiv-src} shows an implicit monitor that implements
a limited capacity blocking queue via a bounded circular array
buffer. This monitor defines two operations, \code{put} and \code{take}, that execute atomically (i.e., the body of each method must appear to execute as one indivisible unit).  The \code{put} operation adds an
object if the queue is not full, and \code{take}
removes an object  if the queue is not empty.
If one of these method calls cannot proceed
(i.e., queue is full or empty), the monitor blocks the calling
thread's execution using a \code{waituntil} statement until the operation can be executed.
For example, the \code{waituntil} statement at line~\ref{ln:src:put:wait} in \code{take} blocks execution until there is at least one object in the queue.

As Figure~\ref{fig:motiv-src} illustrates, implicit-synchronization monitors make concurrent programming simpler because they are \emph{declarative}: they  merely state which operations are atomic and when operations can proceed, but they do \emph{not} specify a particular synchronization protocol for realizing that desired behavior. However, most programming languages do not offer implicit synchronization facilities; so, concurrent programs must instead be implemented in terms of \emph{explicit} synchronization constructs such as locks and condition variables, as we discuss next.

\subsection{Explicit-synchronization monitor}
Figure~\ref{fig:motiv-trg} shows an explicit-synchronization implementation of the bounded queue from Figure~\ref{fig:motiv-src} that is written by an expert. This implementation uses two distinct locks,
\code{putLock} and \code{takeLock}, to protect the \code{put} and \code{take} methods respectively. The explicit-synchronization monitor also uses an atomic integer for the \code{count} field, transforming reads into \code{get()} calls (e.g., line~\ref{ln:trg:atomic:get1})
and writes into the appropriate atomic method (e.g., \code{count.getAndIncrement()} on line~\ref{ln:trg:atomic:inc}).
The expert-written monitor performs explicit signaling via condition variables
\code{notFull} and \code{notEmpty}
that are associated with  \code{putLock} and \code{takeLock}  respectively.
When a thread cannot execute one of these operations, it calls \code{await} on the
appropriate condition variable  to block its execution (lines~\ref{ln:trg:wait:notfull}
and~\ref{ln:trg:wait:notempty}).
A thread blocked in \code{put} can only be unblocked by a corresponding \code{take}
that frees up space in the queue.
To do so, the \code{take} must acquire \code{putLock} and perform a signal
operation on condition variable \code{notFull}
(lines~\ref{ln:trg:takesig:start}--\ref{ln:trg:takesig:end}). The
logic for \code{take} is symmetric (lines~\ref{ln:trg:putsig:start}--\ref{ln:trg:putsig:end}).

Although the expert-written version has {more} locks than
a single global-lock implementation,
its performance will often be better:
Introducing two locks allows \code{put} and \code{take}  to execute concurrently,
although multiple concurrent \code{put}s are still serialized, as are multiple \code{take}s.
A single global lock would admit no concurrency in this case
and would still incur the same synchronization overhead of acquiring and releasing a lock on every method call.
The expert implementation mitigates the overhead of having two locks by acquiring locks selectively:
\code{take} only acquires the \code{putLock}
if it is possible for there to be a \code{put} operation currently blocked waiting for space in the queue,
which happens only if the queue was full when \code{take} ran
(the \code{put}/\code{takeLock} case is symmetric).
This example demonstrates the intricacy of synthesizing fine-grained locking protocols:
instead of only minimizing the total number of locks,
we must also try to maximize the available concurrency.

\subsection{Our Approach}

Our tool \cortado automatically synthesizes the efficient explicit-synchronization monitor in Figure~\ref{fig:motiv-trg}
given the implicit version from Figure~\ref{fig:motiv-src}.
It does so in three phases:
First, it infers when and how signaling operations should take place. Second, it performs static analysis to infer sufficient conditions for the synthesized  monitor to be correct. Third, it encodes the synchronization protocol synthesis problem as a MaxSAT instance and uses a model of the MaxSAT problem to generate an explicit-sychronization monitor. Since prior work can already handle the first phase, we only focus on the the latter two phases in the following discussion.

\paragraph{\textbf{Granularity}.} The granularity of our synthesized locking protocol is at the level of \emph{code fragments},  where each fragment is a single-entry region of code within a single method. For example,
the fragments chosen for the blocking queue example are indicated by comments in Figure~\ref{fig:motiv-src}.
Fragments are the indivisible unit of concurrency in our approach:
we aim to maximize the number of fragments that can run concurrently,
but we do  not modify the code within a fragment to introduce extra concurrency
(e.g., by removing data races).
Hence, the explicit monitor synthesized by our approach
acquires and releases locks only at fragment boundaries.

\paragraph{\textbf{Static analysis}.} To ensure correctness of the synthesized monitor, our technique needs to enforce the following three key requirements:
\begin{enumerate}[leftmargin=*]
    \item {\bf Data-race freedom:} Fragments that involve a data race must not be able to run concurrently.
    \item {\bf Deadlock freedom:}  Locks must be acquired and released in an order that prevents deadlocks.
    \item {\bf Atomicity:} Each monitor operation should \emph{appear} to take place as one indivisible unit. That is, even though the implementation can %
    allow thread interleavings inside monitor operations, the resulting  state should be equivalent to one where each method executes truly atomically. 
\end{enumerate}

Here, the second requirement (i.e., deadlock freedom) does not necessitate any static analysis, as we can prevent deadlocks by imposing a static total order $\preceq$ on locks~\cite{birrell1989introduction} and ensuring that locks are acquired and released in a manner that is consistent with $\preceq$. However, in order to ensure data-race freedom and atomicity, we need to perform static analysis of the source code to identify (1) code fragments that have a data race, and (2) interleaving opportunities between code fragments. Since detection of data races is a well-studied problem, the novelty of our static analysis lies in identifying safe interleaving opportunities. Hence, the key question addressed by our  analysis is the following: Given a code fragment $f$ executed by thread $t$, and two consecutive code fragments $f_1, f_2$ executed by a different thread $t'$, is it safe to interleave the execution of $f$ \emph{in between} $f_1$ and $f_2$ while ensuring that monitor operations \emph{appear} to take place atomically?

To answer this question, our method performs a novel static analysis to identify a set of such \emph{safe interleavings}. For instance, going back to the running example, our analysis determines that it is safe to interleave the execution of fragment 4 in Figure~\ref{fig:motiv-src} in between fragments 5 and 6  by checking a number of commutativity relations between  code fragments. In this instance, since our analysis proves that fragment 4 left-commutes~\cite{reduction-lip} with fragment 5 and right-commutes~\cite{reduction-lip} with 6 and all of its successors,  we identify this as a safe interleaving opportunity. On the other hand, our analysis concludes that interleaving fragment 4 in between 1 and 2 is \emph{not} safe because  fragment 4 does not left-commute with fragment 1 --- intuitively, this is because fragment 4 can falsify predicate \code{count < queue.length} that appears in the \code{waituntil} statement of fragment 1. 

\paragraph{\textbf{MaxSAT overview}.}  Once we identify possible data races and safe interleavings via static analysis,  we use this information to generate a MaxSAT instance whose solution corresponds to a fine-grained synchronization protocol.
Specifically, our MaxSAT encoding uses a variable $\holds_{f_i}^{l_j}$ to indicate that  code fragment $f_i$ must hold lock $l_j$ and generates both hard constraints (for correctness) and soft constraints (for efficiency) over these variables. Thus, if the MaxSAT solver returns a model in which variable  $\holds_{f_i}^{l_j}$  is assigned to true, this means that the synthesized code must  acquire lock $l_j$ prior to executing fragment $f_i$. Similarly, our MaxSAT encoding introduces a variable $a_\emph{fld}$ indicating that field \emph{fld} should be implemented using an atomic type.

The hard constraints in our MaxSAT encoding correspond to the three correctness requirement mentioned earlier, namely (1) data race prevention, (2) deadlock freedom, and (3) atomicity. On the other hand, soft constraints encode our optimization objective. In what follows, we give a brief overview of the different types of constraints in our encoding, focusing only on  constraints that involve lock acquisition variables $\holds_{f_i}^{l_j}$. {However, it is worth noting that our technique also generates constraints on atomic variables $a_\emph{fld}$ and can automatically convert fields to atomic types  whenever doing so is safe and more efficient than  introducing a lock.}

\paragraph{\textbf{Data-race freedom.}} Given a pair of code fragments $(f_i, f_j)$ that have a potential data race according to the static analysis, our MaxSAT encoding introduces hard constraints 
of the form $\bigvee_k (\holds_{f_i}^{l_k} \land \holds_{f_j}^{l_k})$ stating that $f_i$ and $f_j$ must share at least one common lock.
For example, in Figure~\ref{fig:motiv-src}, our analysis determines that fragments 4 and 8 cannot run in parallel since they both write to the same memory location \code{count}. Thus, the  MaxSAT instance contains boolean constraints to make sure that two different threads cannot execute \code{count-{}-} and \code{count++} at the same time.

\paragraph{\textbf{Deadlock freedom.}} Our approach precludes deadlocks by imposing a  total order $\preceq$ on locks. In particular, it enforces that a thread $t$ can only acquire lock $l$  if $t$ does \emph{not} already hold any lock $l'$ where $l' \prec l$.  For example,  in Figure~\ref{fig:motiv-src},  suppose the locking protocol determines that fragments 1 and 2 must hold all locks in  sets $L_1$ and $L_2$ respectively. Between executing the two fragments, the  code will need to acquire all locks in $L_2 \setminus L_1$. Hence, we add constraints $i < j$
for every pair of locks $l_j \in L_2 \setminus L_1$ and $l_i \in L_1 \cap L_2$
so that those  locks can be acquired while respecting the order $\preceq$.

\paragraph{\textbf{Atomicity.}} Our MaxSAT encoding also includes constraints to ensure that monitor operations appear to execute atomically.  Suppose that our static analysis determines that a thread cannot safely execute code fragment $f$ \emph{in between} some other thread's execution of code fragments $f_1$ and $f_2$. To prevent such an unsafe interleaving,  we add hard constraints to ensure that fragments $f, f_1$, and $f_2$ all share at least one common lock. For example, since our analysis determines that fragment 4 (\code{count++}) cannot be interleaved with any other pair of fragments in the same method \code{put} (running concurrently on a different thread), our MaxSAT encoding  includes a hard constraint asserting that fragment 4 must share a lock with all other fragments in the \code{put} method.

\paragraph{\textbf{Soft constraints}.}
Because the efficiency of the synthesized code depends on both
the allowed parallelization opportunities as well as the number of locks,  our optimization objective tries to  \emph{minimize} the number of locks  and \emph{maximize} the number of fragments that can run
in parallel. To encode the latter objective, our MaxSAT encoding includes soft contraints asserting that any two parallelizable fragments must \emph{not} share a lock. On the other hand, to encode the former objective, we add a soft constraint stating that no fragment in $m$ is holding lock $l$. 

\paragraph{\textbf{Monitor generation.}}
A solution of the generated MaxSAT instance determines (a) which fragments should hold which locks, (b) which fields should be implemented using atomic types, and (c) which locks should be associated with which condition variables. Thus, together with the output of the signal placement technique~\cite{expresso}, a model of the MaxSAT problem can be automatically translated into the target monitor implementation. For our running example, \toolname synthesizes precisely the implementation in Figure~\ref{fig:motiv-trg} given the implicit monitor  of Figure~\ref{fig:motiv-src}.

%% file: langs.tex
\section{Preliminaries}
\label{sec:prelim}

\begin{figure}
  \begin{subfigure}{0.5\textwidth}
    \begin{grammar}
      \let\syntleft\relax
      \let\syntright\relax
      <Monitor $\mtr$> $::=$ \code{monitor $\mtr$ \{}$(\mathit{fld}$ | $\mathit{m})\text{*}\code{\}}$

      <Field $\mathit{fld}$> $::=$ $\tau\ f := e$

      <Method $\mathit{m}$> $::=$ $\mathit{m}(\vec{v})\code {\{}  ccr\text{*}\code{\}}$

      <CCR $ccr$> $::=$ $\wuntil(p)\code{;} \mathit{s}$

      <Stmt $s$> $::=$ $\mathit{skip}$ | $v := e$ | $v.f := e$
                  \alt $\mathit{v}.\mathit{m}(\vec{e})$ | $[\code{if } (e)]? \code{ goto l}$
                  \alt $\mathit{ls}_1;\mathit{ls}_2$

       <LStmt $ls$> $::=$ $\code{l:}?\ \mathit{s}$
    \end{grammar}
    \caption{Implicit-synchronization monitor language.}\label{fig:src-lang}
  \end{subfigure}%
  \begin{subfigure}{0.5\textwidth}
    \begin{grammar}
      \let\syntleft\relax
      \let\syntright\relax
      <Monitor $\mtr$> $::=$ \code{monitor $\mtr$ \{}$(\mathit{fld}$ | $\mathit{sync}$ | $\mathit{m})\text{*}\code{\}}$

      <Field $\mathit{fld}$> $::=$ $\tau\ f := e$

      <Sync $\mathit{sync}$> $::=$ $\code{Lock}\ \mathit{l} := \code{new Lock()}$
                             \alt $\code{CondVar}\ \mathit{cv}\ :=\ \mathit{l}.\code{newCondVar()}$
                             \alt $\code{Atomic}[\tau]\ a\ :=\ e$

      <Method $\mathit{m}$> $::=$ $\mathit{m}(\vec{v})\code { \{}\ ccr\text{*}\ \code{\}}$

      <CCR $ccr$>  $::=$ $(\mathit{ls})\text{*}$

      <Stmt $s$> $::=$ $\mathit{skip}$ | $v := e$ | $v.f := e$
                  \alt $\mathit{v}.\mathit{m}(\vec{e})$ | $[\code{if } (v)]? \code{ goto l}$
                  \alt $\mathit{ls}_1;\mathit{ls}_2$
                  \alt $a_{pre} := a.\code{update}(\lambda\chi.e)$

      <LStmt $ls$> $::=$ $\code{l:}?\ \mathit{s}$
    \end{grammar}
    \caption{Explicit-synchronization monitor language.}\label{fig:trg-lang}
  \end{subfigure}%
  \caption{Source \& target languages. We use $e$ and $p$ for
    expressions and predicates respectively.}\label{fig:langs}
\end{figure}

In this section, we describe our source and target
languages and define what it means for an explicit synchronization monitor to correctly implement an implicit one.

\subsection{Background on Monitors}

In this work, %
we assume that all shared resources between threads are handled by a \emph{monitor class} $M$ which consists of fields $F$ and set of operations (methods) $O$.
The  fields $F$ constitute the \emph{only} shared state between threads,  which can only  access  shared state by performing one of the monitor operations $o \in O$.  These operations can be performed by an arbitrary, yet fixed, number of threads, and  locations reachable through  arguments  are  assumed to be thread-local. 
We represent each thread by a unique
identifier from set $\threadset \subseteq \mathds{N}$, and we model memory
locations using \emph{access paths} ($\aps$)~\cite{aps} of the form $\ap = v(.f)*$, consisting of a base variable $v$ optionally followed by a
finite sequence of field accesses. We also assume that a special \code{this} variable  stores the memory location of the
monitor object. %

\begin{definition}
{\textbf{(Monitor State)}.}
A monitor state $\pstate : \threadset \times
\aps \rightarrow \mathds{N}$ is a mapping from pairs $(t, \ap)$ (where $t$ is a thread identifier and $\ap$ an access path) to a value. 
 
\end{definition}

\subsection{Source Language}
\label{sec:src-lang}

Our source language, presented
in Figure~\ref{fig:src-lang}, corresponds to \emph{implicit synchronization monitors} without explicit locking or signaling.  
The body of each monitor operation
consists of a sequence of so-called \emph{Conditional Critical Regions
(CCRs)}~\cite{hoare:ccr}, which in turn consist of a \wuntil statement
followed by one or more regular non-blocking statements. We refer to
the predicate of the \wuntil statement of a CCR as its \emph{guard}
and to the rest of the statements as its \emph{body}.
A thread  executes the body of the CCR
atomically if its guard evaluates to true; otherwise it suspends its
execution and exits the monitor until the predicate becomes
true. 
More formally, the semantics of our source language are defined via the notion of an \emph{implicit monitor history}:

\begin{definition}{\textbf{(Implicit monitor history)}.} Given a set of threads interacting with each other through monitor $M_s=(F, O)$, an \emph{implicit monitor history} $\tis_i$  is a sequence  $(ccr_1, t_1) \ldots (crr_n, t_n)$ where each $ccr_i$ is a CCR  of $\mtr_s$
and $t_i$ is a thread identifier.
\end{definition} 

Given history $\tis_i$, we define an \emph{argument mapping} $\args_i$ to be a list whose $i$'th element maps formal parameters of  $\emph{Method}(ccr_i)$  to their actual value for each event $(ccr_i, t_i)$ in  $\tis_i$.

\begin{definition}{\textbf{(Implicit monitor semantics)}.} Given a monitor $M_s$, initial state $\sigma$, and monitor history $\tis_i$ with argument mapping $\args_i$, the operational semantics of $M$ is defined using a judgment $\mtr_s \vdash (\tis_i, \args_i, \sigma) \Downarrow \sigma'$ indicating that the new monitor state  is $ \sigma'$ after executing  $\tis_i$ on   state $\sigma$.
\end{definition}

Because our source language is very similar to the one used in~\citet{expresso}, we omit a formal definition of the operational semantics. Following that work, we also consider an implicit history to be valid only if it respects the program order of the input monitor.

\subsection{Target Language}\label{sec:target-lang}

Figure~\ref{fig:trg-lang} presents the language of  explicit-synchronization monitors. The overall structure of this
target language is similar to the source language but with a few important
differences. First,  an explicit monitor  contains locks, conditional
variables, and atomic fields,  collectively referred to as \emph{synchronization variables}. Second, CCRs in the
target language do not contain \wuntil statements; instead, the
logic of a \wuntil statement is
 implemented by calling methods on the appropriate condition
variable. We assume that
synchronization variables support  all the standard synchronization
operations present in modern concurrent languages (e.g., \code{await},
\code{signal}, \code{signalAll}, etc.). Finally, our target language
contains a special  \code{update} statement for performing updates on atomic
fields: it takes as argument an atomic
field $a$ and a unary function $f$ and updates the value of $a$
\emph{atomically} as $f(a)$. For instance,
 the statement
\code{c$_{pre}$ := c.update($\lambda\chi.\chi+1$)} atomically
increments \code{c} by one and stores the value of \code{c} before
the update in \code{c$_{pre}$}.

\begin{definition}
 {\textbf{(Explicit monitor history)}.} Given a set of threads executing  in monitor $M_t=(F, O)$, an \emph{explicit monitor history} $\tis_e$  is a sequence  $(s_1, t_1) \ldots (s_n, t_n)$ where each $s_i$ is a (non-composite) statement of a monitor operation $o \in O$ and $t_i$ is a thread identifier. 
\end{definition}

Leveraging  the same  notion of \emph{argument mappings} defined in Section~\ref{sec:src-lang},  we define explicit monitor semantics as follows:

\begin{definition}
 {\textbf{(Explicit monitor semantics)}.} Given a monitor $M_t$, initial state $\sigma$, and monitor history $\tis_e$ with argument mapping $\args_e$, the operational semantics of $M_t$ is defined using a judgment $M_t \vdash (\tis_e, \args_{e}, \sigma) \downarrow \sigma'$ indicating that the new  state  is $ \sigma'$ after executing   $\tis_e$ on initial  state $\sigma$.
\end{definition}

\iffull
{
The full operational semantics of our target language is given in Appendix~\ref{sec:trg-lang-sem}.
}
\else{
The full operational semantics of our target language is given in the extended version of the paper under supplemental material.
}
\fi

\subsection{Relating Implicit and Explicit Histories}

\begin{figure}
    \begin{subfigure}{0.4\textwidth}
    \scriptsize
      \begin{minted}[xleftmargin=3em,autogobble,escapeinside=||]{java}
          class M {
            int x = 0, y = 0, z = 0;
            void foo() { x++; y++; }
            void bar() { z++; } }
      \end{minted}
      \caption{A simple implicit monitor.}
      \label{fig:mtr-sample}
    \end{subfigure}%
     \begin{subfigure}{0.6\textwidth}
     \scriptsize
      \begin{minted}[xleftmargin=3em,autogobble,escapeinside=||]{java}
          class M {
            int x = 0, y = 0, z = 0;
            Lock l1 = new Lock(), l2 = new Lock();
            void foo() { l1.lock(); x++; y++; l1.unlock(); }
            void bar() { l2.lock(); z++; l2.unlock(); } }
      \end{minted}
      \caption{An explicit monitor implementation of Figure~\ref{fig:mtr-sample}.}
      \label{fig:mtr-sample-expl}
    \end{subfigure}
    \begin{subfigure}{\textwidth}
    \[
        \begin{array}{l}
           \tis_i = (foo, t_1)(bar, t_2)\\
           \tis_e = (\code{{\small l1.lock()}}, t_1)(\code{{\small x++}}, t_1)(\code{{\small y++}},t_1)(\code{{\small l1.unlock()}}, t_1)(\code{{\small l2.lock()}}, t_2)(\code{{\small z++}}, t_2)(\code{{\small l2.unlock()}}, t_2)\\
           \tis_e' = (\code{{\small l1.lock()}}, t_1)(\code{{\small x++}}, t_1)(\code{{\small l2.lock()}}, t_2)(\code{{\small y++}},t_1)(\code{{\small z++}}, t_2)(\code{{\small l1.unlock()}}, t_1)(\code{{\small l2.unlock()}}, t_2)
        \end{array}
    \]
    \caption{Examples of implicit and explicit histories.}
    \label{fig:histories}
    \end{subfigure}
    \caption{A simple implicit monitor and its explicit implementation.}
\end{figure}

In order to formalize the correctness of our approach, we need to relate an implicit history $h_i$ of a source monitor $M_s$ with an explicit history $h_e$ of its corresponding target version $M_t$. 
Because every history of an implicit monitor $\mtr_s$ \emph{induces} a corresponding history of its explicit version $\mtr_t$, we  define an operation called that $\concr{}$ that ``translates" an implicit history to an explicit one. That is, given an implicit history $h_i$ with argument mapping $\args_i$ and state $\pstate$,   $\concr{\mtr_t}(\tis_i, \args_i, \pstate)$  returns a pair $(\tis_e, \args_e)$, where $\tis_e$ is a history of $\mtr_t$ containing all statements executed by $\tis_i$
under initial state $\pstate$ and $\args_e$ is the argument mapping for $\tis_e$. 

\begin{example}
Consider the  implicit monitor of Figure~\ref{fig:mtr-sample} and its explicit counterpart in Figure~\ref{fig:mtr-sample-expl}. For histories $\tis_i$ and $\tis_e$ from Figure~\ref{fig:histories} we have $\concr{\mtr_t}(\tis_i, \args_i, \pstate) = (\tis_e, \args_e)$ for some $\args_i$, $\args_e$.
\end{example}

Using this $\concr{}$ operation, we can classify explicit histories as being sequential or interleaved:

\begin{definition}{{\bf (Sequential history)}}\label{def:seq-hist} 
Let $\mtr_t$ be an explicit monitor implementation of $\mtr_s$. 
We say that an explicit history $\tis_{e}$ of monitor $\mtr_{t}$ with argument mapping $\args_e$ is \emph{sequential} iff there exist a history $h_i$ of $\mtr_s$, argument mapping $\args_i$, and initial state $\pstate$ such that $\concr{\mtr_t}(\tis_i,\args_i,\pstate) = (\tis_{e}, \args_{e})$.
\end{definition}

In other words, a sequential history corresponds to an execution  in which statements of the explicit monitor are not interleaved between threads. 

\begin{example}
 Going back to Figure~\ref{fig:histories}, history $\tis_e$ is sequential but $\tis_e'$ is not.
\end{example}

Next, we introduce the notion of \emph{well-formed histories}, which, intuitively, respect the program order of the original implicit monitor:

\begin{definition}{{\bf (Well-formed history)}}
Let $\Pi(\tis, t)$ be the \emph{projection} of $\tis$ onto thread $t$ (i.e.,  it filters out all elements of $\tis$ not involving thread $t$). 
We say that a history $\tis_e$ of  $\mtr_t$ is \emph{well-formed} iff, for every thread $t$, there exists sequential histories $h_{e}^1,\dots,h_{e}^n$ such that $\Pi(\tis_{e}, t) = h_{e}^1\cdots h_{e}^n$.
\end{definition}

Intuitively, well-formed histories respect program dependence in the original monitor for every thread. By definition, every sequential history is also well-formed. In the remainder of this paper, we implicitly mean \emph{well-formed} explicit history whenever we refer to an explicit history. %

\begin{example}
 Histories $\tis_e$, $\tis_e'$ from Figure~\ref{fig:histories} are both well-formed. However, the following history is not well-formed because it does not respect program order: $(\code{{\small l2.unlock()}}, t)(\code{{\small l2.lock()}}, t)$
\end{example}

\begin{definition}{\bf (Interleaved history)}
 We say that a history $\tis_{e}$ of $\mtr_{t}$ is  \emph{interleaved} iff  it is (1) well-formed and (2) not sequential.
\end{definition}

\begin{example}
 History $\tis'_e$ from Figure~\ref{fig:histories} is interleaved.
\end{example}

Next, we define what it means for an explicit history to \emph{simulate} an implicit history.

\begin{definition}{\bf (Simulation relation).} Let $M_t$ be an explicit version of implicit monitor $M_s$.
 We say that an explicit history $\tis_{e}$ of $\mtr_{t}$ with argument mapping $\args_e$  simulates $(\tis_i, \args_i)$ of  $\mtr_s$ on input $\pstate$, denoted $(\tis_{e}, \args_{e}) \backsim (\tis_{i}, \args_{i})$, if there exist sequential history $\tis'_{e}$ and $\args'_{e}$ such that: 
\[
(1) \  \forall t.\ \Pi(\tis_{e}, t) = \Pi(\tis'_{e}, t) \quad \quad \emph{and} \quad \quad 
(2) \ \concr{\mtr_{t}}(\tis_i, \args_i, \sigma) = (\tis'_{e}, \args'_{e}).
\]
\end{definition}

In other words, $\tis_{e}$ simulates a history of the original monitor if it is a  (well-formed) permutation of some sequential history $\tis'_{e}$ of the explicit monitor $\mtr_t$. 

\begin{example}
 Going back to Figure~\ref{fig:histories}, we have $(\tis_e', \args') \backsim (\tis_i, \args)$ for some $\args$, $\args'$.
\end{example}

\subsection{Correctness of Explicit-Synchronization Monitors}
Using the concepts introduced in the previous section, we now formalize what it means for an explicit monitor to \emph{correctly implement} an implicit one. %

\begin{definition}
 {\textbf{(State equivalence)}}\label{def:sequiv}
Let $\pstate$ be a program state of an implicit
monitor $\mtr_s$ and  $\pstate'$  that of an explicit monitor $\mtr_t$. We
say that $\pstate$ and $\pstate'$ are equivalent modulo $\mtr_s$,
denoted $\pstate \equiv_{\mtr_s} \pstate'$, iff for all $(t, \ap)$ in the domain of $\pstate$,  we have
$\pstate(t, \ap) = \pstate'(t, \ap)
$

\end{definition}
Intuitively, this notion of equivalence between two monitor states
ignores any additional synchronization fields and local variables introduced by
translating $\mtr$ to an explicit-synchronization monitor.
Finally, we can define the correctness of an explicit monitor as follows:

\definition{\textbf{(Correctness)}} \label{def:correctness}
 We say that an explicit monitor $M_t$ correctly implements an implicit monitor $M_s$, denoted as $\mtr_s \sim \mtr_t$, iff for all input states $\sigma_s, \sigma_t$ s.t. $\sigma_s \equiv_{\mtr_{s}} \sigma_t$, we have:
\begin{enumerate}[leftmargin=*]
\setlength\itemsep{0.5em}
    \item $\forall \tis_i, \args_i.\  \mtr_s \vdash (\tis_i, \args_i, \pstate_s) \Downarrow \pstate_s' \Longrightarrow \left(\mtr_t \vdash (\concr{\mtr_t}(\tis_i, \args_i, \pstate_s), \pstate_t) \downarrow \pstate_t' \ \land \  \pstate_s' \equiv_{\mtr_s} \pstate_t'\right)$
    \item $\forall \tis_e, \args_e.\  \mtr_t \vdash (\tis_e, \args_e, \pstate_t) \downarrow \pstate_t' \Longrightarrow \left( \exists \tis_i,\args_i.\ (\tis_e, \args_e) \!\backsim \!(\tis_i, \args_i) \land \mtr_s \vdash (\tis_i, \args_i, \pstate_s) \Downarrow \pstate_s' \land \pstate_s' \equiv_{\mtr_s}\pstate_t' \right)$
\end{enumerate}

The first correctness condition simply states that $\mtr_t$ does not eliminate any feasible behaviors of  $\mtr_s$. The  second condition, on the other hand,  states that every feasible history of $\mtr_t$ {simulates}  \emph{some} implicit history  that results in the same state. Intuitively, this means that all statement interleavings allowed by $\mtr_t$ provide the illusion that all operations of $\mtr_s$ are executed atomically.

%% file: algorithm.tex
\section{Main Algorithm}
\label{sec:algorithm}

In this section, we present our main synthesis algorithm.  Specifically, Section~\ref{sec:prelim-np} introduces
some preliminary definitions and proves an NP-completeness result to
justify the reduction to MaxSAT. Then, Section~\ref{sec:synth-alg}
presents the high-level algorithm, Section~\ref{sec:analysis}
presents the static analysis for inferring safe interleavings, and
Sections~\ref{sec:max-sat} presents the details of the MaxSAT
encoding.

\subsection{Fragment Dependency Graphs and NP-Completeness}
\label{sec:prelim-np}

Our main synthesis algorithm is parametrized over a partitioning of the input monitor  into code fragments, where each code fragment defines a unit of computation that we need to assign locks to. In this section, we clarify our assumptions about these code fragments and prove the NP-completeness of the problem for a given choice of partition.

First, to define what we mean by a valid partition, we represent each method of the monitor as a standard control-flow graph (CFG), where each atomic statement  belongs to its own basic block. Given a control-flow graph $G$ and node  $n$, we write $\emph{Preds}(G, n)$ to indicate the predecessor nodes of $n$ in $G$ and $\emph{Succs}(G, n)$ to indicate its successors. Then, a valid partition of a method into code fragments is defined as follows:

\definition{\textbf{(Partition)}} Let $G = (V, E)$ be the CFG representation of a method. Then, a partition of this method is a set  of CFGs $\{G_1, \ldots, G_n\}$ with $G_i = (V_i, E_i)$ such that:
\begin{enumerate}
    \item $V = \uplus_{i=1}^n V_i$ and $E_i = E \cap (V_i \times V_i)$
\item For every $G_i$, there is at most one node $n \in V_i$ such that $\emph{Preds}(G, n) \not \subseteq V_i$
\item Every \code{waituntil}($p$) statement must belong to its own $G_i$ --- i.e., if a node $n \in V$ is a waituntil statement, then there exists a $G_i = (\{n\}, \emptyset) $
\end{enumerate}

Intuitively, a \emph{partition}  is a set of sub-CFGs  such that (1) these sub-CFGs cover all nodes of the original CFG, (2) each sub-CFG has a unique entry node,  and (3) waituntil statements belong to their own sub-CFG. We refer to the code snippet represented by each sub-CFG as a \emph{code fragment} and define a notion of \emph{fragment dependency graph (FDG)} as follows:

\begin{definition}{\textbf{(FDG)}}\label{def:fdg} Given a method $m$ with CFG $G=(V,E)$ and a partition of $G$ into $\{G_1, \ldots, G_n\}$, a \emph{fragment dependency graph (FDG)} is a directed acyclic graph  $(V', E')$ such that (1) every $f_i \in V'$ is the code fragment associated with $G_i$; (2) there is an edge $(f_i, f_j) \in E'$ iff there is an edge in $G$ from any exit node of $G_i$ to the entry node of $G_j$.
\end{definition}

\begin{example}
Figure~\ref{fig:fdg-example} presents the FDG of method \code{take} for the partition in Figure~\ref{fig:motiv-src}, 
\end{example}

\noindent\begin{minipage}{\textwidth}
\useparinfo
\begin{wrapfigure}{r}{0.5\textwidth}
\includegraphics[width=0.5\textwidth,height=0.4\textwidth,keepaspectratio=true]{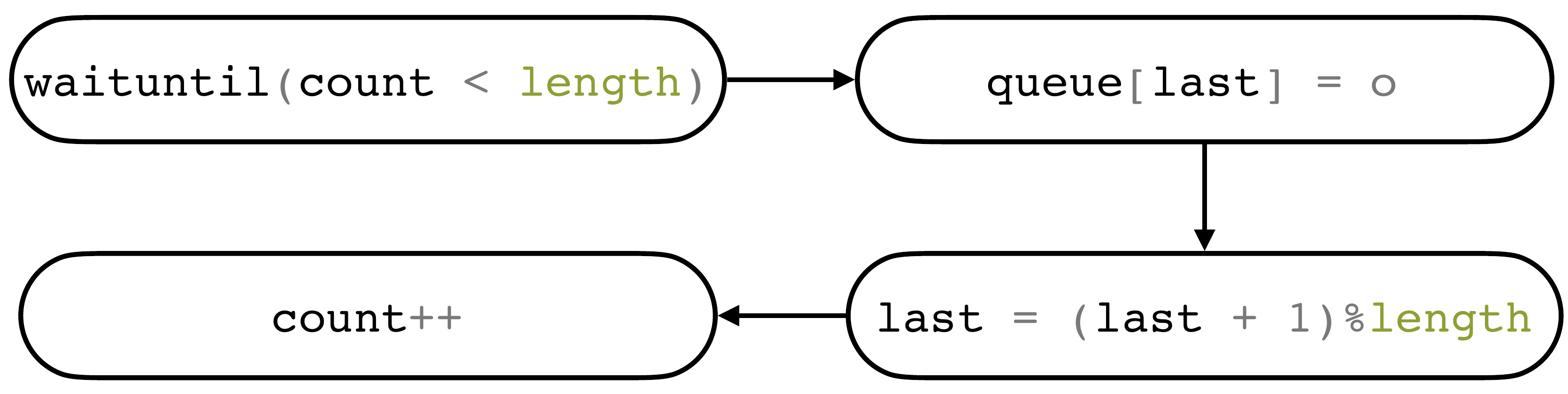}
\caption{FDG for method \code{take}.}\label{fig:fdg-example}
\end{wrapfigure}
Observe that we require the FDG to be acyclic, so some partitions do not give rise to valid FDGs. In the rest of this paper, we assume that partitions obey this restriction so that all cycles are contained within individual nodes of the FDG.  
We also lift this notion of FDG from individual methods to entire monitors in the obvious way (i.e., union of all FDGs).   As we will see in the next section, our  synthesis algorithm operates over FDG representations of monitors.
\end{minipage}

Next, we  state the following NP-completeness result to justify our MaxSAT encoding: 

\begin{theorem}\label{thm:np}
  {\bf{(NP-Completeness)}} Let $\mathcal{G} = (V, E)$ be an FDG
   of  monitor $\mtr$, and let $\Pi \subseteq V \times
  V$ be a set of fragment pairs that can  run in parallel. Then, 
  deciding whether there exists a synchronization protocol with at
  most $k$ locks and that allows all pairs in $\Pi$ to
  run in parallel is NP-Complete.
\end{theorem}

\begin{proof}
By reduction from the edge clique cover
  problem~\cite{edge-clique-cover}.\iffull{ The proof can be found in Appendix~\ref{ssec:npproof}.}\else{ The
  proofs of all theorems can be found in the extended version  submitted as supplemental material.}\fi
\end{proof}

\subsection{Synthesis Algorithm}
\label{sec:synth-alg}

In this section, we describe our core synthesis procedure, which is summarized in Figure~\ref{fig:main-alg}.  At a high level, the {\sc SynthesizeMonitor} algorithm  consists of the following  steps. First, it uses the technique of \citet{expresso} to infer  signaling operations (line~\ref{alg:ln:expresso}). This yields a partially concretized  monitor  $M'$ with signaling operations but no locking. Next, it constructs an FDG representation of the resulting monitor $M'$ as defined in Section~\ref{sec:prelim-np}  (line~\ref{alg:ln:partor}). Third, it infers an \emph{upper bound} $\mathcal{N}_u$ on the maximum number of locks that the synthesized code should use (line~\ref{alg:ln:max-upper}). Then, it  statically analyzes the FDG to infer requirements that the synthesized code needs to obey (line~\ref{alg:ln:analyze}) and uses the results of the previous steps  to generate the MaxSAT encoding (line~\ref{alg:ln:maxsat2}).     Finally, it  instruments $M'$  (line~\ref{alg:ln:instr}) using the  synchronization protocol inferred using MaxSAT. Since the most involved aspects of this algorithm are the MaxSAT encoding and inference of safe interleavings, we defer a detailed discussion of these to the next two subsections and focus on the rest.

\begin{figure}
  \begin{algorithm}[H]
    \begin{algorithmic}[1]
      \Procedure{SynthesizeMonitor}{$\mtr$}
        \State \textbf{input:} $\mtr$: an implicit-synchronization monitor.
        \State \textbf{output:} a semantically equivalent explicit-synchronization monitor.
        \vspace{0.04in}

        \State $\mtr' \gets \textsf{PlaceSignals}(\mtr)$                                       \label{alg:ln:expresso} 
        \Comment{Use technique of Ferles et al. to infer signaling operations}
        \State $\mathcal{G} \gets \mathsf{ConstructFDG}(\mtr')$                                 \label{alg:ln:partor}
                \State $\mathcal{N}_u \gets \textsf{ComputeMaxLocks}(\mathcal{G})$                             \label{alg:ln:max-upper}
        \State $\mathcal{S} \gets \textsc{StaticAnalyze}(\mathcal{G})$                          \label{alg:ln:analyze}

        \State $opt \gets -1$
        \For{$i \in [1, \mathcal{N}_u]$}                                                        \label{alg:ln:loop-start}
          \State $(\mathcal{H}, \mathcal{S}) \gets \textsc{MaxSatEncoding}(\mtr, \mathcal{G}, \mathcal{S}, i)$    \label{alg:ln:maxsat1}
          \State $(p, v, timeout) \gets \mathsf{Solve}(\mathcal{H}, \mathcal{S})$               \label{alg:ln:maxsat2}
          \If{$timeout \lor \left(v \leq opt\right)$} break                                     \label{alg:ln:term-cond}
          \EndIf
          \State $(best, opt) \gets (p, v)$                                                     \label{alg:ln:loop-end}
        \EndFor

        \State \Return $\textsf{Intrument}(best, \mathcal{G}, \mtr')$                           \label{alg:ln:instr}
      \EndProcedure
    \end{algorithmic}
  \end{algorithm}
  \caption{Main Synthesis Algorithm.}
  \label{fig:main-alg}
\end{figure}

\paragraph{\textbf{Iterative exploration of lock count.}} As mentioned above, our synthesis algorithm conceptually reduces the protocol synthesis problem to MaxSAT and uses an off-the-shelf solver to maximize our optimization objective. To achieve this goal, one option is  to generate the MaxSAT encoding based on the maximum possible locks (obtained via the call to \textsf{ComputeMaxLocks}) and then let the solver figure out the optimal number of locks to use.  However, \emph{in practice}, such an approach does not scale  because the size of the encoding increases with respect to the maximum number of locks allowed. That is, for many realistic problems, the MaxSAT solver fails to terminate within a reasonable time limit if we generate the encoding based on the maximum possible locks. 
Thus, instead of directly generating a very large MaxSAT formula up front, our {\sc SynthesizeMonitor} procedure enters a loop (lines \ref{alg:ln:loop-start}--\ref{alg:ln:loop-end}) wherein it gradually increases the maximum number of locks allowed (and hence the size of the MaxSAT encoding). If we get to a point where the MaxSAT solver starts timing out (indicated by boolean variable called \emph{timeout}) or we fail to increase the objective value despite using a larger upper bound on locks (see line~\ref{alg:ln:term-cond}), then the procedure terminates with the best sychronization policy found so far. While this strategy does not guarantee global optimality, it is much more practical than the alternative.

\noindent\begin{minipage}{\textwidth}
\useparinfo
\begin{wrapfigure}{r}{0.4\textwidth}
  \scriptsize
  \begin{minted}[xleftmargin=3em,autogobble,escapeinside=||]{java}
    void put(Object o) {
      waituntil(count < queue.length);
      boolean wasEmpty = count == 0;
      queue[last] = o;
      last = (last + 1) %
      count++;
      broadcast(count == 0, wasEmpty);
    }
  \end{minted}
  \caption{Method put with explicit signals.}
  \label{fig:put-wsig}
\end{wrapfigure}
\paragraph{\textbf{Signaling operations.}}
Our synthesis algorithm uses an auxiliary procedure called \textsf{PlaceSignals}~\cite{expresso} which yields a 
monitor $\mtr'$ that belongs to an intermediate language that is identical
to our source language (Figure~\ref{fig:src-lang}) except that it contains explicit signaling operations. Specifically, this intermediate language contains two additional
signaling directives: (1) \code{signal(p,c)} which notifies a
single thread that is blocked on predicate \code{p}  if condition
\code{c} holds, and (2) \code{broadcast(p,c)} which notifies \emph{all}
 threads blocked on \code{p} if \code{c} holds.  Figure~\ref{fig:put-wsig} shows the result of calling \textsf{PlaceSignals} on the \code{put}  
procedure from Figure~\ref{fig:motiv-src}.
\end{minipage}

\paragraph{\textbf{ FDG construction}} Recall that an FDG is a generalized version of a control-flow graph where nodes are code fragments  rather than basic blocks, and each code fragment is a unit of computation that our algorithm should assign locks to. Since there can be many ways to partition a given method into code fragments, the \textsf{ConstructFDG} procedure invoked at line~\ref{alg:ln:partor} of Figure~\ref{fig:main-alg} implements a particular heuristic for partitioning a method into  code fragments. In particular, the more the number of code fragments, the more the parallelization opportunities; thus, our \textsf{ConstructFDG} procedure tries to maximize the number of code fragments while maintaining the invariant that the FDG is acyclic and that each code fragment must have a unique entry point (see Section~\ref{sec:impl}).

\paragraph{\textbf{Computing upper bound on locks}} Because the MaxSAT encoding assumes a fixed number of locks, our synthesis algorithm calls the \textsf{ComputeMaxLocks} procedure at line~\ref{alg:ln:max-upper} to compute an upper bound on the number of locks needed. Given an FDG $\mathcal {G} = (V, G)$,  the key idea behind this procedure is to construct a so-called \emph{conflict graph}  $G_C = (V, E_C)$ where  $(f, f')$ is in $E_C$ iff fragments $f$ and $f'$ have a data race.  Since it can be shown that the optimal solution to our problem is an \emph{edge clique cover}~\cite{edge-clique-cover} of this conflict graph (\iffull{see Appendix}\else{see extended version}\fi), we can use known theorems (e.g., Mantel's theorem, ~\citet{Alon-1986} etc.) to obtain an upper bound on the number of locks needed without having to solve an NP-complete problem.\footnote{In our implementation, we use multiple upper bounds using known theorems and return the best one.}

\paragraph{\textbf{Static analysis.}} Recall from Section~\ref{sec:overview} that the  constraints in our  MaxSAT encoding utilize information obtained via static analysis. Thus, line~\ref{alg:ln:analyze} of Figure~\ref{fig:main-alg} statically analyzes the input monitor to obtain  the following three pieces of information:
\begin{itemize}[leftmargin=*]
    \item {\bf Atomic fields $\mathcal{F}$:} One of the goals of the analysis is to infer  a set of fields that could \emph{potentially} be implemented using \code{Atomic} types.  Thus, our static analysis checks whether (a) a field of type \code{T} has a corresponding \code{AtomicT} version, and (b) whether all updates to this field can be implemented using one of the methods provided by \code{AtomicT}.
    \item {\bf Data races $\mathcal{R}$:} The second goal of our static analysis is to identify pairs of fragments that would have a data race if they do not use a shared lock. Thus, given a pair of fragments $(f, f')$, our static analysis checks whether $f$ writes to a memory location $l$ that is accessed in $f'$.
    \item {\bf Interleaving opportunities $\mathcal{I}$:} Finally, a third key goal of the analysis is to identify safe interleaving opportunities between fragments. Since this aspect of the analysis is quite involved, we discuss it in the next subsection.
    
\end{itemize}

\paragraph{\textbf{MaxSAT encoding}} As mentioned in Section~\ref{sec:overview}, our MaxSAT encoding uses two types of boolean variables, namely (1) $h_{f_i}^{l_j}$ indicating that fragment $f_i$ must hold lock $l_j$ and (2) $a_{f}$ indicating that field $f$ should be converted to atomic.  Hence, a model of the MaxSAT problem can be easily converted to a so-called \emph{locking protocol} $(\lockmap, \atomfld, \predmap)$, where $\lockmap$ is an assignment from fragments to a set of locks, $\atomfld$ is a set of fields that should be implemented using atomic types, and $\predmap$ is a mapping from \code{waituntil} guards to locks. In particular, we have $l_j \in \lockmap(f_i)$ if and only if $h_{f_i}^{l_j}$ is assigned to true in the model returned by the MaxSAT solver, and we have $\emph{fld} \in \atomfld$ if $a_\emph{fld}$ is assigned to true. Due to the constraints in our MaxSAT encoding, it is similarly easy to derive $\predmap$:  because our encoding ensures that every occurrence of a  \code{waituntil(p)} statement is protected by the \emph{same} set of locks $S$, we associate one of the locks $l$ in $S$ with the condition variable introduced for predicate \code{p}.\footnote{Specifically, when choosing which lock $l$ in set $S$ to designate as the representative, we choose the smallest lock in $S$ according to the total order. Because all locks held by a thread must be released before it blocks on a condition variable and must be acquired after it gets notified (with method \code{{\footnotesize await}} releasing and acquiring $l$ internally), choosing the smallest lock prevents deadlocks.}

\paragraph{\textbf{Instrumentation.}}
The last step of our algorithm is to synthesize the 
explicit-synchronization monitor via the
 $\textsf{Instrument}$ procedure invoked at line~\ref{alg:ln:instr}. Given a synchronization protocol $(\lockmap,\atomfld,\predmap)$, the \textsf{Instrument} procedure performs the following steps:
\begin{enumerate}[leftmargin=*]
 \item First, it introduces all the synchronization fields (locks, condition variables, and atomic fields) that appear in the protocol.
 \item It converts every update to an atomic field to the corresponding atomic \code{update} statement.
 \item Finally, it introduces all the necessary locking and signaling operations to implement the synthesized synchronization protocol.
\end{enumerate}

\iffull{
We refer the interested reader to Appendix~\ref{appendix:mtr-instr} for more details on the instrumentation. 
}
\else{
We refer the interested reader to the extended version of the paper (under supplemental material) for more details on the instrumentation. 
}
\fi

\begin{theorem}\label{thm:top-level-correctness}\iffull{}\else{\footnote{All proofs of this Section are in the extended version of the paper.}}\fi
  {\bf{(Correctness)}} Given an implicit-synchronization monitor $\mtr$ in the language of Figure~\ref{fig:src-lang},  if
  $\textsc{SynthesizeMonitor}(\mtr)$ returns $\mtr'$, then we have $\mtr \sim \mtr'$.
\end{theorem}

\iffull{
\begin{proof}
Can be found in Appendix~\ref{sec:correct-proofs}.
\end{proof}
}
\fi

\subsection{ Analysis to Identify Safe Interleavings}
\label{sec:analysis}

We now describe how to infer safe
interleaving opportunities between threads while ensuring that monitor
operations \emph{appear} to take place atomically. Given a fragment
dependency graph $\mathcal{G} = (V, E)$ for a monitor $\mtr$, an
\emph{interleaving opportunity} (or \emph{interleaving} for short) is a pair $(v, e)$ where $v \in V$ is
a code fragment of $M$ and $e = (v_1, v_2) \in E$ is an edge of the FDG. Intuitively, such an interleaving is safe if some thread can execute $v$ \emph{in between} some other thread's execution 
of $v_1$ and $v_2$ without violating atomicity. The goal of our static analysis is to identify a set $\mathcal{I}$ of such
{safe interleavings}. In what follows,
we  formalize \emph{safe interleavings} and
 describe an analysis for identifying them.

\paragraph{\textbf{Formalizing safe interleavings}}
To formalize the notion of safe interleaving, we need to
keep track of which fragments of the monitor were executed in what order. For this purpose, given an FDG $\mathcal{G} = (V,E)$ of  $\mtr$, we define a \emph{fragmented monitor} $\mtr_{\mathcal{G}}$ to be the same as $\mtr$ except that every fragment in $\mathcal{G}$ is placed in its own method. Observe that histories of $\mtr_{\fdg}$ encode all possible interleavings of {fragments} in $\fdg$. In this sense, histories of $\mtr_{\fdg}$ are similar to explicit monitor histories but are slightly higher level in that they allow interleavings between fragments rather than atomic statements. Thus, we adapt the same notions of \emph{sequential}, \emph{well-formed}, and \emph{interleaved} histories from Section~\ref{sec:target-lang} to fragmented monitors, as illustrated by the following examples.

\begin{example}\label{ex:fragment}
Given monitor $\mtr$ from Figure~\ref{fig:motiv-src}, its fragmented version $\mtr_{\fdg}$ splits \code{put} and \code{take} into four different methods, each named $\code{put}_i$ and $\code{take}_i$. 
Given history $\tis = (take, t)$ and  initial state $\pstate$ with a non-empty queue, we have
\[  
  \concr{\mtr_{\fdg}}(\tis, \args, \pstate) = ((\emph{take}_1, t) (\emph{take}_2, t) (\emph{take}_3, t) (\emph{take}_4, t), \args')
\]
where  $\emph{take}_1, \ldots, \emph{take}_4$ denote fragments 5-8 in Figure~\ref{fig:motiv-src} and $\args$,$\args'$ are empty argument mappings. 

\end{example}

\begin{example}
In the example above, $\concr{\mtr_{\fdg}}(\tis, \args, \pstate)$  is both sequential and well-formed. However,  $(\emph{take}_1, t),  (\emph{take}_2, t)$ is not well-formed because it does not involve all four methods, and $(\emph{take}_1, t), (\emph{take}_3, t), (\emph{take}_2, t), (\emph{take}_4, t)$ is also not well-formed because it executes $\emph{take}_3$ before $\emph{take}_2$.  Finally, the following history is an interleaved (and, by definition, well-formed) history where threads $t$ and $t'$ execute method \code{take} and \code{put} respectively:
 \begin{equation}\label{h:ileave}
 \begin{array}{l}
 \tis_{\fdg} = (\emph{put}_1, t)(\emph{put}_2, t)(\emph{put}_3, t)(\emph{take}_1, t')(\emph{put}_4, t)(\emph{take}_2, t')(\emph{take}_3, t')(\emph{take}_4, t') 
 \end{array}
 \end{equation}
 Furthermore, for this history we have $(\tis_{\fdg}, \args_{\fdg}) \backsim ((put, t) (take, t'), \args)$ for some $\args_{\fdg}$ and $\args$. That is, $\tis_{\fdg}$ \emph{simulates} a history of $\mtr$ where thread $t$ executes method \code{put} and  $t'$ executes \code{take}.
\end{example}

\begin{definition}{\bf (Interleaving)}
Given an FDG $\fdg = (V, E)$ for monitor $\mtr$, an  \emph{interleaving} is a pair $(v, e)$ where $v \in V $ and $e \in E$. Furthermore, given a history $h$ of fragmented monitor $\mtr_\fdg$, we write $\mathcal{X}(\tis_{\fdg})$ to denote the set of all interleavings that occur in $h$.
\end{definition}

\begin{example}
  For the history $\tis_\fdg$ from Eq.~\ref{h:ileave}, we have:
  \[\mathcal{X}(\tis_{\fdg}) = \{ (\emph{take}_1, (\emph{put}_3, \emph{put}_4)), (\emph{put}_4, (\emph{take}_1, \emph{take}_2))\}\] This is the case because this history executes $\emph{take}_1$ in between consecutive fragments $\emph{put}_3$ and $\emph{put}_4$ of some other thread. Similarly, we have $(\emph{put}_4, (\emph{take}_1, \emph{take}_2)) \in \mathcal{\chi}(\tis_\fdg)$ because it executes $\emph{put}_4$ in between  $\emph{take}_1$ and $\emph{take}_2$.
  \end{example}

\begin{definition}{\bf (Safe interleavings).}\label{def:safe-intl} Let $\mathcal{G}$ be an FDG of monitor $\mtr$. We say that a set of interleavings $S$ is \emph{safe}, if for every input state $\sigma$ and every interleaved history $\tis_{\fdg}$ of $\mtr_{\fdg}$ we have:
\[
\textrm{If\ } \mathcal{X}(\tis_{\fdg}) \subseteq S \textrm{\ and\ } \mtr_{\fdg} \vdash (\tis_{\fdg}, \args_{\fdg}, \pstate) \Downarrow \pstate' \textrm{\ then\ } \exists\tis,\args.\ (\tis_{\fdg},\args_{\fdg}) \backsim (\tis,\args) \textrm{\ and\ } \mtr \vdash (\tis, \args, \pstate) \Downarrow \pstate' 
\]
\end{definition}

In other words, a set of interleavings $S$ is safe if for every interleaved history of $\tis_{\fdg}$ whose interleavings are a subset of $S$ we can prove that $\tis_\fdg$ leads to the same final state as \emph{some} history $\tis$ of $M$ where $\tis$ simulates $\tis_\fdg$.
This definition essentially lifts the second correctness criterion of Definition~\ref{def:correctness} to a fragmented monitor.

\paragraph{\textbf{Inferring Safe Interleavings.}}
We now turn our attention to the problem of \emph{inferring} safe interleavings. Given a monitor $\mtr$ and its  FDG
$\mathcal{G} = (V,E)$, our goal is to find a set $\mathcal{I}
\subseteq V \times E$ such that all interleavings in $\mathcal{I}$ are
safe. However, a key challenge is that the space of all safe interleavings
is exponential (i.e., the power set of $V \times E$), so, even
if we had a procedure for checking whether some set $\mathcal{I}$ is safe,
enumerating all candidates would be computationally intractable.

To overcome this challenge, we introduce the notion of \emph{strong safety} that allows us to build $\mathcal{I}$ iteratively. In particular, note that if $S_1$ and $S_2$ are both safe interleaving sets according to Definition~\ref{def:safe-intl}, it may \emph{not} be the case that $S_1 \cup S_2$ is also a safe interleaving. However, to build $\mathcal{I}$ incrementally, we need a notion of safe interleaving that is closed under union. For this purpose, we introduce a notion of \emph{strong safety} for a single interleaving $(v, e)$. Since strongly safe interleavings enjoy the property of being closed under union, this notion lends itself to a computationally feasible technique for computing safe interleaving sets. In the remainder of this section, we define strong safety and present our static analysis for computing safe interleaving sets. 
Towards this goal, we first introduce the notions of \emph{left} and \emph{right commutativity} for our context:

\begin{definition}{{\bf (Left/Right Commutativity).}}\label{def:commute} Given fragments $v$ and $v'$, we say that $v$ \emph{left commutes} with $v'$, denoted \emph{LeftCommute($v$, $v'$)}, iff, whenever 
$\mtr_{\fdg} \vdash \left(
(v', t')(v, t), \args, \sigma
\right)\Downarrow \pstate'$ holds, so does $\mtr_{\fdg} \vdash \left(
(v, t)(v', t'), \args, \sigma
\right)\Downarrow \pstate'$.
Conversely,  $v$ \emph{right commutes} with $v'$, denoted \emph{RightCommute($v$, $v'$)}, iff 
  $\mtr_{\fdg} \vdash\left(
  (v, t)(v', t'),\args,\sigma
  \right) \Downarrow \pstate'$ implies $\mtr_{\fdg} \vdash \left(
  (v', t')(v, t), \args, \sigma
  \right) \Downarrow \pstate'$.
\end{definition}
 In other words, a fragment $v$ left commutes with $v'$ if, whenever $v$ executes just after $v'$, the resulting  state is the same as if $v$ had executed just before $v'$.  For example, $f_4$ (i.e. \code{count++}) in
Figure~\ref{fig:motiv-src} left-commutes with $f_5$ since increasing \code{count} right after \code{waituntil(count > 0)}  is equivalent to increasing \code{count} just before \code{waituntil(count >0)}. That is, assuming that \code{waituntil(count>0)} was not blocked before executing \code{count++}, then it will still not be blocked \emph{after} executing \code{count++}. However, $f_4$ does not left-commute with $f_1$: when \code{count} equals \code{queue.length - 1}, incrementing \code{count} just after \wuntil(\code{count < queue.length})  is \emph{not} equivalent to incrementing \code{queue.length} before the \wuntil statement. That is, if \code{waituntil(count < queue.length)} did not block before executing \code{count++}, we cannot guarantee that it also does not block after executing \code{count++}.
  
Next, we use this notion of left and right commutativity to define strong safety:
  
\begin{definition}{\bf (Strong safety).}\label{def:strong-safe}  Let $\mathcal{G} = (V, E)$ be an FDG for monitor $\mtr$, and let $E^{*}$ denote the reflexive transitive closure of $E$. We say that an interleaving $(v, e)$, where $e = (v_s,v_t)$, is \emph{strongly safe} if the following conditions are satisfied:
  \begin{enumerate}
        \item $\forall v^-. (v^-, v_s) \in E^* \Longrightarrow \emph{LeftCommute}(v, v^-)$
  \item $\forall v^+. (v_t, v^+) \in E^* \Longrightarrow \emph{RightCommute}(v, v^+)$
  \end{enumerate}
  \end{definition}
That is, an interleaving $(v,e)$  is said to be \emph{strongly safe} if we can prove that fragment $v$ left commutes with \emph{every} possible predecessor of $v_s$ and that it right commutes with \emph{every} possible successor of $v_t$. To see why these conditions imply safety, recall that  a set of interleavings $S$ is safe if, for any  history $\tis_{\fdg}$ whose interleavings are a subset of $S$, we can find some (sequential) history of the original monitor that simulates $\tis_{\fdg}$. Assuming $S$ contains only strongly safe interleavings, we can create such a sequential history by ``removing'' interleavings one at a time from $\tis_{\fdg}$. For instance, let $\chi = (v, (v_s,v_t)) \in S$ be an interleaving that occurs in $\tis_{\fdg}$. Since $\chi$ is strongly safe, we can always obtain an equivalent history $\tis'_{\fdg}$ that has strictly less interleavings than $\tis_{\fdg}$ by commuting $v$ past either every successor of $v_t$ or every predecessor of $v_s$ that appears in $\tis_{\fdg}$.

\begin{example}
  For the monitor from Figure~\ref{fig:motiv-src}, we can show that every interleaving $(v, e)$ where $v$ belongs to method \code{take} and edge $e$ belongs to method \code{put} (and vice versa) is strongly safe. However, none of the interleavings where $v$ and $e$ belong to the same method are  strongly safe. Finally,  because both of the interleavings of the history $\tis_{\fdg}$ from Eq.~\ref{h:ileave} are strongly safe, we can derive a sequential history that simulates history $\tis_{\fdg}$ by swapping $ (\emph{take}_1, t')$ with $(\emph{put}_4, t)$. %
\end{example}

We now state a key theorem that underlies the correctness of our approach:

\begin{theorem}\label{thm:safe}
  Let $\mathcal{G}$ be an FDG and let $\chi_1, \ldots, \chi_n$ be \emph{strongly safe} interleavings.  Then,  $\{ \chi_1, \ldots, \chi_n\}$ satisfies Definition~\ref{def:safe-intl} (i.e., is a safe interleaving set for $\mathcal{G}$). 
\end{theorem}

\iffull{
\begin{proof}
Can be found in Appendix~\ref{sec:correct-proofs}.
\end{proof}
}
\fi

\begin{figure}
  \begin{algorithm}[H]
    \begin{algorithmic}[1]
      \Procedure{FindSafeInterleavings}{$\mathcal{G}$}
        \State \textbf{input:} An FDG representation $\mathcal{G} = (V, E)$ of monitor $\mtr$
        \State \textbf{output:} A set $\mathcal{I}$ of all safe interleavings
        \vspace{0.04in}
        \State $\mathcal{I} \gets \emptyset$
        \For{$v \in V,\  e = (v_s, v_t) \in E$}
            \vspace{0.05in}
            \State $V_s^{*} \gets \left\{\ v' \mid (v',v_s) \in E^{*}\ \right\}$ \Comment{All predecessor vertices that reach $v_s$.}
            \vspace{0.05in}
            \State $V_t^{*} \gets \left\{\ v' \mid (v_t, v') \in E^{*}\ \right\}$ \Comment{All successor vertices of $v_t$.}
            \vspace{0.05in}
            \If{$
            \left(
                \forall v_s^* \in V_s^{*} .\ \textsc{LeftCommute}(v, v_s^*)
            \right)
            \land \left(
                \forall v_t^* \in V_t^{*} .\ \textsc{LeftCommute}(v_t^*,v)
            \right)
            $}
            \vspace{0.05in}
                \State $\mathcal{I} \gets \mathcal{I} \cup \{ (v, e) \}$
            \EndIf
        \EndFor
        \State \Return $\mathcal{I}$
      \EndProcedure
      \vspace{1em}
      \Function{LeftCommute}{$v, v'$}
        \State \textbf{input:} Two fragments $v$, $v'$
        \State \textbf{output:} true iff $v$ left commutes with $v'$
        \vspace{0.04in}
        \State $ X\gets \{x \mid x \text{ is a variable in }v \text{ or }v'\}$.
        \State $X_L \gets \{x_l \text{ fresh name} \mid x \in X\}\ \ \ \ \ \ \ \ \ \ \ \ \ \ \ \ $  $ X_R \gets \{x_r \text{ fresh name} \mid x \in X\}$
        \State $S_L \gets (v'; v)[\code{assume}/\wuntil,  X_L /  X] \ \ \ $
         $S_R \gets (v; v')[\code{assert}/\wuntil,  X_R / X]$
        \State \Return \textsf{Verify}$(\{ X_L =  X_R\} \,S_L; S_R \,\{ X_L =  X_R\})$
      \EndFunction
    \end{algorithmic}
  \end{algorithm}
  \caption{Algorithm to find all safe interleavings.}
  \label{fig:algo-safe-intl}
\end{figure}

\paragraph{\textbf{Static analysis algorithm.}} Finally, we conclude this section by presenting our static analysis algorithm (shown in Figure~\ref{fig:algo-safe-intl}) for computing a set $\mathcal{I}$ of safe interleavings.  At a high level, this algorithm identifies which $(v, e)$ pairs are strongly safe and then returns their union, which by Theorem~\ref{thm:safe}, corresponds to a safe interleaving set. To check whether an interleaving $(v, e)$ (for $e=(v_s, v_t)$) is strongly safe,  we must check if $v$ left commutes with each predecessor of $v_s$ and right commutes with each successor of $v_t$. As shown in the {\sc LeftCommute} procedure, we reduce the verification of left commutativity to the problem of verifying a Hoare triple. In particular, given fragments $v, v'$, we generate a code snippet $S_L; S_R$ where (1) $S_L$ is an alpha-renamed version of $v'; v$ with \code{\wuntil}'s replaced by $\code{assume}$ statements, and (2) $S_R$ is an alpha-renamed version of $ v; v'$ with \code{\wuntil}'s replaced by $\code{assert}$ statements. Note that we turn $\code{\wuntil}$'s in $S_L$ into \code{assume}s because the definition of left commutativity assumes that $v'; v$ has terminated. On the other hand, we need to \emph{show} that $S_R$ does \emph{not} block; thus, we assert that the predicates in the \code{\wuntil} statement evaluate to true under the assumption that they also evaluate to true in $S_L$. Finally, in addition to showing that \code{waituntil}'s are not blocked, we also need to establish that the monitor state is the same in $S_L$ and $S_R$. Thus, the Hoare triple we construct checks that the values of variables are the same at the end, assuming that they are the same in the beginning. Note that the implementation of right commutativity is the same with $v$ and $v'$ swapped; thus, {\sc RightCommute}$(v,v')$ can be checked by directly calling {\sc LeftCommute}$(v',v)$.

\subsection{MaxSAT Encoding}
\label{sec:max-sat}

\begin{figure}
  \small
  \[
  \begin{array}{cccc}
    \fbox{\textsc{{\footnotesize Race-1}}} & \multicolumn{3}{c}{
      \irule{
        IsFrag(v_1)  \irulespace IsFrag(v_2) \irulespace
        \irulespace Races = \mathcal{R}(v_1, v_2) \irulespace Races \subseteq \mathcal{F} \irulespace
        Races = \{\ \code{this.f} \ \}
    }{
      Mutex(\{v_1,v_2\}, \mathcal{N}) \lor \toatom_{f} \in \mathcal{H} %
    }}\\\\
    \fbox{\textsc{{\footnotesize Race-2}}} &
    \multicolumn{3}{c}{
    \irule{
        IsFrag(v_1) \irulespace IsFrag(v_2) \irulespace
        Races = \mathcal{R}(v_1, v_2)  \irulespace  Races \neq \emptyset \irulespace
        (|Races| > 1 \lor Races \nsubseteq \mathcal{F})
    }{
      Mutex(\{v_1,v_2\}, \mathcal{N}) \in \mathcal{H} %
    }}
    \\\\
    \fbox{\textsc{{\footnotesize I-Leave}}} &
    \multicolumn{3}{c}{
    \irule{
        IsFrag(v) \irulespace IsEdge(e) \irulespace e = (v_s,v_t) \irulespace \neg \emph{SafeInterleaving}(v, e)
    }{
      Mutex(\{v,v_{s},v_{t}\}, \mathcal{N}) \in \mathcal{H} %
    }}\\\\
    \fbox{\textsc{{\footnotesize Wait}}} &
    \irule{
      \begin{array}{c}
        p \in Preds(\mtr) \\ F = \left \{\ f \mid IsFrag(f),\ f \equiv \code{waituntil(p)} \right \}
      \end{array}
    }{
      Mutex(F, \mathcal{N}) \in \mathcal{H} \irulespace \bigwedge\limits_{i = 1}^{\mathcal{N}}\bigwedge_{v_1, v_2 \in F} \left ( \holds_{v_1}^{l^i} \leftrightarrow \holds_{v_2}^{l^i} \right ) \in \mathcal{H}
    } &
    \fbox{\textsc{{\footnotesize L-Order}}} &
    \irule{
      \begin{array}{c}
        IsEdge(e)
      \end{array}
    }{
      LockOrder(e, \mathcal{N}) \in \mathcal{H} %
    }\\\\
    \fbox{\textsc{{\footnotesize Min-Lock}}} &
    \irule{
      \begin{array}{c}
        m \in Methods(\mtr)\\
        MF = \{\ v \mid IsFrag(v),\ Method(v) = m \ \}
      \end{array}
    }{
      \bigcup\limits_{i = 1}^{\mathcal{N}} \{ \bigwedge\limits_{f \in MF} \neg \holds_f^{l_i} \} \subseteq \mathcal{S} %
    } &
    \fbox{\textsc{{\footnotesize Min-Atom}}} &
    \irule{
      \code{this.fld} \in \mathcal{F}
    }{
      \neg \toatom_{fld} \in \mathcal{S} %
    }\\\\ %
    \fbox{\textsc{{\footnotesize Max-Par}}} &
    \multicolumn{3}{c}{
    \irule{
        IsFrag(v_1) \irulespace IsFrag(v_2) \irulespace
        \mathcal{R}(v_1, v_2) = \emptyset
    }{
      \neg Mutex(\{v_1, v_2\}, \mathcal{N}) \in \mathcal{S} %
    }}\\\\
    \hline
    \\
    \fbox{\textsc{{\footnotesize Aux-Defs}}} & \multicolumn{3}{c}{Mutex(F, \mathcal{N}) = \bigvee\limits_{i=1}^{\mathcal{N}}\bigwedge\limits_{f \in F} \holds_f^{l_i} \irulespace
      LockOrder((v_{s}, v_{t}), \mathcal{N}) =
      \bigwedge\limits_{1 \leq \ell < u \leq \mathcal{N}}
      \lnot\left(
      \holds_{v_{s}}^u
      \land \holds_{v_{t}}^u\land
      \lnot \holds_{v_{s}}^\ell
      \land \holds_{v_{t}}^\ell
      \right)}
    \\
  \end{array}
  \]
  \caption{Inference rules for \textsc{MaxSatEncoding}($\mtr,
    \mathcal{G}, \mathcal{S}, \mathcal{N}$) procedure. $\mathcal{G} =
    (V,E)$ is an FDG of monitor $\mtr$, $\mathcal{S} = (\mathcal{F},
    \mathcal{R}, \mathcal{I})$ are the results of the static analysis,
    and $\mathcal{N}$ is an upper bound on the number of
    locks. Predicate $IsFrag(v)$ is true if $v \in V$, $IsEdge(e)$ if
    $e \in E$, and $SafeInterleaving(v, e)$ if $(v, e) \in
    \mathcal{I}$. Relations $Methods(\mtr)$ and $Preds(\mtr)$ return
    all methods of monitor $\mtr$ and all predicates that appear as an
    argument of a $\wuntil$ statement in $\mtr$ respectively.}
  \label{fig:maxsat-alg}
\end{figure}

In this section, we describe our MaxSAT encoding which is formalized as inference rules in Figure~\ref{fig:maxsat-alg}. Recall that the encoding procedure takes as input (a) an FDG representation of the monitor, (b) the results of the static analysis, and (c) an upper bound on the maximum number of locks, and it produces a set of hard constraints $\mathcal{H}$ and a set of soft constraints $\mathcal{S}$. In the remainder of this section, we describe the inference rules in Figure~\ref{fig:maxsat-alg} for generating these constraints in more detail. 

\paragraph{\textbf{Variables.}} Our MaxSAT encoding uses two types of variables. First, we introduce variables of the form $h_{v_i}^{l_j}$ indicating that fragment $v_i$ needs to hold lock $l_j$. Thus, given an FDG with $n$ vertices and an upper bound $\mathcal{N}$ on the number of locks, our encoding contains $n \times \mathcal{N}$ such variables.  The second type of variable used in our encoding is of the form $a_\emph{fld}$ indicating that $\emph{fld}$ should be implemented using an atomic type. 

\paragraph{\textbf{Mutex encoding.}} Given a set of fragments $F$ and an upper bound $\mathcal{N}$ on the number of locks, we often need to enforce that all fragments in $F$ share at least one of the  $\mathcal{N}$ possible locks. We write \emph{Mutex}($F, \mathcal{N}$) to denote this requirement. In particular, as shown at the bottom of Figure~\ref{fig:maxsat-alg}, this is defined as $Mutex(F, \mathcal{N}) = \bigvee\limits_{i=1}^{\mathcal{N}}\bigwedge\limits_{f \in F} \holds_f^{l_i}$.

\paragraph{\textbf{Hard constraints.}}
Next, we describe the hard constraints generated by our MaxSAT encoding. These hard constraints $\mathcal{H}$ correspond to  correctness requirements on the synthesized protocol  and include (1) data race freedom, (2) correct signaling and deadlock freedom and (3) atomicity. Specifically, the first two rules in Figure~\ref{fig:maxsat-alg} deal with data race freedom, the next rule deals with atomicity, and the last two rules deal with deadlock freedom and correct signaling.

\paragraph{\textsc{Race-1}} The first rule, labeled {\sc Race-1},  deals with data race freedom of two fragments that 
have a data race on a \emph{single} monitor field \code{f}. The premises of this rule stipulate that $v_1, v_2$ are fragments that race \emph{only} on field $f$ which can be converted to atomic (i.e., \code{this.f} $\in \mathcal{F}$). In this case, we prevent data races by either (1) enforcing that $v_1, v_2$ share a lock (the \emph{Mutex} constraint) or (2)  ensuring that field $f$ is converted to an atomic field.

\paragraph{\textsc{Race-2}}
The next {\sc Race-2} rule prevents data races between fragments where the data race cannot be resolved by making one of the fields atomic. In particular, given two fragments $v_1, v_2$ that have a data race, this rule simply enforces that they share a common lock via the \emph{Mutex} function.

\paragraph{\textsc{I-Leave}}
The next rule generates constraints to ensure that monitor operations appear to take place atomically. In particular, if the static analysis cannot prove $(v, e)$ to be a strongly safe interleaving (recall Definition~\ref{def:strong-safe}), then we need to ensure that a thread cannot execute $v$ when some other thread is executing $e$. To do so, we ensure that $v, v_s, v_t$ all share a common lock by generating a \emph{Mutex} hard constraint for these three fragments.

\paragraph{\textsc{L-Order}}
The next rule, labeled \textsc{L-Order}, ensures that the 
resulting synchronization protocol is deadlock-free. Specifically, for
every edge $e = (v_s,v_t)$ in the input FDG, this rule generates a
hard constraint, via $LockOrder(e,\mathcal{N})$ (defined at the bottom of Figure~\ref{fig:maxsat-alg}), that ensures
that every lock acquisition respects the total order on locks. In
particular, for every pair of locks $l,u$ such that $l \prec u$,
$LockOrder(e, \mathcal{N})$ prevents the  synchronization protocol from violating the global order on locks. Recall that  a protocol violates this global order if it acquires the ``smaller" lock $l$  between $v_s$ and $v_t$ while both $v_s$ and $v_t$ hold lock $u$. Thus, the hard constraint generated by $LockOrder(e, \mathcal{N})$ prevents this from happening. 

\begin{example}
  Assuming $\mathcal{N} = 2$, this rule generates $\neg \left( \holds_{v_s}^{l_2}
  \land \holds_{v_t}^{l_2} \land \neg \holds_{v_s}^{l_1} \land
  \holds_{v_t}^{l_1}\right)$ for edge $(v_s,v_t)$.
\end{example}

\paragraph{\textsc{Wait}} The last hard constraint rule, called \textsc{Wait},  is used for associating a single lock with each condition variable. In particular, since all fragments
of the form $\wuntil(p)$ must hold the same set of locks, this rule generates two hard constraints for every $\wuntil$ predicate
$p$ of the input monitor: (1) a mutex constraint for all $\wuntil(p)$
fragments and (2) a constraint that enforces that all $\wuntil(p)$ fragments
must share \emph{all} common locks.

\paragraph{\textbf{Soft Constraints.}} As discussed earlier, our goal is to generate a synchronization protocol that is not only correct-by-construction but one that also results  in efficient code. Hence, as a proxy metric for efficiency, we want to (1) minimize the number of locks and atomic fields that are introduced, and (2) maximize the number of fragments that can run in parallel. The remaining three rules in Figure~\ref{fig:maxsat-alg} introduce soft constraints to encode this optimization objective.

\paragraph{\textsc{Min-Lock}} The rule labeled \textsc{Min-Lock} is used for minimizing the number of locks. However, instead of simply minimizing the total number of locks used by the protocol, the soft constraints generated by this rule minimize the number of locks used \emph{per method}. Even though this is not equivalent to minimizing the number of locks used by the entire protocol, we have found this approach to synthesize protocols with a more even distribution of locks among the monitor methods. In practice, such protocols are more desirable because they avoid scenarios where a subset of the methods incur a higher synchronization cost than others. Specifically,  this rule generates a  soft constraint for every lock $l \in \{l_1 ... l_{\mathcal{N}}\}$ and every method $m$ of $\mtr$ and asserts that none of the fragments in $m$  hold lock $l$.

\paragraph{\textsc{Min-Atom}} The {\sc Min-Atom} rule generates soft constrains to minimize the number of fields that are made atomic by  asserting that $\toatom_{fld}$ is assigned to false.

\paragraph{\textsc{Max-Par}} The last rule called {\sc Max-Par} generates  soft constraints to maximize parallelism. Specifically, for every pair of fragments $(v, v')$ that do not have  data races, we add a soft constraint stating that $v$ and $v'$  \emph{do not share} any locks.

We conclude this Section with a theorem that states the correctness of our MaxSAT encoding.

\begin{theorem}\label{thm:maxsat}
  Let $m$ be a model of the generated MaxSAT instance and $(\lockmap,\atomfld,\predmap)$ be the synchronization protocol constructed as follows:
  \begin{gather*}
     {\small\lockmap = \left\{ v \mapsto \left\{ l \mid m[h_v^l]\right\}\right\}\ \atomfld = \left\{ \code{fld} \mid m[\toatom_{fld}] \right\}}\ 
     {\small\predmap = \left\{p \mapsto  l_i \mid IsWait(v,p), i = min(\{j \mid m[\holds_v^{l_j}]\}) \right\}}
  \end{gather*}
  where, $IsWait(v,p)$ is true if v is a \wuntil statement on $p$. Then, $(\lockmap,\atomfld,\predmap)$ is a correct synchronization protocol.

\end{theorem}

\iffull{
\begin{proof}
Can be found in Appendix~\ref{sec:correct-proofs}.
\end{proof}
}
\fi

%% file: implementation.tex
\section{Implementation}
\label{sec:impl}

We have implemented our approach in a tool called \toolname that emits explicit-synchronization monitors in Java. \toolname is based on the Soot program analysis infrastructure~\cite{soot} and the Z3 SMT solver~\cite{z3}. In particular, we use Soot to perform  various kinds of static analyses needed by our method (e.g., pointer analysis) and to translate the input monitor to an explicit-synchronization monitor in Java. Furthermore, we leverage Z3 for  solving MaxSAT instances and discharging the validity queries that arise when checking commutativity between fragments. In the remainder of this section, we discuss several design choices and optimizations that were not discussed previously. 

\paragraph{\textbf{Weights of soft constraints.}} As expected, the quality of the synthesized protocol  depends on the model returned by the MaxSAT solver. In practice, we have observed certain types of soft constraints to be more important than others for efficiency. Thus, our implementation assigns different weights for different  classes of soft constraints. For instance, because it is always preferable to use an atomic field instead of a lock, \toolname assigns a higher weight to soft constraints generated by rule \textsc{Min-Atom} from Figure~\ref{fig:maxsat-alg} than the ones generated by rule \textsc{Min-Lock}.\footnote{{An ablation study that demonstrates the need for adjusting the weights of soft constraints can be found in \iffull{Appendix~\ref{sec:weight-ablation}}\else{the extended version of the paper}\fi.}}

\paragraph{\textbf{Constructing FDGs.}} As mentioned in Section~\ref{sec:algorithm}, \toolname uses a heuristic to partition the input CFG into fragments. The goal of this heuristic is to maximize parallelization opportunities while ensuring that the partition results in a valid FDG according to Definition~\ref{def:fdg}. Our heuristic places every loop in  its own fragment (to make sure that the FDG is well-formed) and, for code outside loops, \toolname creates a new fragment whenever it detects an update to  monitor state (i.e., \code{this.fld = *}). 
In practice, we found this heuristic to achieve a good balance between the number of parallelization opportunities and the size of the resulting FDG.\footnote{{An ablation study that justifies the design of this heuristic can be found in \iffull{Appendix~\ref{sec:fdg-ablation}}\else{the extended version of the paper}\fi.}}

\paragraph{\textbf{Static analysis optimization.}} Our approach uses an off-the-shelf pointer analysis to detect which pairs of FDG fragments do not have a data race (and, so, can run in parallel). However, such an approach, based on pointer analysis alone,  often leads to imprecision. For example, Soot's pointer analysis cannot prove that fragments 2 and 6 in Figure~\ref{fig:motiv-src} do not contain any races, as it does not reason about individual array elements.  Therefore, in order to increase the precision of the static analysis, \toolname implements an SMT-based static analysis  on top of Soot's built-in pointer analysis and  generates appropriate verification conditions {(similar to the ones generated by~\citet{seahorn})} to prove that memory accesses of two fragments are disjoint.

%% file: evaluation.tex
\section{Evaluation}\label{sec:eval}
We evaluated \toolname{'s ability to generate fine-grained locking protocols} by performing a set of experiments that are designed to answer the following research questions:
\begin{enumerate}[label=\textbf{RQ\arabic*}]
    \item How does the code generated by \toolname compare against explicit-synchronization monitors  written by experts?
    \item How does the technique implemented in \toolname compare against other compile-time state-of-the-art approaches targeting implicit-synchronization monitors?
    \item How does the static analysis for inferring safe interleavings impact the quality of the code generated by \toolname? 
    \item How long does \toolname take to synthesize code and how complex are the resulting protocols?
\end{enumerate}

To answer these research questions, we conducted  experiments on ten explicit-synchronization monitors from popular open source repositories. Aside from \toolname, we consider two additional baselines, described below, that aid us in answering our second and third research questions. %

\paragraph{\textbf{Benchmarks.}} The benchmarks used in our evaluation are collected from popular open source GitHub repositories. We wrote a crawler ({based on GitHub's REST API~\cite{github-api}}) to automatically identify candidate explicit-synchronization monitors implemented in Java by searching for keywords like \code{lock}, \code{unlock}, \code{await}, etc. We then manually inspected class files returned by the crawler in decreasing order of GitHub popularity (stars and forks) and identified self-contained monitor-style classes that encapsulate shared state accessed by multiple threads. {We included such a monitor in our benchmarks only if it satisfies the following conditions: (1) the class has a well-defined API for client threads and (2) it contains \emph{ parallelization opportunities that can be realized via fine-grained locking}.\footnote{{If a monitor does not contain parallelization opportunities, our technique generates code equivalent to that synthesized  by~\citet{expresso}. Since the goal of our evaluation is to evaluate \toolname's ability to generate fine-grained locking protocols, we did not include benchmarks from prior work~\cite{expresso, hung:autosynch} that do not contain such parallelization opportunities.}}} We manually isolated the {shared state and monitor methods} of the class file to obtain a standalone explicit-synchronization monitor and then manually translated it to an equivalent implicit monitor. {To convert a benchmark to an implicit monitor, we simply removed all synchronization code (i.e., locking and signaling operations) and introduced appropriate \wuntil statements. In total, we collected 10 monitors from popular repositories such as Spring Framework (a Java-based framework for creating enterprise applications), Java JDK, Apache Spark (an analytics engine for large-scale data processing), etc.\footnote{All benchmarks are publicly available here: https://github.com/utopia-group/cortado}} %

\paragraph{\textbf{Baselines.}} As mentioned above, our evaluation uses two additional baselines in order to answer RQ2 and RQ3. To compare against other compile-time techniques (RQ2), we evaluate \textsc{Expresso}~\cite{expresso}, a tool that addresses the same problem as this paper but generates an explicit signal monitor using a \emph{single global lock} shared by {all} monitor methods. To evaluate the importance of our static analysis (RQ3), we created an ablated version of \toolname, called \toolnameminus, which uses a very coarse analysis to infer safe interleavings. This ablated version considers $(v, (v_s, v_t))$ a safe interleaving only if $v$ does not have any data races with any predecessor (resp. successor) of $v_s$ (resp. $v_t$). This is a sufficient condition for strong safety, but it only requires checking data races rather than discharging a set of Hoare triples.

\paragraph{\textbf{Evaluating performance.}} Following prior work~\cite{expresso,hung:autosynch}, we evaluate the performance of each monitor implementation by performing \emph{saturation tests}~\cite{cherem:lock-atomic} wherein threads perform monitor operations without doing any additional work. %
We collect our performance measurements using the Java Microbenchmark Harness (JMH)~\cite{alekseyshipilevOpenJDKJmh}. %
All measurements are conducted on a 112-way (56-core $\times$ 2 SMT) Intel Xeon CPU W-3275 2.50GHz with 256 GB of memory using JDK 1.8.0\_272. {In this section, we present results for each benchmark for up to 128 threads, chosen as an arbitrary stopping point past the total number of hyper-threads. Results for up to 256 threads can be found \iffull{in Appendix~\ref{sec:additional-exps}}\else{in the extended version of the paper}\fi.}

\subsection{Performance Results}
 Figure~\ref{fig:perf-results} plots  the average time taken per monitor method invocation (i.e., milliseconds/operation) against the number of threads. In what follows, we analyze the plots in more detail and present several conclusions drawn from these results. Because our benchmarks only contain monitors where fine-grained locking is beneficial,
we emphasize that our conclusions only apply to such monitors.

\paragraph{\textbf{Comparison against hand-written implementations (RQ1).}} For every benchmark, the explicit synchronization monitor generated by \toolname performs \emph{better} than the expert-written  implementation as the number of threads increases. In particular, \toolname-synthesized code performs on average $3.7\times$\footnote{In order to handle outliers such as in JobWrapper, for all reported aggregate speedups (max, mean, etc.) we throw out data points with a z-score greater than two.} and up to $39.1\times$ times faster than the original code.

\paragraph{\textbf{Comparison against \textsc{Expresso} (RQ2).}} \toolname-generated explicit monitors perform better than \textsc{Expresso} explicit monitors generated from the same implicit specification on all benchmarks. \toolname-synthesized code outperforms \textsc{Expresso} by $4.0\times$ on average (and up to $48.7\times$).

\paragraph{\textbf{Comparison against \toolnameminus (RQ3).}} Finally, we analyze how \toolname compares to its simplified version, \toolnameminus, which does not use the results of the safe interleavings analysis from Section~\ref{sec:analysis}. In five cases, the code generated by \toolnameminus is equivalent to the code generated by \textsc{Expresso} and therefore worse than \toolname. In two other cases (PausableThreadPoolExecutor and ProgressTracker), \toolnameminus generates code different from both \textsc{Expresso} and \toolname. For PausableThreadPoolExecutor, the code generated by \toolnameminus is slower than that of \textsc{Expresso} because many of the operations it parallelizes are very cheap, so the overhead of extra locks outweighs their benefit.
On the other hand, our static analysis detects several safe interleavings which enable \toolname to synthesize a protocol with cheaper synchronization operations. Finally, for the remaining three cases, the code generated by \toolnameminus matches the one generated by \toolname. This ablation study demonstrates that the safe interleaving analysis from Section~\ref{sec:analysis} helps extract additional concurrency on five of our benchmarks.

\begin{figure}
    \begin{subfigure}{0.33\textwidth}
        \includegraphics[width=\textwidth,keepaspectratio]{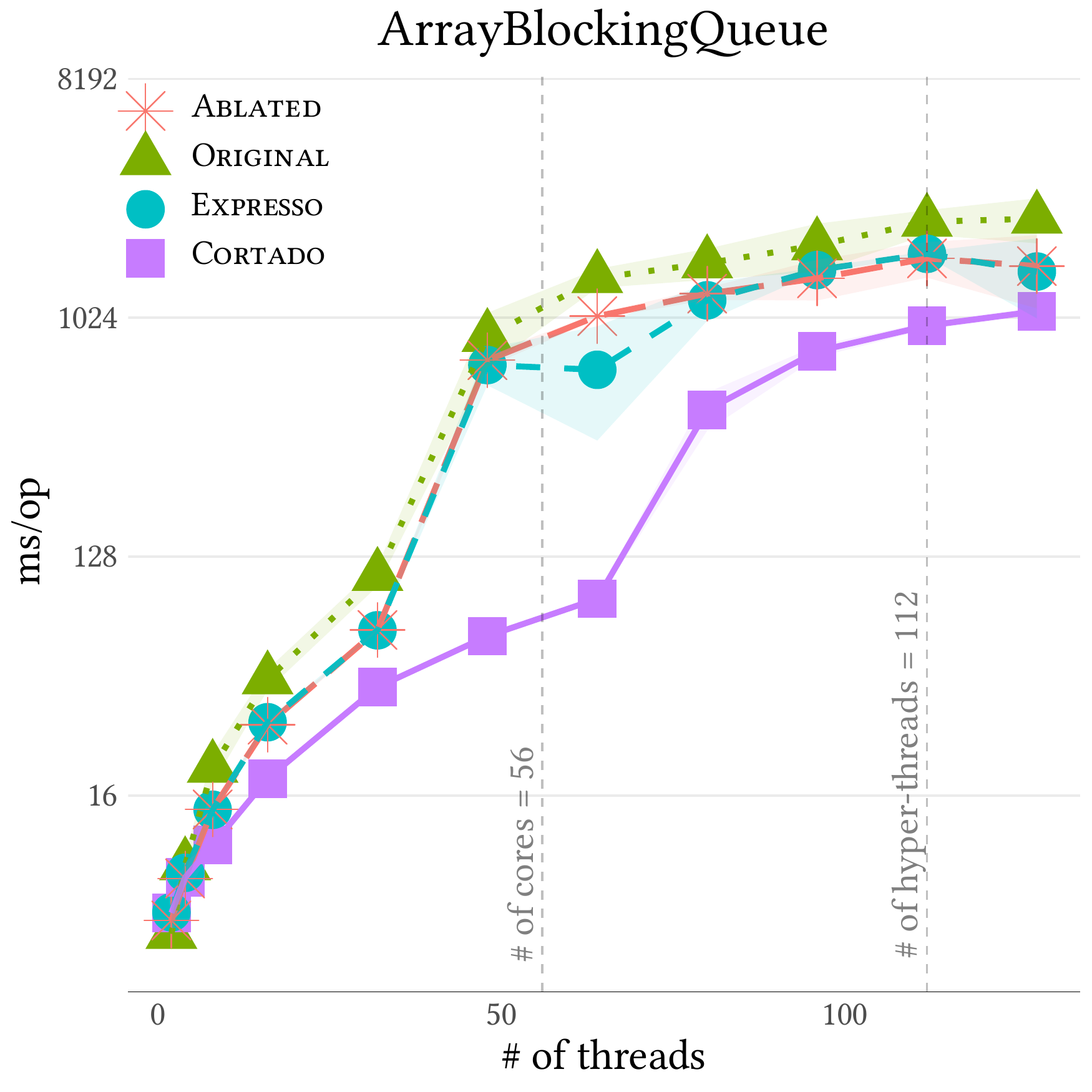}
    \end{subfigure}%
    \begin{subfigure}{0.33\textwidth}
        \includegraphics[width=\textwidth,keepaspectratio]{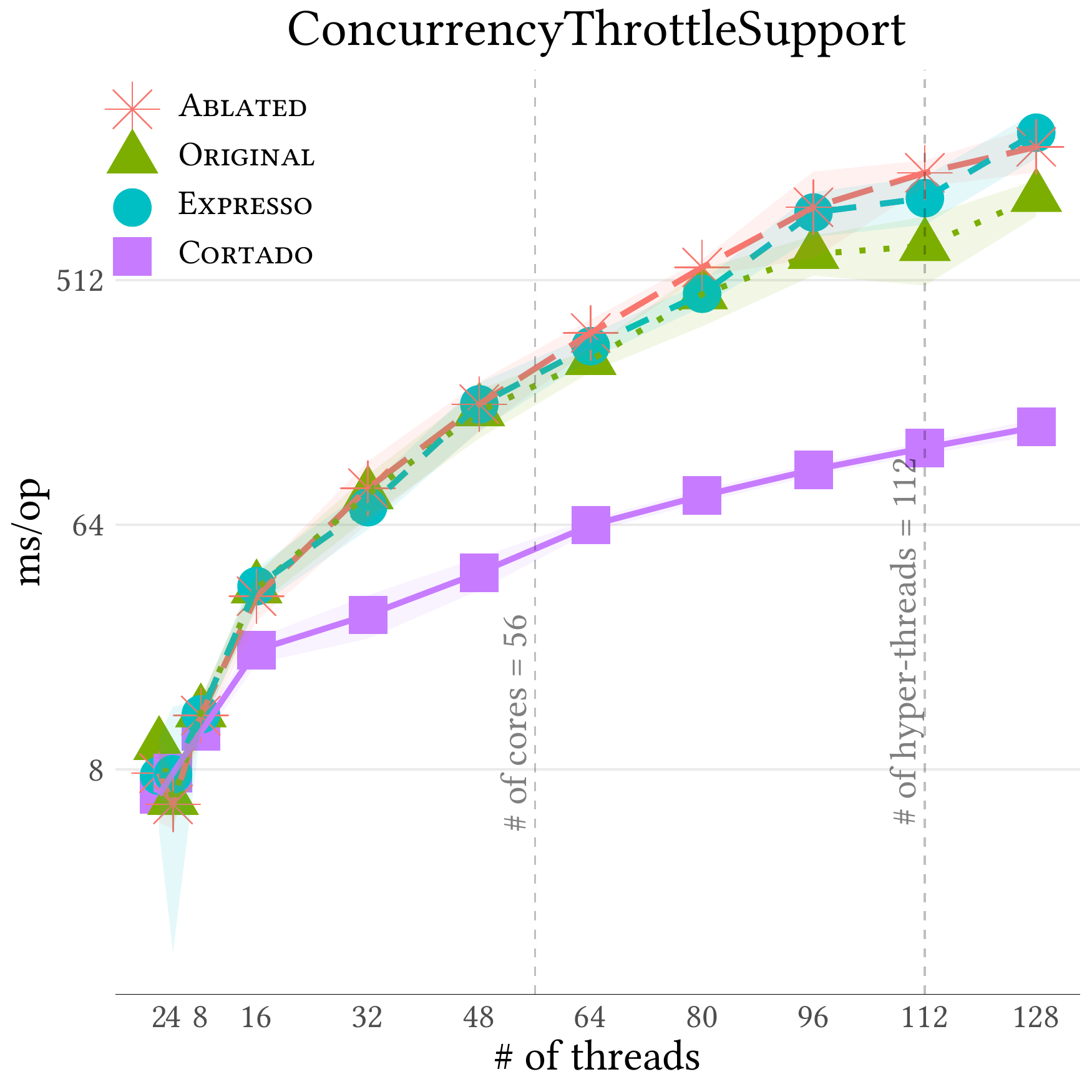}
    \end{subfigure}%
    \begin{subfigure}{0.33\textwidth}
        \includegraphics[width=\textwidth,keepaspectratio]{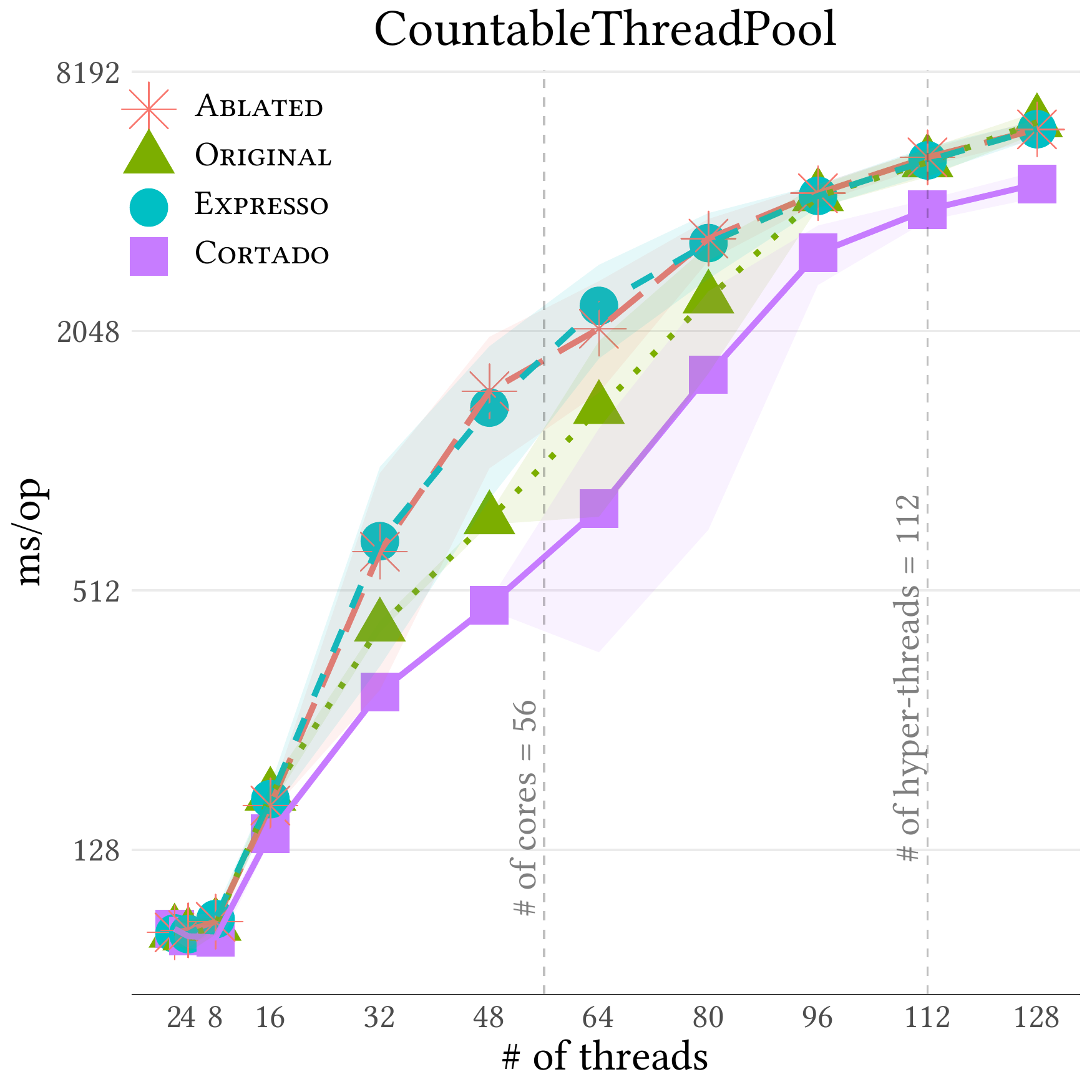}
    \end{subfigure}
        \begin{subfigure}{0.33\textwidth}
        \includegraphics[width=\textwidth,keepaspectratio]{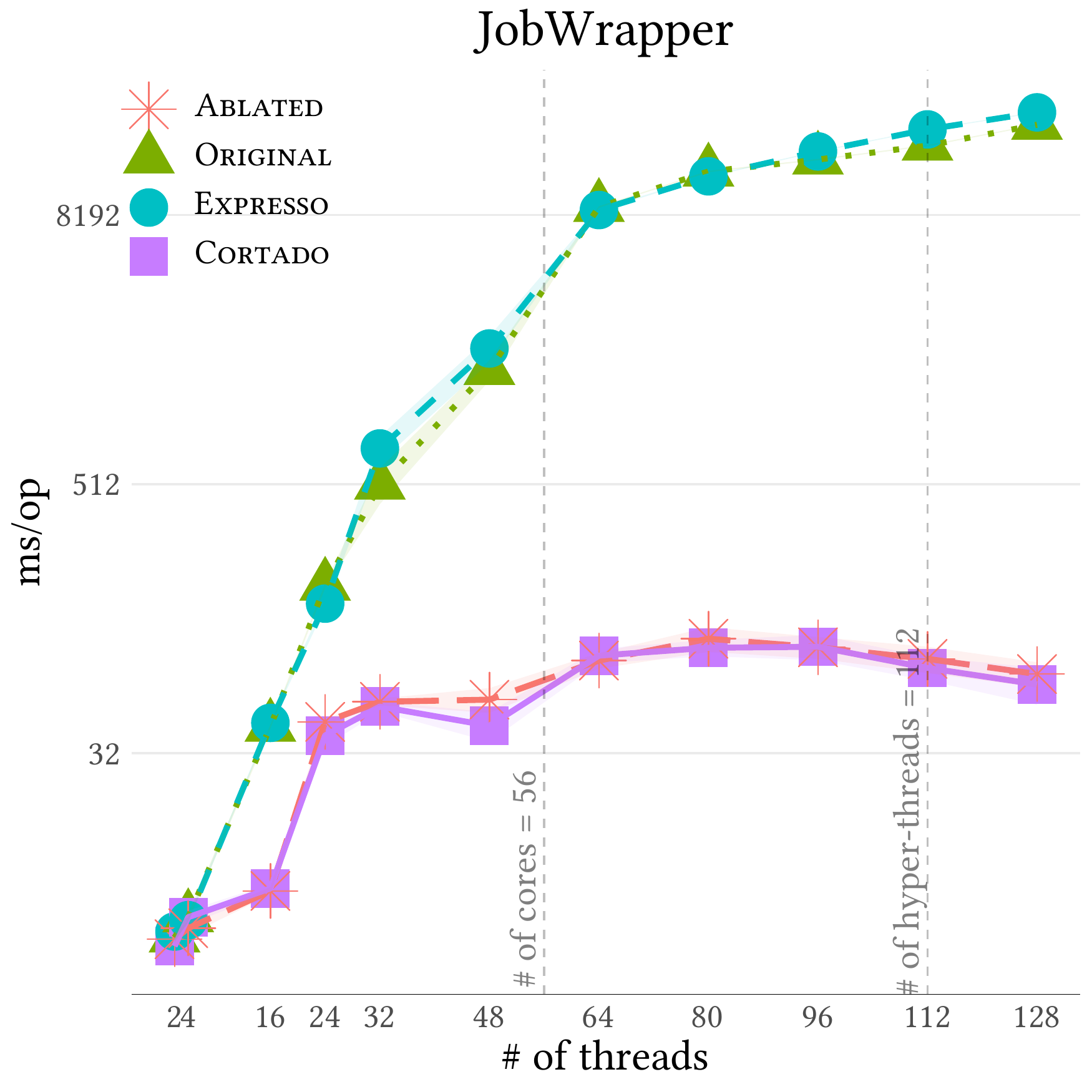}
    \end{subfigure}%
    \begin{subfigure}{0.33\textwidth}
        \includegraphics[width=\textwidth,keepaspectratio]{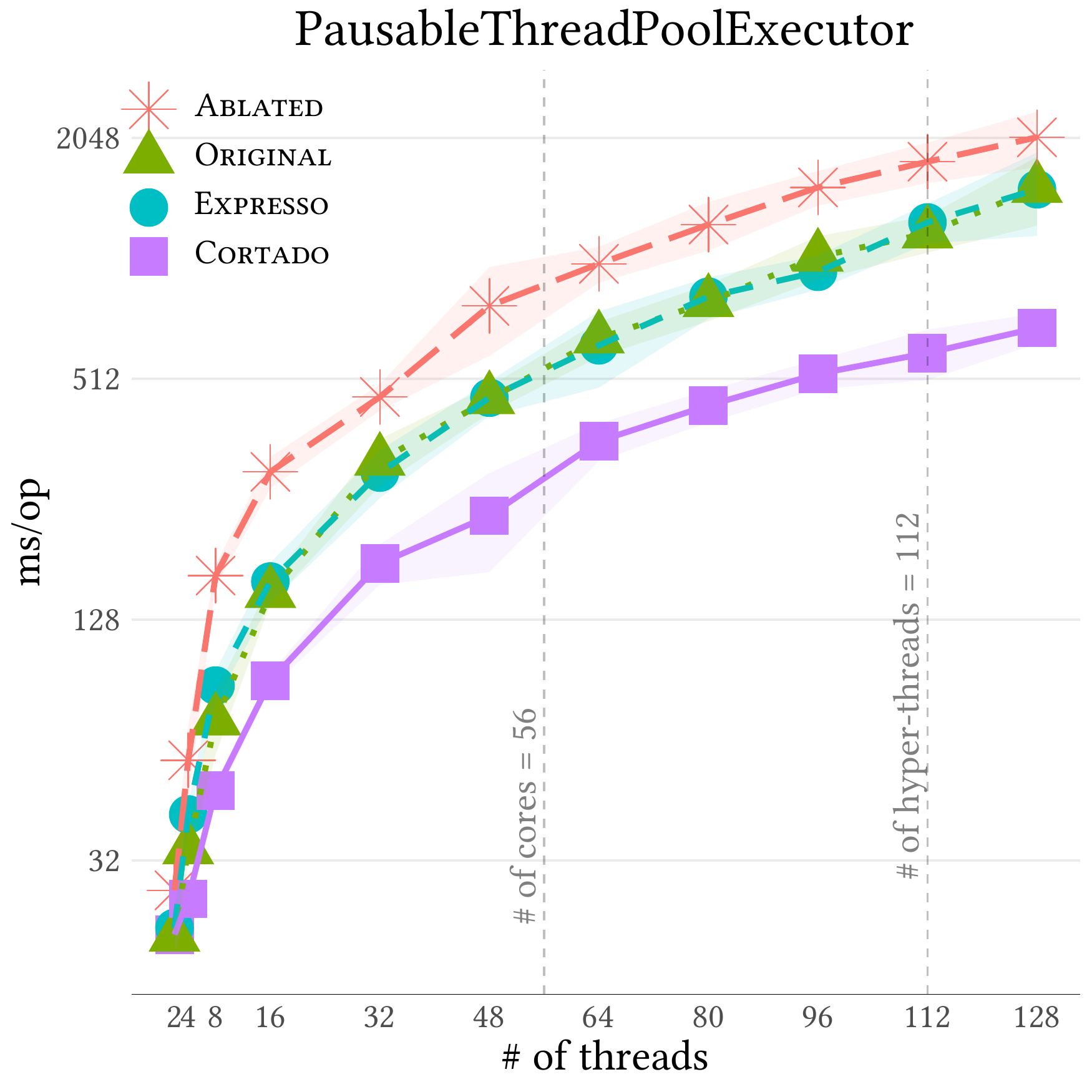}
    \end{subfigure}%
    \begin{subfigure}{0.33\textwidth}
        \includegraphics[width=\textwidth,keepaspectratio]{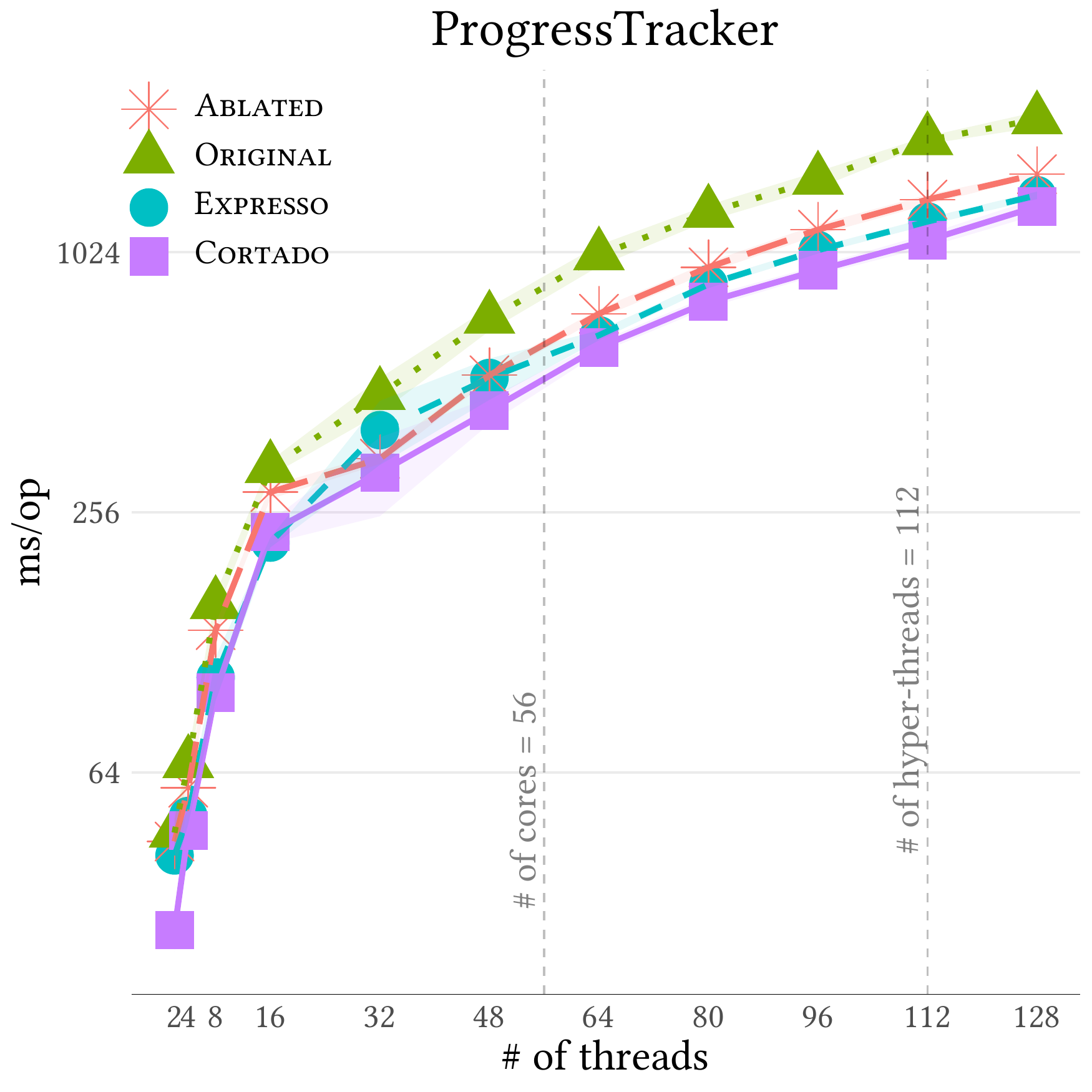}
    \end{subfigure}
    \begin{subfigure}{0.33\textwidth}
        \includegraphics[width=\textwidth,keepaspectratio]{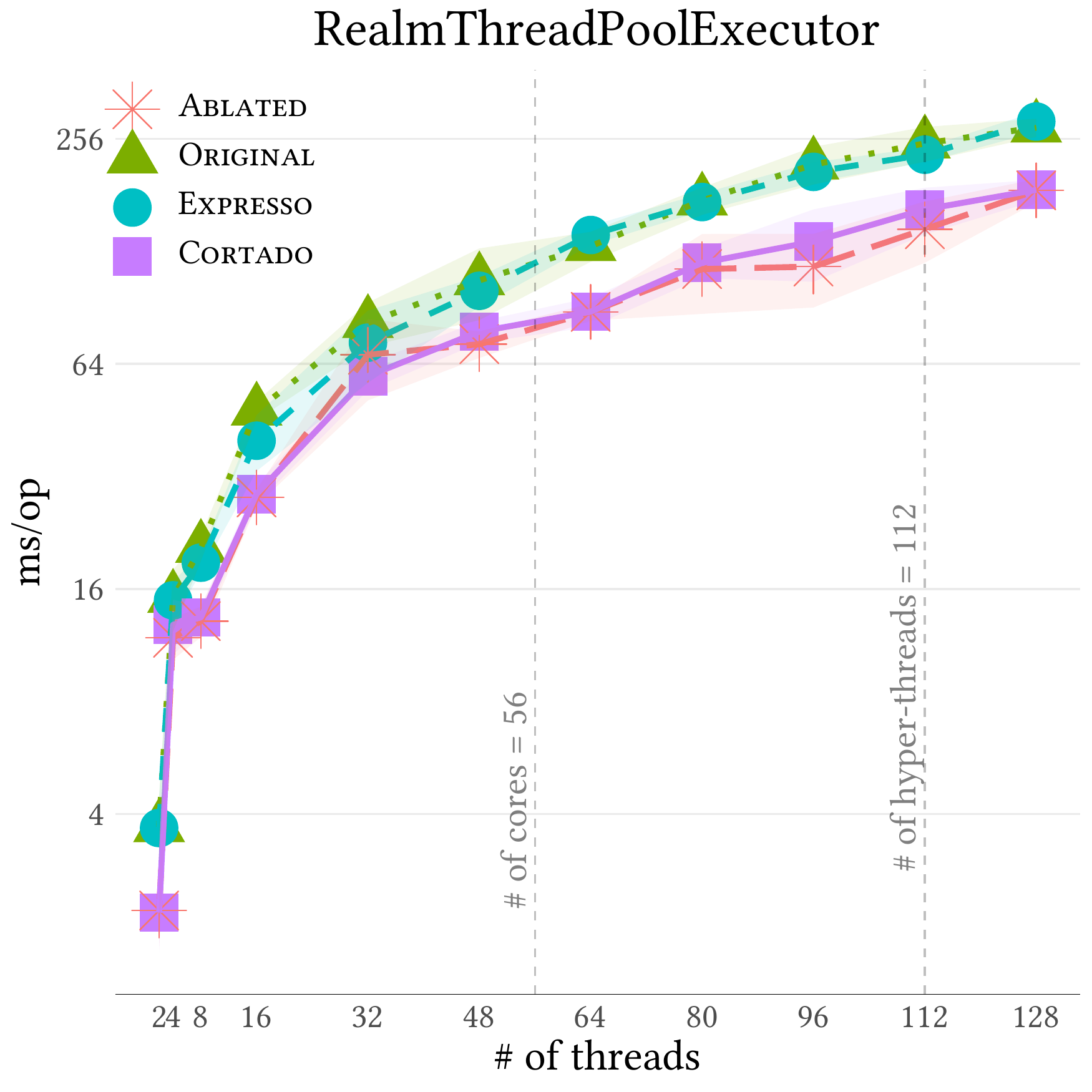}
    \end{subfigure}%
        \begin{subfigure}{0.33\textwidth}
        \includegraphics[width=\textwidth,keepaspectratio]{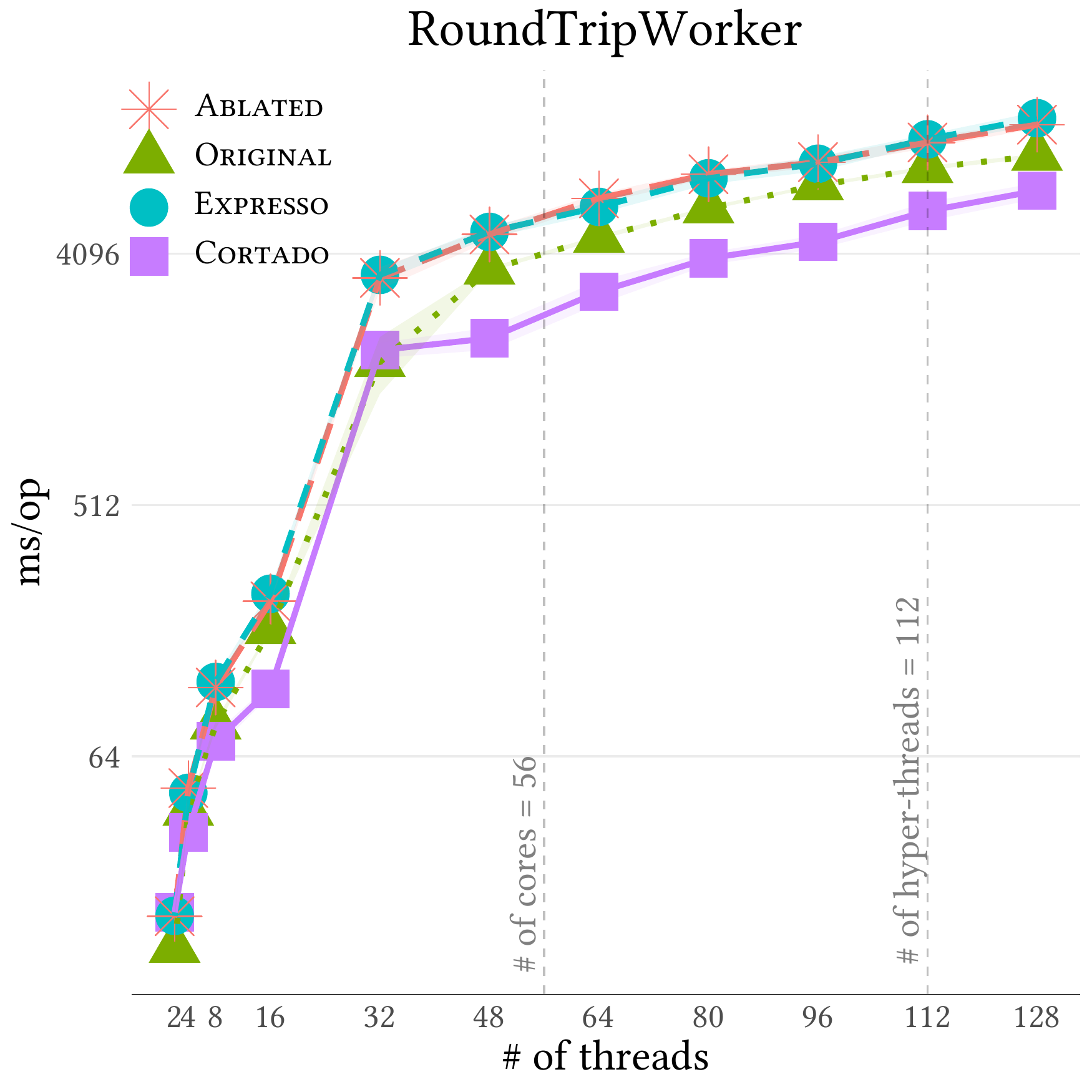}
    \end{subfigure}
    \begin{subfigure}{0.33\textwidth}
        \includegraphics[width=\textwidth,keepaspectratio]{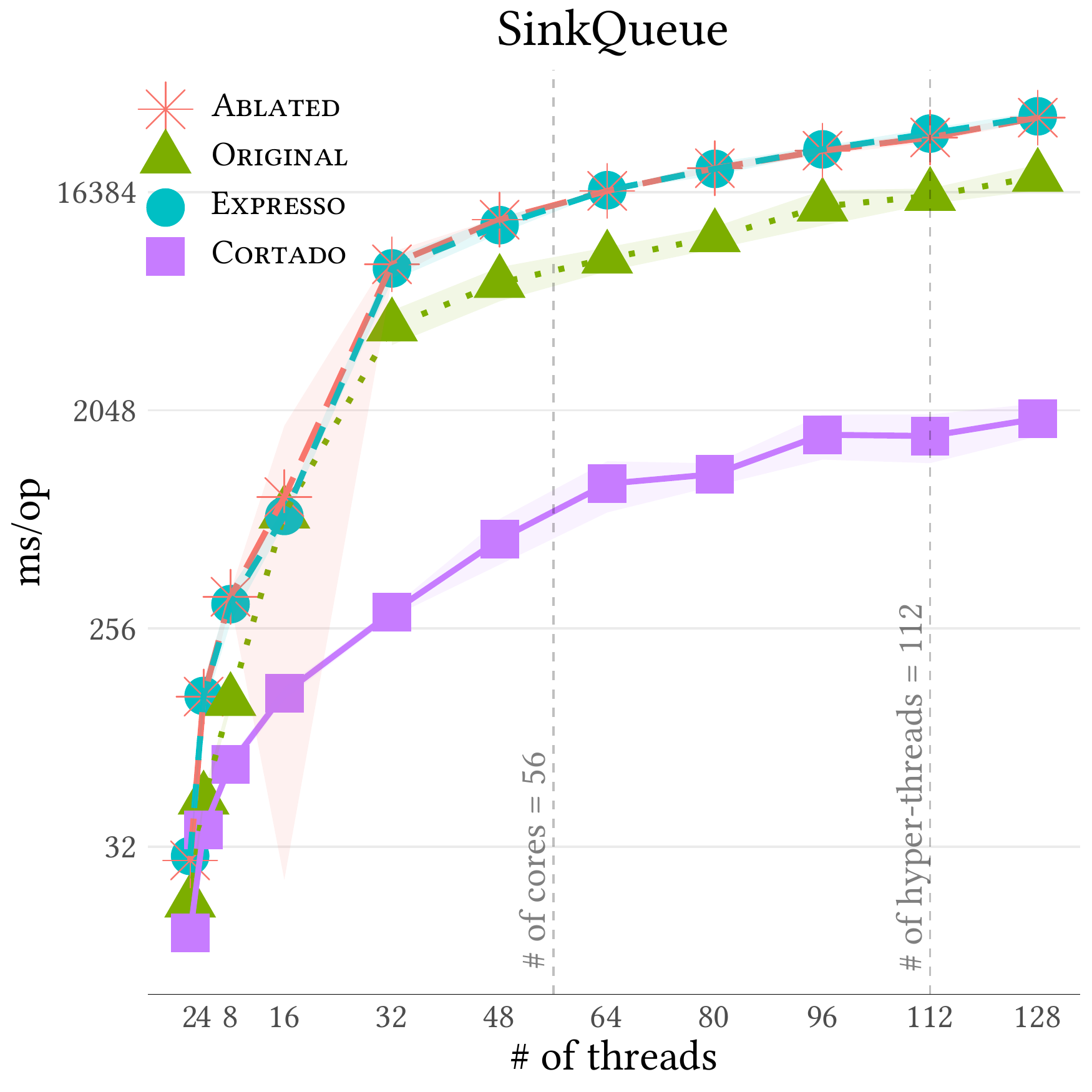}%
    \end{subfigure}%
    \begin{subfigure}{0.33\textwidth}
        \includegraphics[width=\textwidth,keepaspectratio]{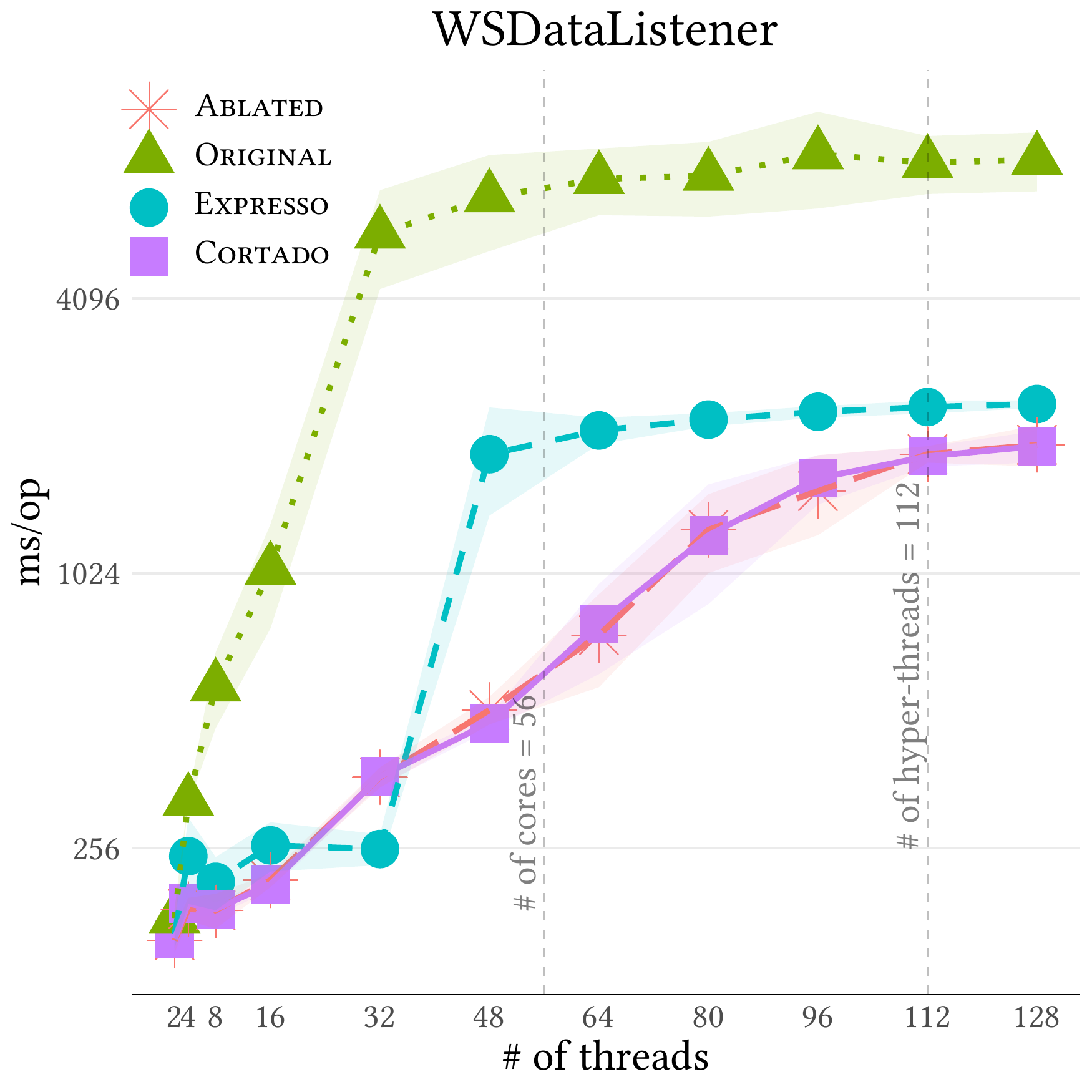}
    \end{subfigure}
    \caption{Performance Results For All Tools. The y-axis is in log scale and time measurements are in milliseconds. The shadowed regions surrounding each line present 99.9\% confidence interval of each measurement.}
    \label{fig:perf-results}
\end{figure}

\subsection{Synthesis Time \& Protocol Complexity}

To evaluate the cost and complexity of synthesizing code with \toolname (RQ4), Table~\ref{table:runtime-protocols} summarizes its running time and presents some statistics about the synthesized protocols. For each benchmark, we report the running time for the various phases of the tool: pointer analysis with Soot, signal placement with \textsc{Expresso}, and synthesis with \toolname. We also report the number of locks and atomic fields in the synthesized protocol.

Table~\ref{table:runtime-protocols} shows that \toolname terminates in under one minute for all but two benchmarks. For these two outliers, the synthesis time  is dominated by {\sc Expresso}'s monitor invariant inference, which is necessary for signal placement. Overall, \toolname is able to extract better performance than \textsc{Expresso} alone with only a small additional compile-time cost. %

The last three columns in Table~\ref{table:runtime-protocols} provide statistics about the synthesized explicit monitors. Most monitors benefit from \toolname's ability to introduce atomic fields, which reduces the overhead of operations on monitor state that would otherwise require a lock. The lines of code (LOC) results show that \toolname synthesizes explicit monitors that are on average $1.7\times$ larger than their implicit specifications.

\begin{table}
    \centering
    \footnotesize
    \begin{tabular}{l r @{\quad} r r r r @{\quad} r r r r r}
    \toprule
    && \multicolumn{4}{c}{Synthesis Time (secs)} & \multicolumn{3}{c}{Synthesized Protocol} \\
    \cmidrule(r){3-6} \cmidrule{7-9}
      Benchmark 
          & LOC
          & Soot
          & \textsc{Expresso}
          & \textsc{Cortado}   
          & Total    
          & \#Lock/\#Op
          & \#Atomic/\#Op
          & LOC
          \\
      \midrule
      ArrayBlockingQueue 
          & 287
          & 16.7
          & 1897.7
          & 372.7
          & 2302.5
          & 2 / 18
          & 1 / 25
          & 514
          \\
      ConcurrencyThrottleSupport 
          & 33
          & 17.0
          & 1.4
          & 0.2
          & 18.7
          & 1 / \phantom{0}1
          & 1 / \phantom{0}4
          & 68
          \\
      CountableThreadPool    
          & 54
          & 20.4
          & 0.8
          & 0.2
          & 21.5
          & 1 / \phantom{0}1 
          & 1 / \phantom{0}5
          & 85
          \\
      JobWrapper
          & 33
          & 17.0
          & 0.6
          & 0.1
          & 17.7
          & 1 / \phantom{0}1 
          & 1 / \phantom{0}3
          & 63
          \\
      PausableThreadPoolExecutor 
          & 79
          & 20.7
          & 0.8
          & 0.9
          & 22.5
          & 3 / \phantom{0}4
          & 2 / \phantom{0}9
          & 122
          \\
      ProgressTracker     
          & 65
          & 0.7
          & 6.8
          & 0.2
          & 8.0
          & 2 / \phantom{0}7
          & 1 / \phantom{0}4
          & 119
          \\
      RealmThreadPoolExecutor 
          & 34
          & 18.4
          & 0.3
          & 0.1
          & 18.7
          & 1 / \phantom{0}1 
          & 1 / \phantom{0}3
          & 61
          \\
      RoundTripWorker 
          & 62
          & 16.7
          & 3.1
          & 0.4
          & 20.4
          & 2 / \phantom{0}4 
          & 1 / \phantom{0}4
          & 103
          \\
      SinkQueue    
          & 75
          & 15.7
          & 981.4
          & 34.6
          & 1031.8
          & 2 / \phantom{0}4 
          & 1 / \phantom{0}8
          & 131
          \\
      WSDataListener 
          & 158
          & 18.3
          & 5.2
          & 1.0
          & 25.1
          & 4 / 11
          & 0 / \phantom{0}0
          & 222
          \\
    \bottomrule
    \end{tabular}

    \caption{Synthesis time for each phase of \toolname and summaries of the synthesized protocols.
    LOC is lines of code.
    Soot indicates pointer analysis time, \textsc{Expresso} is the time for monitor invariant generation and signal placement,
    and \toolname shows the additional time on top of Soot and \textsc{Expresso}.}
    \label{table:runtime-protocols}
\end{table}

%% file: related.tex
\section{Related Work}

\paragraph{Monitor abstractions.}
The notion of monitors as an organizing abstraction for concurrent programming
originates with \citet{hoare:monitors} and \citet{hansen:osp}.
Monitors offer the same synchronization facilities as semaphores---%
the ability to coordinate multiple threads and enforce mutual exclusion---%
but encapsulate the state protected by those facilities
and automate mutual exclusion when entering and exiting the monitor's operations.
\citet{lampson:mesa} further extended the monitor abstraction in the Mesa programming language
to handle spurious wake ups.

These early monitor abstractions are \emph{explicit-signal} monitors
in the taxonomy of \citet{buhr:monitors}
because they require the programmer to explicitly insert condition variables and signalling operations
to coordinate threads within the monitor.
This requirement places both a safety and a liveness burden on the programmer:
they must place signals correctly to preserve invariants about the monitor's state,
but must also insert enough signals to avoid deadlock.
An alternative is to use an \emph{implicit-signal} (or \emph{automatic-signal}) monitor,
in which signals are inserted automatically by the compiler, language runtime, or operating system.
\citet{hoare:ccr} proposed the notion of \emph{conditional critical regions} (CCRs),
which allow for monitor operations to block until a \emph{guard} predicate over the monitor state
is satisfied by some other thread.
A CCR implementation would automatically block and signal threads in a fashion consistent
with this guard semantics.

Implicit-signal monitors simplify concurrent programming,
but come at a steep performance cost---%
\citet{buhr:monitors} estimate that implicit-signal monitors are 10--50$\times$ slower than explicit ones.
More recent work has tried to lower the cost of implicit-signal monitors.
AutoSynch~\citep{hung:autosynch} uses a combination of compile-time instrumentation and run-time evaluation
to efficiently compute which threads should be woken when monitor state changes.
This approach lowers the cost of implicit monitors to be close to, or sometimes better than,
explicit ones.
Expresso~\citep{expresso} takes a different approach,
using compile-time static analysis
to \emph{synthesize} an explicit-signal monitor equivalent to an implicit-signal version given as input.
In this way, Expresso is able to erase the dynamic cost of implicit-signal monitors,
and in most cases is comparable to hand-written explicit monitors. However, Expresso uses a single lock for the entire monitor and does not allow concurrent execution of threads within the monitor even when safe. 
Our work expands on this direction by using a richer static analysis
to infer additional concurrency opportunities and uses MaxSAT to synthesize a safe and efficient locking protocol. 
Hence, 
our key contribution is to synthesize an explicit monitor
that \emph{appears} to match the semantics of the implicit one,
but runs monitor operations concurrently when possible and efficient. As we show in the evaluation, our proposed approach  can often make the synthesized monitor \emph{faster} than a hand-written equivalent.

\paragraph{Automatic synchronization.}
An appealing approach to lower the difficulty of concurrent programming
is to deploy program analysis and synthesis techniques for automation.
The common abstraction for much of this work is for the programmer to annotate \emph{atomic sections}
that should be executed atomically.
\citet{emmi:lock-allocation} present a technique for lock allocation to an annotated program.
They reduce the problem to integer linear programming
and deploy the resulting tool on large-scale C and Java programs. 
Other approaches~\cite{hicks:lock,component-locking,autolocker}, on the other hand, take a purely static analysis route and attempt to maximize parallelism based solely on the results of the analysis. 
\citet{cherem:lock-atomic} present an alternative technique
that uses runtime support to enable finer-grained concurrency.
Compared to these efforts,
\cortado applies to the more limited domain of monitors,
but in exchange for this limitation is able to reason about conditional signalling
and can allow atomic sections to run concurrently so long as the illusion of atomicity is maintained.

Other approaches start from a sequential program and automatically generate an equivalent concurrent program. The closest work to ours in this space is that of~\citet{golan2011automatic} which generates concurrent data structures given their sequential implementation. Compared to our method, their approach is applicable only to data structures that satisfy certain shape properties and all synthesized programs adhere to the same locking protocol, whereas \toolname generates a synchronization protocol specialized to the input monitor.

\paragraph{Concurrency verification.}
\cortado reasons about concurrent program executions
by building on work in concurrent program analysis and verification.
Our notion of left- and right-commutativity (Definition~\ref{def:commute})
comes from \citeauthor{reduction-lip}'s work on reduction as a concurrency proof technique~\citep{reduction-lip}.
Reduction translates interleaved program executions
to simpler, equivalent sequential executions
by exploiting the commutativity properties of individual program steps.
We use the same idea but in reverse:
starting with a sequential history (Definition~\ref{def:seq-hist}),
we use a static analysis of commutativity
to determine how to safely introduce interleavings into that history,
and use that information to determine how to assign locks to program fragments.

%% file: concl.tex
\section{Conclusion}

We presented a technique for synthesizing fine-grained synchronization protocols for implicitly synchronized monitors. Our approach first employs a novel static analysis to identify safe interleavings opportunities between code fragments and uses the results  of this analysis to generate a MaxSAT encoding whose solution can be used to synthesize an efficient and correct-by-construction explicit-synchronization monitor. We have implemented our method in a tool called \toolname and evaluated its effectiveness  eight monitors collected from popular open source applications. The results of our experimental evaluation demonstrate that \toolname is able to generate non-trivial synchronization protocols that are $3.7\times$ times faster than the original implementation on average (and up to $39.1\times$ times for some outliers).

%% file: ablations.tex
\newcommand{\cortadoPlot}[1]{
    \begin{subfigure}{0.33\textwidth}
        \includegraphics[width=\textwidth,keepaspectratio]{#1}
    \end{subfigure}\nolinebreak
}

\newcommand{\allCortadoPlots}[1]{
    \cortadoPlot{#1ArrayBlockingQueue.pdf}
    \cortadoPlot{#1ConcurrencyThrottleSupport.pdf}
    \cortadoPlot{#1CountableThreadPool.pdf}
    
    \cortadoPlot{#1JobWrapper.pdf}
    \cortadoPlot{#1PausableThreadPoolExecutor.pdf}
    \cortadoPlot{#1ProgressTracker.pdf}
    
    \cortadoPlot{#1RealmThreadPoolExecutor.pdf}
    \cortadoPlot{#1RoundTripWorker.pdf}
    
    \cortadoPlot{#1SinkQueue.pdf}
    \cortadoPlot{#1WSDataListener.pdf}
}

\section{Additional Experimental Evaluations}\label{sec:additional-exps}

This section presents additional experimental data for the evaluation of Section~\ref{sec:eval} as well as ablations related to several design decisions in the implementation of \toolname. In particular, Section~\ref{sec:more-data-points} presents additional data points for the evaluation presented in Section~\ref{sec:eval}, Section~\ref{sec:fdg-ablation} presents an ablation study that justifies the design of our heuristic for constructing an FDG, and Section~\ref{sec:weight-ablation} presents an ablation study that demonstrates the need for adjusting the weights of soft constraints in our MaxSAT encoding.

\subsection{Additional Data Points per Benchmark}\label{sec:more-data-points}

In this Section we present additional data points for all the experiments presented in Section~\ref{sec:eval}. As mentioned in Section~\ref{sec:eval}, we chose 128 threads as a stopping point past the number of total hyper-threads in the machine used in our evaluation. In this Section, we provide data points up to 256 threads.

Figure~\ref{fig:perf-results-256} presents the results for all benchmarks in our evaluation up to 256 threads. As demonstrated by Figure~\ref{fig:perf-results-256}, the general trend for all benchmarks is the same as the date presented in Section~\ref{sec:eval}. Note that for benchmarks where the code generated by \toolname exhibits similar run-time performance with the other three implementation we consider in our evaluation (e.g., RoundTripWorker), context-switches seem to dominate the running  time as the number of threads increases.

\begin{figure}
    \allCortadoPlots{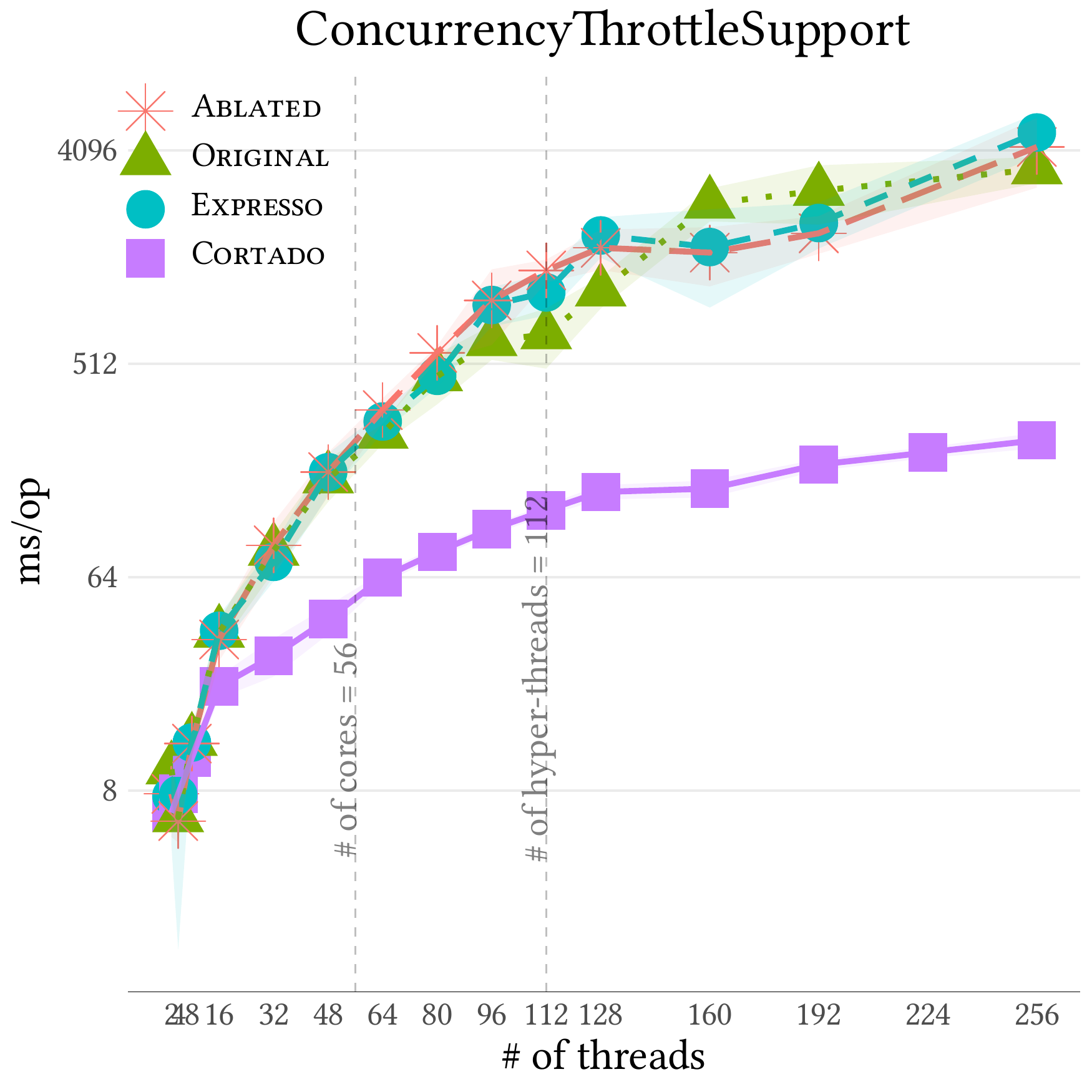}

    \caption{Performance results for all tools up to 256 threads. The y-axis is in log scale and time measurements are in milliseconds. The shadowed regions surrounding each line present 99.9\% confidence interval of each measurement.}
    \label{fig:perf-results-256}
\end{figure}

\subsection{Ablation Study for FDG Construction Heuristic}\label{sec:fdg-ablation}

This section presents an ablation study for the design of our FDG construction heuristic described in Section~\ref{sec:impl}. To justify the decisions behind the design of our heuristic, we implemented two additional versions of \toolname that differ in the way they construct the FDG. In particular, we implemented a version that creates finer-grained FDGs than \toolname and one that creates coarser-grained ones. The finer-grained version, named $\textsc{Stmt-Ablation}$, puts every non-composite statement outside of a loop in its own fragment. Statements inside a loop are grouped together in the same fragment since FDGs are acyclic. The coarser version of our tool, named $\textsc{CCR-Ablation}$, simply puts each CCR in its own fragment.

Figure~\ref{fig:fragment-ablation} presents the results of this ablation study. As demonstrated by the results, there are several cases where \toolname performs better than at least one of its two modified versions. The cases where all three versions perform similarly are benchmarks where the synthesized synchronization protocol mainly exploits data-level parallelism among different CCRs in the monitor, which our tool can exploit given any FDG. Overall, this study demonstrates the need for a customized heuristic for constructing an FDG suitable for maximizing parallelism of implicit synchronisation monitors.

\begin{figure}
    \allCortadoPlots{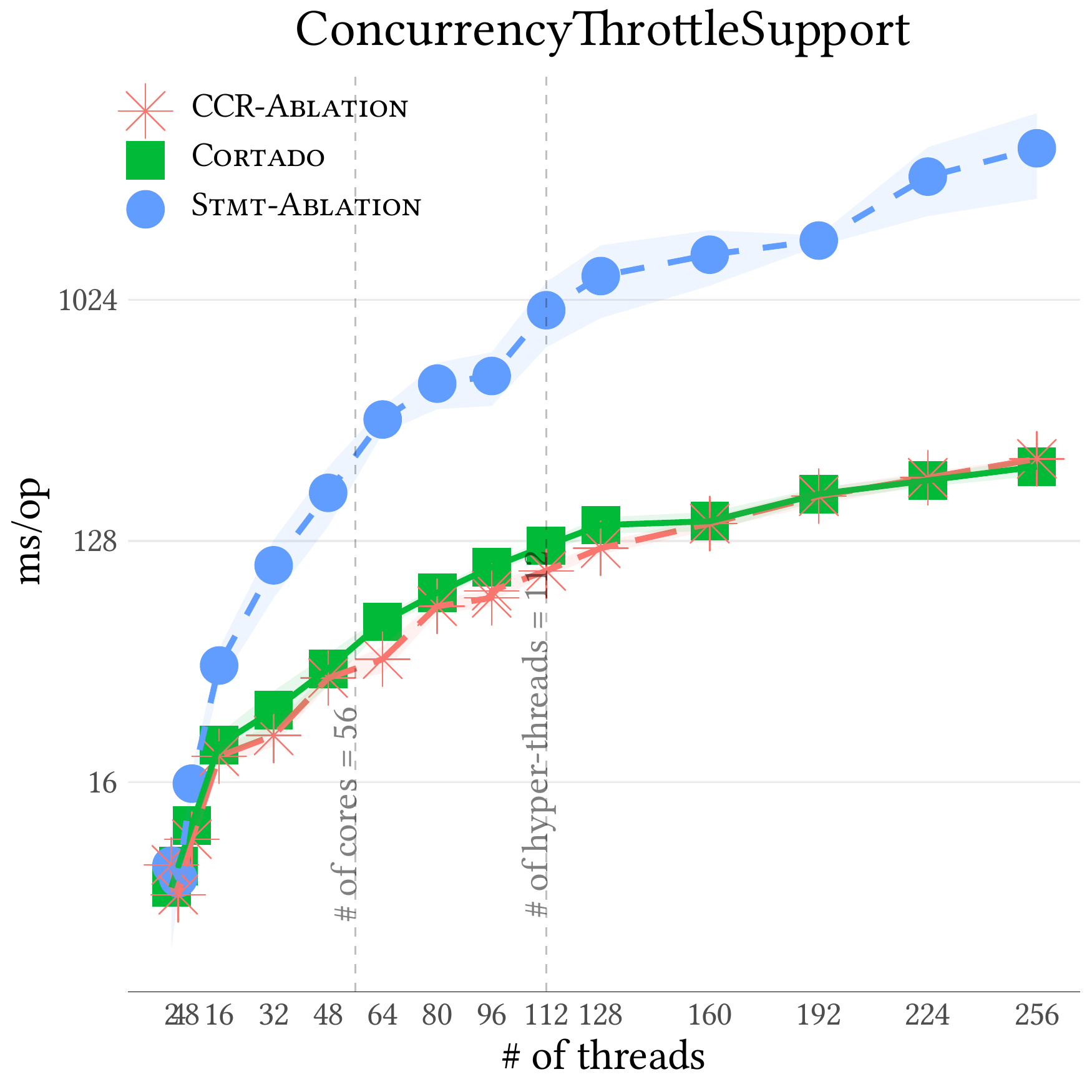}

    \caption{Performance results for the FDG ablation study. The y-axis is in log scale and time measurements are in milliseconds. The shadowed regions surrounding each line present 99.9\% confidence interval of each measurement.}
    \label{fig:fragment-ablation}
\end{figure}

\subsection{Ablation Study for MaxSAT Soft Constraint Weights}\label{sec:weight-ablation}

Finally, this Section presents an ablation of the soft constraints weights in our MaxSAT encoding. As mentioned in Section~\ref{sec:impl}, \toolname assigns different weights to different classes of soft constraints because some of them are more important for synthesizing the optimal synchronization protocol. To demonstrate this, we have created a modified version of our tool, named $\textsc{Weight-Ablation}$, that assigns the same weight to all soft constraints.

\noindent\begin{minipage}{\textwidth}
\useparinfo
\begin{wrapfigure}{r}{0.4\textwidth}
  \scriptsize
      \begin{minted}[xleftmargin=3em,autogobble,escapeinside=||]{java}
          class M {
            int x = 0;
            void foo() { 
              waituntil(x < 10);
              x++;
            }
            
            void bar() { 
                x--; 
            }
        }
      \end{minted}
  \caption{A simple implicit monitor.}
      \label{fig:mtr-weight-ablation}
\end{wrapfigure}
The results of this ablation are presented in Figure~\ref{fig:weight-ablation}. As this figure demonstrates, there are several cases where the ablated version of the tool performs significantly worse than \toolname. To give a concrete example of why this is the case, consider the monitor of Figure~\ref{fig:mtr-weight-ablation} that contains two methods both of which modify field $x$. Because the body of method \code{bar} can be interleaved between the \wuntil statement and the increment statement of method \code{foo}, the optimal synchronization protocol would convert field $x$ to an atomic integer and introduce a lock that would only be held in method \code{foo}. However, an equivalent protocol would be to simply protect both \code{foo} and \code{bar} with the same global lock. So, if all constraints have an equal weight, \toolname could generate both of these protocols, since they would have the same optimum objective value. As mentioned in Section~\ref{sec:impl}, \toolname's MaxSAT encoding prefers assignments where a race between two fragments (like the one on field $x$) are resolved via an atomic field rather than a lock. This forces \toolname to generate the optimal solution for the monitor of Figure~\ref{fig:weight-ablation}. The adjusted weights for other classes of soft constraints try to steer \toolname to better performing synchronization protocols in a similar way.
\end{minipage}

\begin{figure}
    \allCortadoPlots{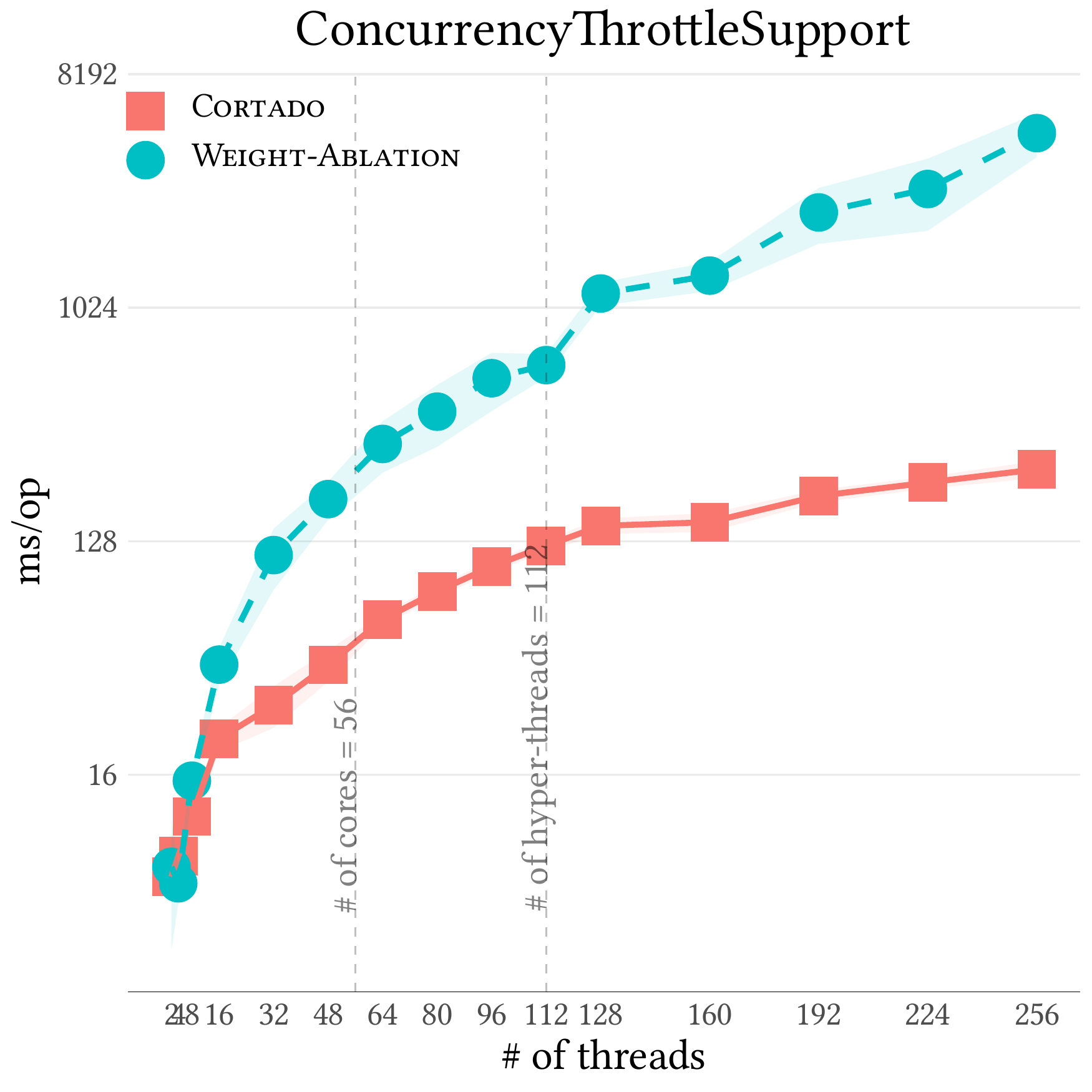}

    \caption{Performance results for the soft constraints weights ablation. The y-axis is in log scale and time measurements are in milliseconds. The shadowed regions surrounding each line present 99.9\% confidence interval of each measurement.}
    \label{fig:weight-ablation}
\end{figure}

%% file: np-complete-proof.tex
\section{Proof of NP-Completeness}\label{ssec:npproof}

To aid the reader, we restate Theorem \ref{thm:np}.
\begin{theorem}
  {\bf{(NP-Completeness)}} Let $\mathcal{G} = (V, E)$ be the FDG
  representation of a monitor $\mtr$ and let $\Pi \subseteq V \times
  V$ be a set of fragment pairs that can safely run in parallel. Then,
  deciding whether there exists a synchronization protocol with at
  most $k$ locks and atomic fields that allows all pairs in $\Pi$ to
  run in parallel is an NP-Complete problem.
\end{theorem}

\begin{proof}[Proof of Theorem \ref{thm:np}]

     We prove the theorem by reduction to the edge clique cover
  problem~\cite{edge-clique-cover}.
  Let $G = (V, E)$ be an undirected graph. For each $v \in V$, let $E(v)$ be the set of edges incident to $v$.
  
  Define monitor $M$ as follows: for each edge $e \in E$ define a new field $f_e$ of the monitor, initially set to zero. For each vertex $v \in V$, define a CCR $inc_v()$ which increments each $f_e$ for $e \in E(V)$, has a guard of $\top$, and returns nothing. 
  
  Let $G_M$ be the control-flow graph of monitor $M$. Note that there is one $waituntil(\top)$ statement for each $inc_v()$ and $|E(v)| = degree(v)$ increment statements in $inc_v()$, so there are $|V| + 2|E|$ total nodes in $G_M$. Define $\{G_{1,M},\dots, G_{|V| + 2|E|,M}\}$ to be the partition of $G_M$ into singletons. Then, let $\mathcal G_M = (V_M, E_M)$ to be the fragment dependency graph obtained from this partition, and let $\Pi \subseteq V_M\times V_M$ be the set of fragment pairs that can safely run in parallel.
  
We write $frag_{v,e}$ for the fragment which increments $f_e$ in method $inc_v()$, and $waituntil_v$ for the $waituntil(\top)$ statement at the beginning of $inc_v()$. Observe that
\begin{align}\label{eq:np:pi}
    \Pi = &\big\{
        (f_1, f_2) \mid
        \exists v \in V \text{ such that }
        f_1 = waituntil_v, \text{ or }
        f_2 = waituntil_v
    \big\}
    \\
    \nonumber
    \cup&\big\{
        (frag_{v_1,e_1}, frag_{v_2,e_2}) \mid
        e_1 \neq e_2
    \big\}.
\end{align}

Note that any synchronization protocol which implements $f_e$ for some $e = (u,v)\in E$ as an atomic variable is equivalent to one which wraps a unique lock around each $frag_{u,e}$ and $frag_{v,e}$. Therefore, we only need to consider synchronization protocols which use only locks.

Suppose we are given a synchronization protocol which allows all pairs in $\Pi$ to run in parallel and uses exactly $k$ locks, $\{\ell_1,\dots,\ell_k\}$, for some $k\in \mathbb N$. Define the vertices holding each lock to be
\begin{equation}\label{eq:np:fragsholdinglock}
    C_i = \big\{
        v \in V \mid \exists e \in E(v) \text{ such that } frag_{v,e} \text{ holds lock }\ell_i
    \big\}.
\end{equation}
We claim that $\{C_1,\dots, C_k\}$ is a clique edge cover of $G$. First, observe that every edge $e =(u,v) \in E$ corresponds to two fragments in $\mathcal G_M$: $frag_{u,e}$ and $frag_{v,e}$. Since these fragments must not run in parallel (due to a data race), they must share some lock. Let $\ell_i$ be that lock. By the definition of $C_i$ in Equation \ref{eq:np:fragsholdinglock}, $u\in C_i$ and $v\in C_i$. Therefore, every edge appears in $C_i$ for some $1\leq i \leq k$. Second, suppose that $u\neq v \in V$ are both contained in $C_i$ for some $i$. By Equation \ref{eq:np:fragsholdinglock}, there must be some $e_u \in E(u)$ and $e_v \in E(v)$ such that both $frag_{v,e_v}$ and $frag_{u,e_u}$ hold lock $i$. Since the two fragments share a lock, we know $(frag_{v,e_v}, frag_{u,e_u})\notin \Pi$. Therefore, $e_v = e_u$ (by Equation \ref{eq:np:pi}). Hence, there is an edge $e_u = e_v = (u, v)$ in $E$. Consequently, any two distinct vertices in $C_i$ are incident, so $C_i$ is a clique.

A symmetric argument shows how to construct a synchronization protocol using exactly $k$ locks from any edge clique-cover of $G$ which has $k$ cliques.

We have shown that, given an arbitrary graph $G$, in polynomial time we may compute a monitor $M$ such that $M$ has a synchronization protocol using at most $k$ locks and atomic variables if and only if $G$ has an edge clique-cover with at most $k$ cliques. 

\end{proof}

%% file: trg-lang-semantics.tex
\section{Target Language Operational Semantics}
\label{sec:trg-lang-sem}

\begin{figure}
\small
\[
\begin{array}{cc}
    (1) &
    \irule{
        \begin{array}{c}
            e = (s, t) \irulespace \textsc{CheckNotif}(\sems, t) \irulespace \neg \mathsf{Synch}(s) \irulespace (s, \args), \pstate \Downarrow \pstate' \irulespace \sems' = \textsc{UpdateState}(\sems, t)
        \end{array}
    }{
        (\pstate, e, \args, \sems) \Rightarrow (\pstate', \epsilon, \epsilon, \sems')
    }\\\\
    (2) &
    \irule{
        \begin{array}{c}
        e = (s, t) \irulespace \textsc{CheckNotif}(\sems, t) \irulespace s = \code{l.lock()} \\
        \textsc{LockHeld}(\sems, t, \pstate[l]) \irulespace \sems' = \textsc{BlockThreadOnLock}(\sems, t, \pstate[l]) \irulespace \textsf{NoDeadLocks}(\sems')
        \end{array}
    }{
        (\pstate, e, \args, \sems) \Rightarrow (\pstate', \epsilon, \epsilon, \sems')
    }\\\\
    (3) &
    \irule{
        \begin{array}{c}
        e = (s, t) \irulespace \textsc{CheckNotif}(\sems, t) \irulespace s = \code{l.lock()} \\
        \neg \textsc{LockHeld}(\sems, t, \pstate[l]) \irulespace \sems' = \textsc{AcqLock}(\sems, t, \pstate[l])
        \end{array}
    }{
        (\pstate, e, \args, \sems) \Rightarrow (\pstate', \epsilon, \epsilon, \sems')
    }\\\\
    (4) &
    \irule{
        \begin{array}{c}
        e = (s, t) \irulespace \textsc{CheckNotif}(\sems, t) \irulespace s = \code{l.unlock()} \irulespace \sems' = \textsc{RelLock}(\sems, t, \pstate[l])
        \end{array}
    }{
        (\pstate, e, \args, \sems) \Rightarrow (\pstate', \epsilon, \epsilon, \sems')
    }\\\\
    (5) &
    \irule{
        \begin{array}{c}
        e = (s, t) \irulespace \textsc{CheckNotif}(\sems, t) \irulespace s = \code{c.await()} \irulespace \sems' = \textsc{BlockOnCVar}(\sems, t, \pstate[c])
        \end{array}
    }{
        (\pstate, e, \args, \sems) \Rightarrow (\pstate', \epsilon, \epsilon, \sems')
    }\\\\
    (6) &
    \irule{
        \begin{array}{c}
        e = (s, t) \irulespace \textsc{CheckNotif}(\sems, t) \irulespace s = \code{c.signal()} \irulespace \sems' = \textsc{SigCVar}(\sems, \pstate[c])
        \end{array}
    }{
        (\pstate, e, \args, \sems) \Rightarrow (\pstate', \epsilon, \epsilon, \sems')
    }\\\\
    (7) &
    \irule{
        \begin{array}{c}
        e = (s, t) \irulespace \textsc{CheckNotif}(\sems, t) \irulespace s = \code{c.signalAll()} \irulespace \sems' = \textsc{BCastCVar}(\sems, \pstate[c])
        \end{array}
    }{
        (\pstate, e, \args, \sems) \Rightarrow (\pstate', \epsilon, \epsilon, \sems')
    }\\\\
    (8) &
    \irule{
        \begin{array}{c}
            (\pstate, e, \args, \sems) \Rightarrow (\pstate', \epsilon, \epsilon, \sems')
        \end{array}
    }{
        (\pstate, e::\tis, \args::\args', \sems) \Rightarrow (\pstate', \tis, \args', \sems')
    }\\\\
\end{array}
\]
\caption{Semantics for our target language~\ref{fig:trg-lang}. Here, \tis is a history, $\args$ a list of arguments for every element in $\tis$, and $\sems$ a tuple of sets and mappings that keep track of all pending signaling and locking operations. Methods and predicates that appear in small caps are defined the text.}\label{fig:trg-sem}
\end{figure}

\begin{figure}
  \begin{algorithm}[H]
    \begin{algorithmic}[1]
    \Procedure{UpdateState}{$\sems, t$}
    \State \textbf{input:} $\sems = \semtupple$
    \State \textbf{input:} $t$, thread executing a non-synchronization statement.
    \State \textbf{output:} updated mappings.
    \vspace{0.04in}
    \If{$(t, c) \in \sigblock$}
        \State $\sigblock' \gets \sigblock \setminus \{ (t, c) \}$ $\signotif' \gets \signotif \setminus \{ (t,c) \}$
    \EndIf
    \If{$(t,l) \in \lockblock$}
       \State $\lockhold' \gets \lockhold \cup \{ (t, l) \}$
       \State $\locknotif' \gets \locknotif \setminus \{ (t, l) \}$
       \State $\lockblock' \gets \lockblock \setminus \{ (t,l) \}$
    \EndIf
    
    \State \Return $(\sigblock', \signotif', \lockhold', \lockblock, \locknotif')$
    \EndProcedure
    \end{algorithmic}
 \end{algorithm}
 \caption{Procedure \textsc{UpdateState}}
 \label{fig:update-state}
\end{figure}

This section presents the semantics of our target language presented in Figure~\ref{fig:trg-lang}. As mentioned in Section~\ref{sec:prelim}, given an explicit monitor $M_t$, initial state $\sigma$, and monitor history $\tis_e$ with argument mapping $\args_e$, the operational semantics of $M_t$ is defined using a judgment $M_t \vdash (\tis_e, \args_{e}, \sigma) \downarrow \sigma'$ indicating that the new  state  is $ \sigma'$ after executing   $\tis_e$ on initial  state $\sigma$. The semantics of such a monitor are implemented using the inference rules of Figure~\ref{fig:trg-sem} that use judgements of the form $(\pstate, \tis_e, \args_e, \sems) \Rightarrow (\pstate', \tis'_e, \args'_e,\sems')$. Here, $\sems$ and $\sems'$ are tuples of the form $\semtupple$ and their role is to keep track all signaling and locking operations of the history. Specifically, each element of the tuple is defined below:
\begin{itemize}
    \item $\sigblock \subseteq T \times CVar$: a set of thread-condition variable pairs, $(t, c) \in \sigblock$ means that thread $t$ is blocked on condition variable $c$.
    \item $\signotif \subseteq T \times CVar$: a set of thread-condition variable pairs, $(t, c) \in \signotif$ means that thread $t$ has been notified on condition $c$.
    \item $\lockhold \subseteq T \times L$: a set of thread-lock pairs, $(t, l) \in \lockhold$ means that thread $t$ holds lock $l$.
    \item $\lockblock \subseteq T \times L$: a set of thread-lock pairs, $(t, l) \in \lockblock$ means that thread $t$ is blocked waiting to acquire lock $l$.
    \item $\locknotif \subseteq T \times L$: a set of thread-lock pairs, $(t, l) \in \locknotif$ means that thread $t$ can acquire a previously held lock $l$
\end{itemize}
We say that $M_t \vdash (\tis_e, \args_{e}, \sigma) \downarrow \sigma'$ if and only if $(\pstate, \tis_e, \args_e, \sems) \Rightarrow^{*} (\pstate', \epsilon, \epsilon,\sems')$, where $\Rightarrow^{*}$ is the reflexive transitive closure of relation $\Rightarrow$. In other words, a $\tis_e$ is a valid explicit history according to our operational semantics only if the rules of Figure~\ref{fig:trg-sem} can ``consume'' the entire history. If none of the rules of Figure~\ref{fig:trg-sem} apply to a history, then we consider the computation of relation $\Rightarrow$ stuck and thus the history is not valid.

On a high level, the rules of our operational semantics iterate over all statements of input history $\tis$ and updates the sets inside $\sems$ accordingly. Because during the execution of a $\tis$ a thread $t$ might perform a blocking operation (e.g., call \code{l.lock} on a state where \code{l} is being held), the rules require every statement to be executed in a state where a thread is not blocked. To ensure this, every rule in Figure~\ref{fig:trg-sem} requires predicate $M_t \vdash (\tis_e, \args_{e}, \sigma) \downarrow \sigma'$, defined below, to hold for the executing thread $t$.
\begin{equation*}
    \begin{split}
        \textsc{CheckNotif}(\semtupple, t) = &\ \ (t,c) \in \sigblock \leftrightarrow (t,c) \in \signotif\ \ \land\\  
        &\ \ (t, l) \in \lockblock \leftrightarrow (t, l) \in \locknotif
    \end{split}
\end{equation*}
Essentially, $\textsc{CheckNotif}(\semtupple, t)$ requires every thread $t$ that was previously blocked by some operation ($(t,\_) \in \sigblock$ or $(t, l) \in \lockblock$) to be first notified ($(t,\_) \in \signotif$ or $ (t, l) \in \locknotif$) in order to execute a statement. 

In what follows, we explain each of the rules of Figure~\ref{fig:trg-sem} in more detail.

\paragraph{Rule (1).} This rule applies for all statements that are not a synchronization statement (i.e., lock or signal operation). Because the operational semantics for non-synchronization statements are well-studied, we assume the existence of an oracle $\Downarrow$ that give a statement $s$ and its argument $\args$, it returns the resulting monitor state $\pstate'$. Furthermore, because thread $t$ could be blocked before executing statement $s$, this rule uses procedure $\textsc{UpdateState}$ (defined in Figure~\ref{fig:update-state}) to update the sets inside $\sems$ accordingly. Specifically, if thread $t$ was blocked in some condition variable $c$, then procedure $\textsc{UpdateState}$ removes pair $(t,c)$ from both $\sigblock$ and $\signotif$ (recall that if $t$ was blocked then it is guaranteed to be notified). Similarly, if thread $t$ was blocked on some lock $l$, then procedure $\textsc{UpdateState}$  add the pair $(t,l)$ to $\lockhold$ (i.e., now $t$ holds lock $l$) and removes it from $\locknotif$ and $\lockblock$ (same as in the condition variable case).

\paragraph{Rule (2).} This rule applies to all statements where a thread $t$ attempts to acquire lock \code{l} that is currently held by another thread. In order to determine whether a lock is held by another thread, this rule makes use of predicate $\textsc{LockHeld}$ defined as follows:
\begin{equation*}
 \begin{split}
    \textsc{LockHeld}(\semtupple, t, l) = l \in \lockhold[t'].\ t \neq t'
  \end{split}
\end{equation*}
Then, the rule marks thread $t$ as blocked on lock $l$ by using the following procedure that updates map $\lockblock$ by adding pair $(t,l)$:
\begin{equation*}
 \begin{split}
    \textsc{BlockThreadOnLock}(\semtupple, t, l) = (\sigblock, \signotif, \lockhold, \lockblock', \locknotif)\\
    \textrm {where\ } \lockblock' = \lockblock \setminus \{ (t,l) \}
  \end{split}
\end{equation*}
Finally, the rule requires that the new attempt to acquire lock $l$ does not introduce any deadlocks by invoking oracle $\textsf{NoDeadLocks}$. This oracle detects any cycles in the lock acquisition by examining maps $\lockblock$ and $\lockhold$. 

\paragraph{Rule (3).} Conversely, the third rule applies to all cases where thread thread $t$ attempts to acquire a lock not currently held by some other thread. In this case, the rule simply adds pair $(t,l)$ in map $\lockhold$ as follows:
\begin{equation*}
 \begin{split}
    \textsc{AcqLock}(\semtupple, t, l) = (\sigblock, \signotif, \lockhold', \lockblock, \locknotif)\\
    \textrm {where\ } \lockhold' = \lockhold \cup \{ (t, l) \}
  \end{split}
\end{equation*}

\paragraph{Rule (4).} This rule is triggered when a thread $t$ releases lock $l$. The rule performs the following two updates to maps $\lockhold$ and $\locknotif$:
\begin{equation*}
 \begin{split}
    \textsc{RelLock}(\semtupple, t, l) = (\sigblock, \signotif, \lockhold', \lockblock, \locknotif')\\
    \textrm {where\ } \lockhold' = \lockhold \setminus \{ (t, l) \},\ \locknotif' = \locknotif \cup \{ (t',l) \} \textrm{ s.t.\ } (t',l) \in \lockblock
  \end{split}
\end{equation*}
Specifically, it removes pair $(t,l)$ from $\lockhold$ and notifies some thread $t'$ currently blocked on lock $l$.

\paragraph{Rule (5).} This rule applies when a thread $t$ calls method \code{await} on a condition variable $c$. The rule simply adds pair $(t,c)$ in set $\sigblock$.
\begin{equation*}
 \begin{split}
    \textsc{BlockOnCVar}(\semtupple, t, c) = (\sigblock', \signotif, \lockhold, \lockblock, \locknotif)\\
    \textrm {where\ } \sigblock' = \sigblock \setminus \{ (t, c) \}
  \end{split}
\end{equation*}

\paragraph{Rules (6) and (7).} These two rules are used when a thread signals or broadcasts a  condition variable $c$. They simply update set $\signotif$ as follows:
\begin{equation*}
 \begin{split}
    \textsc{SigCVar}(\semtupple, c) = (\sigblock, \signotif', \lockhold, \lockblock, \locknotif)\\
    \textrm {where\ } \signotif' = \signotif \cup \{ (t,c) \} \textrm{\ s.t.\ } (t,c) \in \sigblock
  \end{split}
\end{equation*}
\begin{equation*}
 \begin{split}
    \textsc{BcastCVar}(\semtupple, c) = (\sigblock, \signotif', \lockhold, \lockblock, \locknotif)\\
    \textrm {where\ } \signotif' = \signotif \cup \{ (t,c) \mid (t,c) \in \sigblock \}
  \end{split}
\end{equation*}
Specifically, rule 6 adds \emph{a single} thread $t$ currently blocked on condition variable $c$ in $\signotif$, whereas rule 7 adds \emph{all} such threads in $\signotif$. 

\paragraph{Rule (8).} Finally, rule 8 recursively applies the procedure to the whole input history $\tis$. 

%% file: mtr-instrumentation.tex
\section{Monitor Instrumentation.}
\label{appendix:mtr-instr}

In this section, we describe procedure \textsf{Instrument} which given an implicit-synchronization monitor $\mtr$, it corresponding FGD $\fdg = (V,E)$, and a synchronization protocol $\mathcal{S} = (\lockmap, \atomfld, \predmap)$, it instrument protocol $\mathcal{S}$ into $\mtr$ yielding an explicit-synchronization monitor $\mtr'$ equivalent to $\mtr$. This is achieved by first introducing all the necessary synchronization fields (locks, condition variables, and atomic fields) in the input class and then instrumenting locking and signaling operations in all methods as follows:  
\begin{itemize}[leftmargin=*]
    \item {\bf Lock acquisition and release:} The synthesized code must ensure that all the locks in $\lockmap(f)$ are held when executing fragment $f$. Thus, for every edge $(f, f')$ in the FDG, we instrument the code to acquire locks $\lockmap(f') \backslash \lockmap(f)$ and release locks $\lockmap(f) \backslash \lockmap(f')$. Furthermore, as  mentioned in Section~\ref{sec:overview}, we acquire and release these locks according to a static total order  to prevent deadlocks.
    \item {\bf Blocking on predicates:} Our instrumentation must also convert every \code{waituntil} statement to a sequence of operations on locks and condition variables. Specifically, we instrument a \code{waituntil(p)} statement as follows:
    {\small 
    \begin{verbatim}
 while(!p) {  ln.unlock(); ...; l2.unlock(); c.await(); l2.lock(); ...; ln.lock(); }
    \end{verbatim}
    }
    
    where \code{c} is the condition variable associated with \code{p};  \code{l1, ... ln} are the locks associated with this fragment, and \code{l1} is the lock associated with condition variable \code{c}. 
    \item {\bf Signaling operations:} Finally, we instrument signaling operations introduced by \textsf{PlaceSignals} to acquire and release the appropriate locks. In particular, given a statement \code{signal(p,c)} (similarly for \code{broadcast(p,c)}), our instrumentation generates the following code:
    
    {\small 
    \begin{verbatim}
     if (c)  { lp.lock(); cp.signal(); lp.release(); }
     \end{verbatim}
     }
     where \code{cp} is the condition variable for predicate \code{p} and  \code{lp} is the corresponding lock for \code{cp}.

    \end{itemize}

Procedure \textsf{Instrument} is presented in Figure~\ref{fig:instrumentation} in the form of inference rules that use the following two judgements:
\begin{itemize}
    \item $\args \vdash \Delta \rightsquigarrow \Delta'$, where $\args$ is a subset of the arguments of procedure \textsf{Instrument} (we overload operator $\rightsquigarrow$ depending on the arguments) and $\Delta$ is one of the following: the input monitor, a field, a method, a CCR or a statement.
    \item $\lockmap, \fdg \vdash v \hookrightarrow v'$, where $\lockmap$ is the lock map of the input synchronization protocol $\mathcal{S}$, FDG is the input $\fdg$, and $v$ is a fragment in $\fdg$.
\end{itemize}
The meaning of each judgement is that whenever procedure \textsf{Instrument} is applied to an element that appears on the left-hand side of an arrow ($\rightsquigarrow,\hookrightarrow$) it generates the element on the right-hand side.

\paragraph{\textbf{Overall Structure.}} The core logic of this procedure is to recursively iterate every element of the input monitor and use the inferred synchronization protocol in order to convert each element to an equivalent element of the target language. At a top level, the procedure begins by transforming every field and method of the input monitor. For every method, the procedure recursively visits every CCR using operator $\rightsquigarrow$. Then, for every CCR, it collects all its fragments and uses operator $\hookrightarrow$ to instrument all the lock operations dictated by the input protocol. In what follows, we give a brief description of every rule presented in Figure~\ref{fig:instrumentation}.

\paragraph{\textsc{Mtr}.} This is the top-level rule called by procedure \textsf{Instrument} and performs the following tasks: 1. it introduces all the synchronization fields (locks and condition variables) needed by the synchronization protocol and initializes them accordingly and 2. it recursively calls itself for every field, and method of $\mtr$.

\paragraph{\textsc{Fld-1} \& \textsc{Fld-2}.} These two rules are used to translate fields of $\mtr$, with the first one being applicable to fields that must be converted to atomic fields and the second one to fields that should remain the same. Only the first rule alters the original field by converting to an atomic field with the same name as the original.

\paragraph{\textsc{Method} \& \textsc{CCR}.} These two rules simply recursively apply operator $\rightsquigarrow$ to their constituent elements.

\paragraph{\textsc{CCR-Statement}.}
This rule is the one that splits each top-level CCR-Statement to a set of fragments that belong in the input FGG $\fdg$ and then recursively transforms each of the fragments. Note, because of the properties of FDG (Definitions 4.1 and 4.2), there is only one way to decompose each CCR to its constituents fragments.

\paragraph{\textsc{Frag-Stmt}.} This rule applies to all statements $s$ that are a fragment in the input $\fdg$. It first uses operator $\hookrightarrow$ to instrument all necessary lock operations and then uses a special oracle $\rightarrow_{\mathcal{A}}$ that converts all operations that involve a field converted to atomic to the equivalent $\code{update}$ statement in the target language.\footnote{Due to its simplicity, we omit a formal description of oracle $\rightarrow_{\mathcal{A}}$.}

\paragraph{\textsc{Wait}.} Rule labeled \textsc{Wait} is a special case of the above rule because, by definition, every \wuntil statement defines its own fragment in an FDG. Similar to the rule above, this rule also uses operator $\hookrightarrow$ to instrument the appropriate lock operation in the fragment but it additionally translates the \wuntil statement into an equivalent statement in the target language that uses condition variables. As mentioned in Section~\ref{sec:algorithm}, each \wuntil statement is translated into a while loop that waits on the appropriate condition variable and properly releases and acquires all locks before and after the call to method \code{await}. Additionally, it acquires all locks needed by its successor statement $v$ in the FDG and releases all locks held by it but not needed by $v$ (similar to the logic described below).

\paragraph{\textsc{Sig}.} This rule applies to all fragments that are a signalling directive of the monitor's intermediate representation.\footnote{For simplicity, we assume that every signaling operation defines its own fragment.} In a similar manner as the rule for \wuntil statements, this rule first uses operator $\hookrightarrow$ to instrument all lock operations needed to implement the synchronization protocol. Then, it consults the predicate map $\predmap$ of the synchronization protocol to acquire the appropriate lock and perform the signaling operation on the associated condition variable.

\paragraph{\textbf{Instrumenting Fragments With Lock Operations.}} Finally, we describe operator $\hookrightarrow$ which given a fragment $v$, the lock map $\lockmap$ of the input synchronization protocol $\mathcal{S}$, and the FDG $\fdg$, it instruments all the necessary lock operations. The logic of this operator is split between two groups of rules, described in more detail below:
\begin{itemize}
    \item Rules for entry \& exit fragments (i.e., fragments without predecessors and successors respectively), which are handled by rules \textsc{Entry-Frag} and \textsc{Exit-Frag} respectively. These rules simply lookup the entry or exit fragment in $\lockmap$ and acquire or release the locks returned by the $\lockmap$ accordingly.
    \item Rules for fragments with successors. Fragments that contain some successor in the graph are handled by rules $\textsc{Branch-Frag}$, $\textsc{Reg-Frag-1}$ and $\textsc{Reg-Frag-2}$. The logic for each of these rule is similar, i.e., for any successor fragment $v_s$ of fragment $v$, the instrumentation releases all locks required by $v$ but not by $v_s$ ($\lockmap[v] \setminus \lockmap[v_s]$) and acquires all locks required by $v_s$ but not $v$ ($\lockmap[v_s] \setminus \lockmap[v]$). All these operation are operation in accordance to the global lock order to prevent deadlocks. Last, it is worth mentioning that the main difference of these three rules is how they instrument the edge between $v$ and $v_s$. That is, if $v$ ends with a \code{goto} statement (conditional or not), then the instrumentation redirects the control-flow appropriately so all lock operations occur along edge $(v,v_s)$.
\end{itemize}

Finally, we conclude with the following theorem that states the correctness of our instrumentation phase.

\begin{theorem}\label{thm:correct-instr}
 Let  $\mathcal{S} = (\lockmap, \atomfld, \predmap)$ be a synchronization protocol inferred over FGD $\mathcal{G} = (V,E)$ of input monitor $M$ and $M'$ be the result of procedure $\textsf{Instrument}$ for $M$. Then, the following three conditions hold:
 \begin{enumerate}
     \item For every fragment $v \in V$, $l_i \in \lockmap[v]$ iff fragment $v$ holds lock $l_i$ in $M'$
     \item If $i < j$, then $l_i$ is never acquired whenever $l_j$ is held.
     \item Field $f \in \atomfld$ iff all its occurrences in $M$ have been replaced with an atomic operation in $M'$.
 \end{enumerate}
\end{theorem}
\begin{proof}
Proof can be find in Appendix~\ref{sec:correct-proofs}.
\end{proof}

\begin{figure}
\small
  \[
  \begin{array}{clcl}
    \fbox{\textsc{{\footnotesize Entry-Frag}}} &
    \irule{
      \neg\exists v_p.\ (v_p, v) \in E \irulespace A = \lockmap[v]
    }{
      \lockmap, (F, E) \vdash v \hookrightarrow Acq(A);v
    } &
    \fbox{\textsc{{\footnotesize Exit-Frag}}} &
    \irule{
      \neg\exists v_s.\ (v,v_s) \in E \irulespace R = \lockmap[v]
    }{
      \lockmap, (F, E) \vdash v \hookrightarrow v;Rel(R)
    }\\\\
    \fbox{\textsc{{\footnotesize Branch-Frag}}} &
    \multicolumn{3}{c}{
      \irule{
        \begin{array}{c}
          Exit(v) \equiv \code{if (c) goto l} \irulespace (v, v_{s1)} \in E \irulespace (v,v_{s2}) \in E \irulespace v_{s1} \neq v_{s2} \irulespace v_{s2} \equiv \code{l: s}\\
          (A_i, R_i) = (\lockmap[v_{si}] \setminus \lockmap[v], \lockmap[v] \setminus \lockmap[v_{si}]),\ i \in \{1,2\} \irulespace v' = v[\code{l'}/\code{l}]\\
          e_1 = (Rel(R_1);Acq(A_1);\code{goto l''}) \irulespace e_2 = (\code{l'}: Rel(R_2);Acq(A_2);\code{goto l})
        \end{array}
      }{
        \lockmap, (F, E) \vdash v \hookrightarrow v';e_1;e_2;\code{l'': skip}
      }
    }\\\\
    \fbox{\textsc{{\footnotesize Reg-Frag-1}}} &
    \multicolumn{3}{c}{
    \irule{
      \begin{array}{c}
        Exit(v) \equiv \code{goto l} \irulespace (v, v_s) \in E \irulespace (A, R) = (\lockmap[v_s] \setminus \lockmap[v], \lockmap[v] \setminus \lockmap[v_s])\\
        v' = v[\code{l'}/\code{l}];(\code{l'}: Rel(R);Acq(A);\code{goto l})
      \end{array}
    }{
      \lockmap, (F, E) \vdash v \hookrightarrow v'
    }}\\\\
    \fbox{\textsc{{\footnotesize Reg-Frag-2}}} &
    \multicolumn{3}{c}{
    \irule{
      \begin{array}{c}
        Exit(v) \nequiv \code{goto l} \irulespace (v, v_s) \in E \irulespace (A, R) = (\lockmap[v_s] \setminus \lockmap[v] , \lockmap[v] \setminus \lockmap[v_s])\\
      \end{array}
    }{
      \lockmap, (F, E) \vdash v \hookrightarrow v;Rel(R);Acq(A)
    }
    }\\\\
    \fbox{\textsc{{\footnotesize Wait}}} &
    \multicolumn{3}{c}{
      \irule{
        \begin{array}{c}
          \lockmap, (F,E) \vdash w \hookrightarrow w' \irulespace (\code{c},\code{e}) = NewLabels() \irulespace l_p = \predmap[w]\\
          L_w = \lockmap[w] \irulespace rel = Rel(L_w \setminus \{l_p\}) \irulespace acq = Acq(L_w \setminus \{l_p\})\\
          (w, v) \in E \irulespace (A, R) = (\lockmap[w] \setminus \lockmap[v], \lockmap[v] \setminus \lockmap[w]) \irulespace succLocks \equiv Rel(R);Acq(A)\\
        w'' = w'[\code{(c: if (p) goto e);(rel;c$_p$.await();acq;goto c);(e: succLocks)}/w]
        \end{array}
      }{
        (\lockmap, \atomfld, \predmap), (F,E) \vdash w \equiv \wuntil(p) \rightsquigarrow w''
      }
    }\\\\
    \fbox{\textsc{{\footnotesize Frag-Stmt}}} &
    \irule{
      \begin{array}{c}
        IsFrag(s) \irulespace \lockmap,\fdg \vdash s \hookrightarrow s' \irulespace s' \rightarrow_{\atomfld} s''
      \end{array}
    }{
      (\lockmap, \atomfld, \predmap), \fdg \vdash s \rightsquigarrow s''
    } &
    \fbox{\textsc{{\footnotesize Sig}}} &
    \irule{
      \begin{array}{c}
        SigOp(s) \irulespace  \lockmap, \fdg \vdash s \hookrightarrow s' \\ s'' = s'[ExplSig(s, \predmap)/s]
      \end{array}
    }{
      (\lockmap, \atomfld, \predmap), \fdg \vdash s \rightsquigarrow s'
    }\\\\
    \fbox{\textsc{{\footnotesize CCR-Statement}}} &
    \multicolumn{3}{c}{
    \irule{
      \begin{array}{c}
        \fdg \equiv (F,E) \irulespace s \equiv s_1\code{;}...\code{;}s_n \irulespace s_i \in F \irulespace s_i \rightsquigarrow s_i'
      \end{array}
    }{
      \mathcal{S}, \fdg \vdash s \rightsquigarrow s_1'\code{;}...\code{;}s_n'
    }}
    \\\\
    \fbox{\textsc{{\footnotesize Method}}} &
    \irule{
      \begin{array}{c}
        \mathcal{S}, \fdg \vdash c_i \rightsquigarrow c_i'\\
      \end{array}
    }{
      \mathcal{S}, \fdg \vdash \mathit{m}(\vec{v})\code {\{}  c_1 ... c_n \code{\}} \rightsquigarrow \mathit{m}(\vec{v})\code {\{}  c_1' ... c_n' \code{\}}
    } &
    \fbox{\textsc{{\footnotesize CCR}}} &
    \irule{
      \begin{array}{c}
        \mathcal{S}, \fdg \vdash s \rightsquigarrow s'\\
        \mathcal{S}, \fdg \vdash w \rightsquigarrow w'
      \end{array}
    }{
      \mathcal{S}, \fdg \vdash w\code{;}s \rightsquigarrow w';s'
    }\\\\
    \fbox{\textsc{{\footnotesize Fld-1}}} &
    \irule{
      \code{fld} \in \atomfld
    }{
      \atomfld \vdash \tau\ \code{fld} := e \rightsquigarrow \code{Atomic[$\tau$] fld} := e
    } &
    \fbox{\textsc{{\footnotesize Fld-2}}} &
    \irule{
      \code{fld} \notin \atomfld
    }{
      \atomfld \vdash \tau\ \code{fld} := e \rightsquigarrow \tau\ \code{fld} := e
    }\\\\
    \multicolumn{4}{c}{
      \fbox{\textsc{{\footnotesize Mtr}}} \ \
      \irulespace
      \irule{
        \begin{array}{c}
          \mathcal{S} \equiv (\lockmap, \atomfld, \predmap) \irulespace\fdg \equiv (F, E) \irulespace l_i \triangleq \code{Lock $l_j$ := new Lock()}\ s.t.\ l_j \in Locks(\lockmap)\\
          cv_i \triangleq \code{CondVar cv$_p$ := $l_j$.newCV()}\ s.t.\ (p, l_j) \in \predmap \irulespace \atomfld \vdash f_i \rightsquigarrow f_i' \irulespace \mathcal{S}, \fdg \vdash m_i \rightsquigarrow m_i'
        \end{array}
        
      }{
        \mathcal{S}, \fdg \vdash \code{mtr $\mtr$ \{ } f_1 ... f_m\ \ m_1 ... m_n \code{ \}} \rightsquigarrow
                                            \code{mtr $\mtr$ \{ } l_1 ... l_k\ cv_1 ... cv_l\ f_1' ... f_m'\ m_1' ... m_n' \code{ \}}
      }
    }\\\\
    \hline
    \\
    \fbox{\textsc{{\footnotesize Aux-Defs}}} &
    \multicolumn{3}{c}{
      Acq(L) \triangleq l_{i_1}\code{.lock()};...;l_{i_k}\code{.lock()}\ s.t.\ l_{i_j} \in L, \forall j.\ 1 \leq j \leq k, i_j < i_{j+1}
    }\\\\
    & \multicolumn{3}{c}{
      Rel(L) \triangleq l_{i_k}\code{.unlock()};...;l_{i_1}\code{.unlock()}\ s.t.\ l_{i_j} \in L, \forall j.\ 1 \leq j \leq k, i_j < i_{j+1}
    }\\\\
    \multicolumn{4}{c}{
      ExplSig(s, \predmap) = \begin{cases}
                              \code{if (c) \{ l.lock();}\code{c$_p$.signal();}\code{l.unlock();\}} & s \equiv signal(p, c), l \equiv \predmap[p]\\
                              \code{if (c) \{ l.lock();}\code{c$_p$.signalAll();}\code{l.unlock();\}} & s \equiv bcast(p, c), l \equiv \predmap[p]
                              \end{cases}
    }
  \end{array}
  \]
  \caption{Procedure \textsf{Instrument}($\mathcal{S}, \fdg$), where $\mathcal{S}$ is a synchronization protocol for FDG $\fdg$.}
  \label{fig:instrumentation}
\end{figure}

%% file: correctness-proof.tex
\section{Correctness-Related Proofs}
\label{sec:correct-proofs}

This section contains all the proofs related to the correctness of our approach. Section~\ref{sec:top-level-proof} presents the proof of Theorem~\ref{thm:top-level-correctness}, Section~\ref{sec:safe-ileave-proof} presents the proof of Theorem~\ref{thm:safe}, Section~\ref{sec:maxsat-proof} presents the proof of Theorem~\ref{thm:maxsat}, and Section~\ref{sec:instr-proof} presents the proof of Theorem~\ref{thm:correct-instr}.

\subsection{Proof of Theorem \ref{thm:top-level-correctness}}
\label{sec:top-level-proof}

Theorem~\ref{thm:top-level-correctness} states the following:

\paragraph{\textbf{(Correctness)}}
 We say that an explicit monitor $M_t$ correctly implements an implicit monitor $M_s$, denoted as $\mtr_s \sim \mtr_t$, iff for all input states $\sigma_s, \sigma_t$ s.t. $\sigma_s \equiv_{\mtr_{s}} \sigma_t$, we have:
\begin{enumerate}[leftmargin=*]
\setlength\itemsep{0.5em}
    \item $\forall \tis_i, \args_i.\ \  \mtr_s \vdash (\tis_i, \args_i, \pstate_s) \Downarrow \pstate_s' \Longrightarrow \left(\mtr_t \vdash (\concr{\mtr_t}(\tis_i, \args_i, \pstate_s), \pstate_t) \downarrow \pstate_t' \ \land \  \pstate_s' \equiv_{\mtr_s} \pstate_t'\right)$
    \item $\forall \tis_e, \args_e.\ \  \mtr_t \vdash (\tis_e, \args_e, \pstate_t) \downarrow \pstate_t' \Longrightarrow \left( \exists \tis_i,\args_i.\ (\tis_e, \args_e) \backsim (\tis_i, \args_i) \land \mtr_s \vdash (\tis_i, \args_i, \pstate_s) \Downarrow \pstate_s' \land \pstate_s' \equiv_{\mtr_s} \pstate_t' \right)$
\end{enumerate}  
\begin{proof}
    For all proofs in this Section, we assume the correctness of procedure \textsc{PlaceSingals} (proved in previous work~\cite{expresso}).
    
    The proof of condition (1) above follows directly from Theorems~\ref{thm:maxsat},~\ref{thm:correct-instr}, and the correctness of procedure \textsc{PlaceSingals}. The proof of condition (2) follows directly from Theorem~\ref{thm:safe}, Theorem~\ref{thm:maxsat}, Theorem~\ref{thm:correct-instr}, and correctness of \textsc{PlaceSignals}.
\end{proof}

\subsection{Proof of Theorem \ref{thm:safe}}
\label{sec:safe-ileave-proof}

In this section, we present the proof of Theorem~\ref{thm:safe} which we reiterate here for convenience.

Before presenting the actual proof, we first introduce some auxiliary notation, relations, and lemmas. First, given a history $\tis$, we define a predicate  $\tis\tbr{(v_1, t_1)_{i1}, \ldots, (v_k, t_k)_{ik}}$ that evaluates to true iff each event $(v_i, t_i)_{i}$ is $i$-th element in $\tis$ and $i1 < \ldots < ik$. That is, this predicate encodes that these event occur in this particular order within $\tis$. Second, we define $\emph{Next}(\tis, i, t)$ as follows: $min(\{ j \mid j > i,\ \tis_{\fdg}\tbr{(\_, t)_{j}}\})$. In other words, $\emph{Next}(\tis, i, t)$ returns the first element in $\tis$ after index $i$ whose thread identifier is $t$. Additionally, we use $\tis[i]$ to denote the $i$-th element in $\tis$ and $\tis[i:j],\ i < j$ to denote the ``sub-history'' of $\tis$ between its $i$-th  element (inclusive) and $j$-th element (exclusive). Finally, we extend the definition of a history projection to filter out elements that \emph{do not} involve a thread , e.g., $\Pi(\tis, \neg t)$ filters out all events of $\tis$ that involve thread $t$. 

Next, using the notation above we define some relations that identify interleavings inside a history of a fragmented monitor $\mtr_{\fdg}$.

\begin{definition}{\bf (History Interleaving).} Given history $\tis_{fdg} = (V, E)$ and interleaving $\chi = (v, e = (v_s, v_t))$. We define the occurrences of $\chi$ as follows:
\[
Interleavings(\chi, \tis_{\fdg}) = [ (j, (i,k)) \mid \tis_{\fdg}\tbr{(v_s, t)_i, (v, t')_j, (v_t, t)_k},\ t \neq t',\ k = \emph{Next}(\tis_{\fdg}, i, t)]
\]
Also, we write $\mathcal{X}(\tis_{\fdg})$ to be the set of all interleavings that occur in $\tis_{\fdg}$, i.e.
\[
\mathcal{X}(\tis_{\fdg}) = \{ \chi \mid X = Interleavings(\chi, \tis_{\fdg}), |X| > 0 \}
\]
Finally, we write $\mathcal{X}_{\#}(\tis_{\fdg})$ to denote the \emph{number} of interleavings inside $\tis$. Formally:
\[
\mathcal{X}_{\#}(\tis_{\fdg}) = \sum\limits_{\chi \in V \times E}|Interleavings(\chi, \tis_{\fdg})|
\]
\end{definition}

Next, given a fragment $v$ we assume the existence of two predicates, namely, $\emph{EntryFrag}$ and $\emph{ExitFrag}$, that hold only if $v$ is the entry fragment or the exit fragment of its CCR respectively. Based on these relations, we define the next relation that partitions a fragment history into CCR sub-histories.

\paragraph{\textbf{CCR History Partition.}} Let $\tis_{\fdg}$ be a history of fragments in FDG $\fdg$. We define the CCR partition of $\tis_{\fdg}$ that returns a list of potentially overlapping sub-histories of $\tis_{\fdg}$ as follows:
\[
 CCRPart(\tis_{\fdg}) = \left[ \tis_{\fdg}[i:j] \bigm \vert \begin{array}{c}
      \tis_{\fdg}[i] = (v_{in}, t, \_),\  EntryFrag(v_{in}),\ \\
      j = min(\{ k \mid k > i,\ \tis_{\fdg}[k+1] = (v_{out}, t, \_), ExitFrag(v_{out})\})
 \end{array} \right]
\]
Let $P = CCRPart(\tis_{\fdg})$, then we use $P[i]$ to refer to the $i$-th sub-history in $P$. Note we assume that partitions returned by $CCRPart$ are ordered according to the index of the first element in the sub-history. That is, if $ccr_1$ began its execution before $ccr_2$ in $\tis_{\fdg}$, then the partition of $ccr_1$ appears before the partition of $ccr_2$ in $P$. Furthermore, given a CCR partition $P[i]$, we write $Thread(P[i])$ to represent the thread of the first element in sub-history $P[i]$.

\paragraph{\textbf{Removing Interleavings from Histories}.} Before we prove our main theorem, we define some transformations on interleaved histories that helps us remove interleavings.

First, given a CCR sub-history that is interleaved, we define its sequential history as follows:

\begin{definition}{\textbf{Sequential CCR Sub-history}} Let $\tis_{\fdg}$ be an interleaved history, $P = CCRPart(\tis_{\fdg})$, and $\tis'_{\fdg} = P[i]$ be an interleaved CCR partition. We define that the sequential history of $\tis'_{\fdg}$, denoted as $Seq_{|CCR}(\tis'_{\fdg})$ to be the following history: $\Pi(\tis'_{\fdg}, t)\Pi(\tis'_{\fdg}, \neg t)$, where $t = Thread(P[i])$. Given an argument mapping $\args$ for history $\tis'_{\fdg}$, we write $Seq(\args)$ to denote the corresponding argument mapping for $Seq_{|CCR}(\tis'_{\fdg})$.
\end{definition}

We now prove the following useful lemmas about sequential CCR sub-histories.

\begin{lemma}\label{lem:decr-ileavings}
  Let $\tis_{\fdg}$ be an interleaved sub-history and $\tis'_{\fdg} = Seq_{|CCR}(\tis_{\fdg})$, then the following two things hold:
  \begin{enumerate}
    \item $\mathcal{X}(\tis'_{\fdg}) \subseteq \mathcal{X}(\tis_{\fdg})$
    \item $\mathcal{X}_{\#}(\tis'_{\fdg}) < \mathcal{X}_{\#}(\tis_{\fdg})$
  \end{enumerate}
\end{lemma}

\begin{proof}
 Both of the properties logically follow from the construction of $\tis'_{\fdg}$. That is, a sequential CCR sub-history of the form $(v_1,t)\ldots(v_i,t),(v_{i+1},t'),\ldots,(v_j,t'')$, where all elements before the $i$-th position are from thread $t$ and all elements past that are from some thread $t'$ s.t. $t \neq t'$. On the other hand, the original history $\tis_{\fdg}$ is of the form:
 \[
 (v_1, t)\ldots(v_k, t)(v_{k+1}, t')\ldots(v_j,t)
 \]
 Where $\tis_{\fdg}[1:k+1] = \tis'_{\fdg}[1:k+1]$ (i.e., $\tis_{\fdg}$ and $\tis'_{\fdg}$ have a common prefix). Therefore, since by its construction $\tis'_{\fdg}$ does not move the relative order of element is $\tis_{\fdg}$ that do not involve thread $t$, if an interleaving $\chi \in \mathcal{X}(\tis'_{\fdg})$ then we also have $\chi \in \mathcal{X}(\tis_{\fdg})$. Conversely, any thread interleaving that involved thread $t$ in $\tis_{\fdg}$ does not appear in $\tis_{\fdg}$ (by construction). Since by its definition the interleaved history $\tis_{\fdg}$ contains at least one interleaving that involves an edge executed by $t$, we can conclude that $\mathcal{X}_{\#}(\tis'_{\fdg}) < \mathcal{X}_{\#}(\tis_{\fdg})$.
\end{proof}

\begin{lemma}\label{lemma:equiv-state}
  Let $\tis_{\fdg}$ be an interleaved sub-history and $\tis'_{\fdg} = Seq_{|CCR}(\tis_{\fdg})$, then if $\mathcal{X}(\tis_{\fdg})$ is a set of \emph{strongly} safe interleavings we have that: $\forall \pstate,\args. \mtr_{\fdg} \vdash (\tis_{\fdg}, \args, \pstate) \Downarrow \pstate' \Rightarrow \mtr_{\fdg} \vdash (\tis'_{\fdg}, Seq(\args), \pstate) \Downarrow \pstate'$
\end{lemma}

\begin{proof}
  We prove this by induction on the number of distinct interleavings of history $\tis_{\fdg}$ ($\mathcal{X}_{\#}{\tis_{\fdg}}$).
  
  \paragraph{Base Case: $\mathcal{X}_{\#}({\tis_{\fdg}}) = 1$.} If there is a single interleaving in $\tis_{\fdg}$, this implies that $\tis_{\fdg}$ is of the form:
  \[
    (v_1, t)\ldots(v_i, t')\ldots(v_j,t)
  \]
  where $t = Thread(\tis_{\fdg})$ and $(v_i, t')$ is the only element in $\tis_{\fdg}$ not executed by $t$.  Because, the interleaving of $\tis_{\fdg}$ is strongly safe, we have that fragment $v_i$ executed by $t'$, \emph{right commutes with any possible successor of the edge it interleaves}. Also, by definition, $\tis'_{\fdg}$ is $(v_1, t)\ldots(v_{j},t)(v_i, t')$. Combining this two facts with lemma~\ref{lem:decr-ileavings}, we can prove the theorem for our base case:
  \[
  \forall \pstate,\args. \mtr_{\fdg} \vdash (\tis_{\fdg}, \args, \pstate) \Downarrow \pstate' \rightarrow \mtr_{\fdg} \vdash (\tis'_{\fdg}, Seq(\args), \pstate) \Downarrow \pstate'
  \]
  
  \paragraph{Inductive Step.} In our inductive step, we assume that our lemma holds for $\mathcal{X}_{\#}({\tis_{\fdg}}) = n$ and we are going to prove it for $n+1$. The logic is similar to the base case, specifically, we get the right-most interleaved fragment in $\tis_{\fdg}$ and right-commute to the end of the history while obtaining a semantically equivalent history $\tis''_{\fdg}$. After that, we can apply our inductive hypothesis on $\tis''_{\fdg}$, which again proves our goal. 
\end{proof}

Finally, we prove our main theorem, which we re-iterate below for convenience.

\paragraph{\textbf{Theorem~\ref{thm:safe}}} Let $\mathcal{G}$ be an FDG and let $\chi_1, \ldots, \chi_n$ be \emph{strongly safe} interleavings.  Then,  $S = \{ \chi_1, \ldots, \chi_n\}$ is a safe interleaving set for $\mathcal{G}$.

\begin{proof}
  By definition of safe set of interleavings, we have to prove the following for every interleaved history $\tis_{\fdg}$ of monitor $\mtr_{\fdg}$
\[
\textrm{If\ } \mathcal{X}(\tis_{\fdg}) \subseteq S \textrm{\ and\ } \mtr_{\fdg} \vdash (\tis_{\fdg}, \args_{\fdg}, \pstate) \Downarrow \pstate' \textrm{\ then\ } \exists\tis,\args.\ (\tis_{\fdg},\args_{\fdg}) \backsim (\tis,\args) \textrm{\ and\ } \mtr \vdash (\tis, \args, \pstate) \Downarrow \pstate' 
\]

In order to prove that, we have to prove that for every interleaved history $\tis_{\fdg}$ that only allows interleavings in $S$ and argument mapping $\args_{\fdg}$ we can find a history of the original monitor with corresponding argument mapping s.t., ($(\tis_{\fdg}, \args_{\fdg}) \backsim (\tis, \args)$. Which in turn means that we have to find a sequential history of $\mtr_{\fdg}$ $\tis'_{\fdg}$ s.t. 
\[
(1) \  \forall t.\ \pi(\tis_{\fdg}, t) = \pi(\tis'_{\fdg}, t) \quad \quad \emph{and} \quad \quad 
(2) \ \concr{\mtr_{\fdg}}(\tis, \args, \pstate) = (\tis'_{\fdg}, \args_{\fdg}) 
\]

To prove the goal above, we start with an arbitrary interleaved history $\tis_{\fdg}$ s.t. $\mathcal{X}(\tis_{\fdg}) \subseteq S$ and convert it to a sequential history $\tis'_{\fdg}$ with the above properties. We perform this proof, by first creating the CCR partition of $\tis_{\fdg}$, $P = CCRPart(\tis_{\fdg})$, and then induct on the number of partitions in $P$ that \emph{are} interleaved.

\paragraph{\textbf{Base Case: One interleaved CCR in P}} Let $\tis^{ccr}_{\fdg} = P[i]$ be the interleaved history in $\tis_{\fdg}$. Now, let $\tis'_{\fdg} = \tis_{\fdg}[Seq_{|CCR}(\tis^{ccr}_{\fdg})/\tis^{ccr}_{\fdg}]$. Because  $\tis^{ccr}_{\fdg}$ is the only interleaved sub-history in $P$ and because of  lemma~\ref{lem:decr-ileavings}, we have that $\tis'_{\fdg}$ is a sequential history s.t. $\forall t, \args_{\fdg}.\ \pi(\tis_{\fdg}, t) = \pi(\tis'_{\fdg}, t)$. Furthermore, because of lemma~\ref{lemma:equiv-state} we have $\forall \pstate, \args. \mtr_{\fdg} \vdash (\tis_{\fdg}, \args_{\fdg}, \pstate) \Downarrow \pstate' \Rightarrow \mtr_{\fdg} \vdash (\tis'_{\fdg}, Seq(\args_{\fdg}), \pstate) \Downarrow \pstate'$. Finally, because we have $\mtr_{\fdg} \vdash (\tis_{\fdg}, \args_{\fdg}, \pstate) \Downarrow \pstate'$ for some $\pstate$, this implies that $\concr{\mtr_{\fdg}}(\tis, \args_{\fdg}, \pstate) = (\tis'_{\fdg}, Seq(\args_{\fdg}))$, which in turns implies $(\tis'_{\fdg}, \args_{\fdg}) \backsim (\tis, Seq(\args_{\fdg}))$.

\paragraph{\textbf{Inductive Step}} Next, we assume that our theorem holds for up to $n$ interleaved CCRs in $P$, and will prove it for $n + 1$. Similarly as above, we find the smallest $i$ s.t. $P[i] = \tis^{ccr}_{\fdg}$ is an interleaved history. Again, we construct $\tis'_{\fdg} = \tis_{\fdg}[Seq_{|CCR}(\tis^{ccr}_{\fdg})/\tis^{ccr}_{\fdg}]$. Because of lemma~\ref{lem:decr-ileavings}, we have that the number of interleaved histories in $CCRPart(\tis'_{\fdg})$ has strictly fewer number of interleaved sub-histories than $P$. Therefore, by our inductive hypothesis, we have that $(\tis'_{\fdg}, \args_{\fdg}' \backsim (\tis, Seq(\args_{\fdg}))$ for some history of $\tis$ of $\mtr$. This, combined with lemma~\ref{lemma:equiv-state}, proves that $(\tis_{\fdg}, \args_{\fdg}) \backsim (\tis, Seq(\args{\fdg}))$.

\end{proof}

\subsection{Proof of Theorem \ref{thm:maxsat}}
\label{sec:maxsat-proof}

We now prove theorem~\ref{thm:maxsat} which states the correctness of our MaxSAT encoding.

\paragraph{\textbf{Theorem~\ref{thm:maxsat}}}
 Let $m$ be a model of the generated MaxSAT instance and $(\lockmap,\atomfld,\predmap)$ be the synchronization protocol constructed as follows:
  \begin{gather*}
     {\small\lockmap = \left\{ v \mapsto \left\{ l \mid m[h_v^l]\right\}\right\}\ \atomfld = \left\{ \code{fld} \mid m[\toatom_{fld}] \right\}}\ 
     {\small\predmap = \left\{p \mapsto  l_i \mid IsWait(v,p), i = min(\{j \mid m[\holds_v^{l_j}]\}) \right\}}
  \end{gather*}
  where, $IsWait(v,p)$ is true if v is a \wuntil statement on $p$. Then, $(\lockmap,\atomfld,\predmap)$ is a correct synchronization protocol.
  
\begin{proof}
  As mentioned earlier, a synchronization protocol must meet the following correctness criteria:
  \begin{enumerate}
      \item If two fragments $v_1,v_2$ have a race (i.e., $\mathcal{R}(v_1,v_2) \neq \emptyset$), then the protocol must prevent this race with a lock or an atomic field.
      \item If a fragment interleaving $\chi = (v,e)$ is not safe, then the synchronization protocol must not allow fragment $v$ to execute in between edge $e$.
      \item The protocol must be deadlock-free.
  \end{enumerate}
  
  We show that, by construction, a model $m$ returned by a MaxSAT solver always satisfies the above conditions.
  \begin{enumerate}
      \item Model $m$ prevents any races between two fragments because it must satisfy \emph{all} hard constraints generated by rules \textsc{Race-1} and \textsc{Race-2} from Figure~\ref{fig:maxsat-alg}. Therefore, $m$ will force two racy fragments to either share a lock or, when possible, convert all operations involving the racy field to equivalent atomic ones.
      \item Similarly, because model $m$ must satisfy the hard constraints generated by rule \textsc{I-Leave}, any interleaving that was deemed unsafe by our static analysis is guaranteed to be infeasible in the resulting synchronization monitor.
      \item Finally, because of rules \textsc{Wait} and \textsc{L-Order}, the resulting synchronization monitor is guaranteed to be deadlock-free. Specifically, the hard constraints generated by rule \textsc{L-Order} enforce the invariant that all lock acquisitions respect the global lock order. Whereas, the hard constraints of rule \textsc{Wait}, enforce the same invariant for the translation of a \wuntil statement into an equivalent statement in the target language (see Figure~\ref{fig:instrumentation}).
  \end{enumerate}
\end{proof}

\subsection{Proof of Theorem~\ref{thm:correct-instr}}\label{sec:instr-proof}

Finally, we prove the correctness of our monitor instrumentation procedure (Fig.~\ref{fig:instrumentation}).

\begin{theorem}
 Let  $\mathcal{S} = (\lockmap, \atomfld, \predmap)$ be a synchronization protocol inferred over FGD $\mathcal{G} = (V,E)$ of input monitor $M$ and $M'$ be the result of procedure $\textsf{Instrument}$ for $M$. Then, the following three conditions hold:
 \begin{enumerate}
     \item For every fragment $v \in V$, $l_i \in \lockmap[v]$ iff fragment $v$ holds lock $l_i$ in $M'$
     \item If $i < j$, then $l_i$ is never acquired whenever $l_j$ is held.
     \item Field $f \in \atomfld$ iff all its occurrences in $M$ have been replaced with an atomic operation in $M'$.
 \end{enumerate}
\end{theorem}

\begin{proof}
All three conditions can be proved by providing certain guarantees for a subset of the rules of Figure~\ref{fig:instrumentation}. Note that operator $\rightsquigarrow$ (Figure~\ref{fig:instrumentation}) is guaranteed to visit every code fragment $v \in V$ in the FDG, since it recursively visits every element of the input monitor until it discovers \emph{all} fragments of the given FDG. Next, we prove all three conditions.

\paragraph{Condition (1):}
For this condition, we need to prove that both fragments will only hold the locks required by the synthesized protocol $\mathcal{S}$. The logic of this proof depends on the number and type of predecessors of fragment $v$. We now present a case analysis:

\paragraph{Zero predecessors.}
This is the case of an entry fragment of a method in $\fdg$. Due to the structure of our input language and the definition of an FDG, this fragment \emph{must} be a fragment that contains a single \wuntil statement. The instrumentation of such a fragments is handled by rules $\textsc{Wait}$ and $\textsc{Entry-Frag}$. Note, that rule $\textsc{Wait}$ first calls $\textsc{Entry-Frag}$ which acquires all locks needed by the fragmented defined by the $\wuntil$ fragment.

\paragraph{At least one predecessor.}
These types of fragments are handled by rules $\textsc{Wait}$, $\textsc{Branch-Frag}$, $\textsc{Reg-Frag-1}$, and $\textsc{Reg-Frag-2}$. All these rules maintain the following invariant for the fragment $v$ that triggers them: before transferring control to any of $v$'s successor, they release all locks needed by $v$ but not needed by the successor (locksets of the form $R_i$) and acquire all locks needed by the successor but not held by $v$ (locksets of the form $A_i$). This invariant combined with the fact that these are the only ways to transfer control flow in our input language, ensure that before executing a fragment in the output monitor all necessary locks (and only those) will be acquired.

\paragraph{Condition (2):}
This directly follows from:
\begin{enumerate}
    \item That procedure $\mathsf{Instrument}$ uses auxiliary relation $\mathsf{Acq}$ to instrument lock acquisitions, which as shown in Figure~\ref{fig:instrumentation} does so in increasing order of lock indices.
    \item The guarantee provided by Theorem~\ref{thm:maxsat} that the synthesized protocol acquires locks in increasing order along every control-flow edge. 
\end{enumerate}

\paragraph{Condition (3):}
This condition is ensured by rule $\textsc{Frag-Stmt}$ of Figure~\ref{fig:instrumentation} that ensures oracle $\rightarrow_{\mathcal{A}}$ is called on every fragment of $\fdg$.
\end{proof}